%% file: main.tex
\documentclass[letterpaper,twocolumn,10pt]{article}
\usepackage{usenix}

\PassOptionsToPackage{pdfpagelabels=false}{hyperref}
\newcommand{\allnotes}[1]{}

\usepackage[pdftex,bookmarks=false,colorlinks=true,citecolor=blue,filecolor=black,linkcolor=blue,urlcolor=black]{hyperref}
\usepackage{amsmath}
\usepackage{amsthm}
\usepackage{endnotes}
\usepackage{enumitem}
\usepackage[T1]{fontenc}
\usepackage{tabularx}
\usepackage{booktabs}
\usepackage[small,compact]{titlesec}
\usepackage{alltt}
\usepackage[skip=5pt]{caption}
\usepackage[aboveskip=2pt]{subcaption}
\usepackage{balance}
\usepackage{fancyvrb}
\usepackage{epsfig}
\usepackage{textcomp}
\usepackage{listings}
\usepackage[framemethod=tikz]{mdframed}
\usepackage[scaled]{beramono}
\usepackage[T1]{fontenc}
\usepackage{multirow}
\usepackage{mathtools}

\newtheorem{definition}{Definition}

\mdfdefinestyle{protocolBox}{,,roundcorner=5pt,tikzsetting={draw=gray, line width=1pt}}%

\newcommand{\captionfonts}{\bf \footnotesize}

\makeatletter  %
\long\def\@makecaption#1#2{%
	\vskip\abovecaptionskip
	\sbox\@tempboxa{{\captionfonts #1: #2}}%
	\ifdim \wd\@tempboxa >\hsize
	{\captionfonts #1: #2\par}
	\else
	\hbox to\hsize{\hfil\box\@tempboxa\hfil}%
	\fi
	\vskip\belowcaptionskip}
\makeatother   %

\newcommand{\squishlist}{
  \begin{list}{$\bullet$}{
    \setlength{\itemsep}{0pt}       \setlength{\parsep}{3pt}
    \setlength{\topsep}{3pt}        \setlength{\partopsep}{0pt}
    \setlength{\leftmargin}{1em}    \setlength{\labelwidth}{1em}
    \setlength{\labelsep}{0.5em} } }

\newcommand{\squishend}{
  \end{list} }

\newcommand{\squishenum}{
  \begin{enumerate}{}{
    \setlength{\itemsep}{0pt}       \setlength{\parsep}{3pt}
    \setlength{\topsep}{3pt}        \setlength{\partopsep}{0pt}
    \setlength{\leftmargin}{1em}    \setlength{\labelwidth}{1em}
    \setlength{\labelsep}{0.5em} } }

\newcommand{\squishenumend}{
\end{enumerate} }

\newcommand{\shortarrow}[1][4pt]{\!\mathrel{%
   \vcenter{\hbox{\rule[-.5\fontdimen8\textfont3]{#1}{\fontdimen8\textfont3}}}%
   \mkern-4mu\hbox{\usefont{U}{lasy}{m}{n}\symbol{41}}}\!}

\setlist[enumerate]{leftmargin=1.5em,topsep=3pt,itemsep=0pt}
\microtypecontext{spacing=nonfrench}

\usepackage{ulem}
\usepackage{xspace}
\usepackage{xcolor}
\usepackage{outlines}
\usepackage{tikz}
\usetikzlibrary{snakes,arrows,shapes,calc,positioning}

\usepackage{pdflscape}
\usepackage{array}
\usepackage{thmtools,thm-restate}
\usepackage{rotating}
\usepackage{amsmath}
\usepackage{amssymb}
\usepackage{comment}
\usepackage{graphicx}
\usepackage{soul}
\usepackage{xfrac}

\excludecomment{Exclude}

\definecolor{light-gray}{gray}{0.94}
\sethlcolor{light-gray}

\usepackage[textsize=tiny,textwidth=0.6in,disable]{todonotes}
\renewcommand{\allnotes}[1]{\textit{#1}}

\newcolumntype{x}[1]{>{\centering\arraybackslash\hspace{0pt}}p{#1}}

\begin{document}

\title{Efficient Direct-Connect Topologies for Collective Communications\thanks{Distribution Statement ``A'' (Approved for Public Release, Distribution Unlimited).}}

\author{
{\rm Liangyu Zhao}\\
University of Washington
\and
{\rm Siddharth Pal}\\
Raytheon BBN
\and
{\rm Tapan Chugh}\\
University of Washington
\and
{\rm Weiyang Wang}\\
MIT CSAIL
\and
{\rm Jason Fantl}\\
Raytheon BBN
\and
{\rm Prithwish Basu}\\
Raytheon BBN
\and
{\rm Joud Khoury}\\
Raytheon BBN
\and
{\rm Arvind Krishnamurthy}\\
University of Washington
}

\maketitle
\input{abstract}
\pagestyle{plain}

\input{intro}
\input{back}

\input{topo}
\input{expansions}

\input{finder}

\input{BFS_schedule}

\input{implementation2.tex}

\input{eval}

\input{analysis}

\input{conclusion}

\section{Acknowledgements}
This research was developed with funding from the Defense Advanced Research Projects Agency (DARPA) under Contract No.HR001120C0089. The views, opinions and/or findings expressed are those of the author and should not be interpreted as representing the official views or policies of the Department of Defense or the U.S. Government.

\newpage
\bibliographystyle{acm}
\bibliography{refs}

\newpage
\appendix
\input{appendix}

\end{document}

%% file: abstract.tex
\begin{abstract}
We consider the problem of distilling efficient network topologies for collective communications. We provide an algorithmic framework for constructing direct-connect topologies optimized for the latency vs. bandwidth trade-off associated with the workload.
Our approach synthesizes many different topologies and schedules for a given cluster size and degree and then identifies the appropriate topology and schedule for a given workload.
Our algorithms start from small, optimal base topologies and associated communication schedules and use techniques that can be iteratively applied to derive much larger topologies and schedules.
Additionally, we incorporate well-studied large-scale graph topologies into our algorithmic framework by producing efficient collective schedules for them using a novel polynomial-time algorithm.
Our evaluation uses multiple testbeds and large-scale simulations to demonstrate significant performance benefits from our derived topologies and schedules. 
\end{abstract}

%% file: intro.tex
\section{Introduction}

Collective communication operations involve concurrently aggregating and distributing data on a cluster of nodes and are used in both machine learning (ML) and high-performance computing (HPC).
With the improved computational capabilities of accelerators, collective operations are a significant overhead in large-scale distributed ML training~\cite{sergeev2018horovod,gibiansky2017bringing,wang2020blink}.

An emerging approach to meet these challenging demands has been to employ various forms of optical circuit switching to achieve higher bandwidths at reasonable capital expenditure and energy costs~\cite{sipml,topo-opt,sip-switch,x-nest,hybrid-opt,tpuv4,tpuv4fabric}.
Hosts communicate using a limited number of optical circuits that can be reconfigured at timescales appropriate for the hardware (see \S\ref{sec:fabric}), thus exposing network topology as a configurable component.
We refer to this setting as \emph{direct-connect} with circuits configured and fixed for an appropriate duration.

Existing optical-circuits-based ML systems~\cite{sipml, sip-switch,x-nest,hybrid-opt} fit this direct-connect model but do not exploit the flexibility topology reconfiguration offers. 
Collective operations such as allreduce are still limited to a few well-known algorithms that can fit the degree constraints of the optical fabric (e.g., rings, multi-rings, tori, and trees with bounded degrees) and accept the consequent performance tradeoffs.
For example, ring allreduce, while bandwidth-efficient, has a high graph diameter, causing high total-hop latency.
A double binary tree, on the other hand, has a logarithmic diameter but suffers from load imbalances and bandwidth inefficiencies. Conversely, the broader spectrum of well-known collective algorithms that achieve desired latency and bandwidth (e.g., recursive-doubling, Bruck algorithm)~\cite{thakur2005optimization,coll-survey} use dynamic communication patterns ideal for switch networks but are ill-suited for degree-constrained direct-connect networks.

To fill this gap, we seek to identify new custom-built topologies and communication schedules for direct-connect networks. We pose the following question: \emph{How to efficiently construct high-performance direct-connect topologies and communication schedules for collectives given the network's performance characteristics and degree constraints?}

This question poses several challenges. First, jointly optimizing \emph{both} the network topology and the corresponding communication schedule is intractable at a large scale. Prior efforts reduce the search cost by optimizing only one or the other (e.g., schedules for a given topology~\cite{wang2020blink,SCCL,shah2021synthesizing} or topology permutations while retaining a ring schedule~\cite{topo-opt}).
The combination of topological structure and communication schedule as degrees of freedom explodes the search space, making this a seemingly intractable problem.
Second, the optimization must carefully consider the workload and the network's performance characteristics when distilling a topology and schedule. For example, minimizing the topology's diameter is ideal not only for latency-sensitive allreduce at small data sizes but also for all-to-all throughput; however, this could come at the cost of load imbalance across links in bandwidth-sensitive allreduce at large data sizes.
Finally, lowering the synthesized schedules to the underlying hardware and runtimes~\cite{ncclRAR, msccl, oneccl} in an efficient way requires careful scheduling to achieve the desired performance in practice.

Our work addresses these issues by developing an algorithmic toolchain for quickly synthesizing efficient topologies and schedules for collective communications.
\squishenum
\item We devise a range of \textbf{expansion techniques} for synthesizing custom large-scale network topologies and schedules. The expansions start with small, optimal topologies and communication schedules and expand them to achieve near-optimal large-scale topologies and schedules.
\item We devise a \textbf{polynomial-time schedule generation algorithm} to produce optimal collective communication schedules for large-scale topologies with certain symmetry properties. This exposes many well-known topologies as options for the direct-connect network fabric.
\item We devise a \textbf{topology enumeration and search algorithm} to identify the best option for a target cluster and workload by exploring the \emph{Pareto-efficient} options that provide different tradeoffs for bandwidth efficiency, total-hop latency, and also all-to-all throughput.
\item We develop \textbf{compilers} to realize our optimized schedules. We offer efficient implementations for both GPUs and CPUs and integrate with ML frameworks (e.g., PyTorch) through the MSCCL~\cite{msccl} and oneCCL~\cite{oneccl} runtimes.
\squishenumend

We evaluate our approach using two testbeds: a 12-node GPU cluster capable of topology reconfiguration, and torus clusters on Frontera~\cite{frontera} supercomputer with up to 54 CPU nodes.
Our techniques reduce collective communication times by $>\!30\%$ for DNN training on the GPU testbed and up to $3.1\!\times$ for HPC workloads on the supercomputer. Simulations for large-scale DNN training show up to an order of magnitude reduction in total communication time from topology and schedule optimization. Our schedule generation algorithm is orders of magnitude faster than the state-of-the-art (e.g., SCCL~\cite{SCCL} and TACCL~\cite{shah2021synthesizing}), capable of producing schedules for topologies with thousands of nodes in a minute.

%% file: back.tex
\section{Background \& Related Work}

\input{network_fabric}

\input{related_work}

\section{Formal Model of Collective Communications}\label{sec:collectiveintro}

\begin{table}[tb]
    \centering
    \resizebox{\columnwidth}{!}{
    \begin{tabular}{|c|l||c|l|}
        \hline
        $M$ & total data size &
        $\alpha$ & single-hop latency \\
        $N$ & number of nodes &
        $B$ & total egress bandwidth of a node \\
        $S$ & data shard ($|S|=\frac{M}{N}$) &
        $B/d$ & bandwidth of a single link \\
        $C$ & data chunk ($C\subseteq S$) &
        $T_L(A)$ & total-hop latency of schedule \\
        $d$ & degree of topology &
        $T_B(A)$ & bandwidth runtime of schedule \\
        $V_G$ & vertex/node set of $G$ &
        $T^*_L\!(\!N\!,\!d)$ & Moore optimality (Def \ref{def:mooreopt}) \\
        $E_G$ & edge/link set of $G$ &
        $T^*_B(N)$ & bandwidth optimality $\frac{M}{B}\cdot\frac{N-1}{N}$ \\
        $D(\!G\!)$ & graph diameter of $G$ &
        $N^+_x(u)$ & nodes at distance $x$ from $u$ \\
        \cline{1-2}
        \multicolumn{2}{|c||}{Table~\ref{table:toposummary} for graph symbols} &
        $N^-_x(u)$ & nodes at distance $x$ to $u$ \\
        \hline
    \end{tabular}
    }
    \caption{Summary of Important Notations}
    \label{tab:notations}
\end{table}

We provide a formal model of \textbf{\emph{reduce-scatter}}, \textbf{\emph{allgather}}, and \textbf{\emph{allreduce}} collectives.
In each operation, there are $N$ nodes operating on a vector of data of total size $M$.
The data can be divided into $N$ shards.
In \emph{reduce-scatter}, each node $i$ reduces the $i$-th shard from all other nodes;
in \emph{allgather}, each node $i$ broadcasts the $i$-th shard to all other nodes;
in \emph{allreduce}, each node $i$ ends up with the fully reduced vector of data.

Throughout the paper, we only elaborate on \emph{allgather} schedule construction because the other two collectives are direct transformations. Since \emph{allgather} and \emph{reduce-scatter} are, respectively, simultaneous broadcasts and reductions for each node, we can \textbf{construct \emph{reduce-scatter} schedules} in bidirectional topologies by simply reversing the communications in \emph{allgather} schedules~\cite{conc_comp}.
In unidirectional topologies, we utilize graph transposition to achieve a similar transformation~(Appendix~\ref{sec:reducescatterallgather}).\footnote{The main text of this paper focuses on high-level ideas of various techniques. We provide detailed mathematical analyses in the appendix.}
\textbf{To construct an \emph{allreduce} schedule,} we concatenate \emph{reduce-scatter} and \emph{allgather}.

\subsection{Communication Topology \& Schedule}

\begin{figure}[tb]
    \centering
    \scalebox{0.8}{
        \input{figures/example_actual.tex}
    }
    \caption{The allgather schedule of complete bipartite graph $K_{2,2}$. \normalfont Shard $S$ is divided into two half chunks $C_1$ and $C_2$. From $a$, at the 1st comm step, $a$ sends the entire shard $S$ to both $c$ and $d$. At the 2nd comm step, $c$ and $d$ send the two half chunks $C_1$ and $C_2$ respectively to $b$. Thus, every node receives the full shard from $a$. By applying similar broadcast from $c,b,d$ in parallel, we have a complete BW-optimal \emph{allgather} schedule with $T_L\!=\!2\alpha,T_B\!=\!\frac{M}{B}\!\cdot\!\frac{3}{4}$.}
    \label{fig:notationexample}
\end{figure}

The network topology is modeled as a directed graph (digraph) $G=(V,E)$, where $V$ denotes the set of nodes ($|V|=N$) and $E$ denotes the set of directed links/edges.
The direct-connect network imposes a constraint that all nodes have degree $d$, which is the number of connection ports on each host and is typically low and independent of $N$.

A communication \textbf{algorithm} $(G,A)$ uses the communication \textbf{schedule} $A$ on topology $G$. Schedule $A$ can be specified as what chunk $C$ is communicated over which link in which \textbf{communication (comm) step $t$}. We define \textbf{chunk} $C$ as a subset of shard $S$. Both $C$ and $S$ are specified as \emph{index sets} of elements. Typically, $S$ is interval $[0,1]$ representing the whole shard, and $C$ is some subinterval. We denote $v$'s chunk $C$ as $(v,C)$, which is a subset of $v$'s starting shard $(v,S)$.
Let $((v,C),(u,w),t)$ denote that $v$'s chunk $C$ is sent by node $u$ to its neighbor $w$ at comm step $t$. Schedule $A$ then is specified as a list of tuples $((v,C),(u,w),t)$.
Figure~\ref{fig:notationexample} gives an example of an allgather schedule in such a tuple notation.
Within a schedule, chunks can be different-sized subsets of $S$.
Appendix~\ref{sec:reducescatterallgather} gives formal definitions of reduce-scatter/allgather schedules.

\subsection{Cost Model}\label{sec:costmodelmain}

We use the well-known $\alpha$-$\beta$ cost model~\cite{hockney1994communication}.
The cost of sending a message of size $H$ over a link is $\alpha + \beta H$. This cost comprises two components: the constant single-hop latency $\alpha$ and a bandwidth component $\beta$, which is the inverse of link bandwidth, i.e., $\beta=\frac{1}{b}$.
This simple model has been shown to be appropriate for GPU interconnects~\cite{li2019evaluating, SCCL, shah2021synthesizing}.
In our analysis, we use node bandwidth $B$ with $B=db$.
In this paper, we focus on homogeneous networks, although some techniques also support heterogeneous ones (Appendix~\ref{sec:heterogeneous}).

The runtime of a schedule $A$ can be broken down into a total-hop latency component and a bandwidth component. The \textbf{total-hop latency} component $T_L(A)\!=\!t_{\max}\alpha$, where $t_{\max}$ is the number of comm steps. $T_L(A)$ represents the cost of performing schedule $A$ on an infinitesimal amount of data. The bandwidth component $T_B(A)$, or \textbf{bandwidth (BW) runtime}, is the sum of the BW runtime of each comm step, i.e., $\sum_tT_B(A_t)$. The BW runtime of comm step $t$ is the max amount of data transmitted by a link within the comm step, divided by link bandwidth $b\!=\!B/d$.
For $N$-node $d$-regular graphs, the optimal runtime of both \emph{reduce-scatter} and \emph{allgather} is approximately $\alpha\log_d N\!+\!\frac{1}{B}\!\cdot\!\frac{M(N-1)}{N}$. The 1st term represents the total-hop latency required for communicating across the diameter of a topology, while the 2nd term represents the transmission time for any node to send/recv $N\!-\!1$ shards in reduce-scatter/allgather. \textbf{One should not confuse total-hop latency with overall latency,} which is the sum of total-hop latency and BW runtime. We omit the computational time of reduction and discuss this in Appendix~\ref{sec:computeopt}.

We analyze the optimality of total-hop latency and BW runtime separately. An algorithm $(G,A)$ is optimal in one component if no algorithm with the same $N,d$ can perform better. For BW runtime, an algorithm is \textbf{bandwidth (BW) optimal} iff its $T_B$ equals $T_B^*(N)\!\coloneqq\!\frac{M}{B}\!\cdot\!\frac{N-1}{N}$.
For total-hop latency, given $G$, the lowest $T_L$ achievable is $\alpha\!\cdot\! D(G)$, where $D(G)$ is the graph diameter of $G$. Thus, the optimal total-hop latency equals the smallest diameter of any $N$-node $d$-regular graph, which remains an open question in graph theory~\cite{moore_bound}. Therefore, we define \emph{Moore optimality} based on \emph{Moore bound}, which provides a lower bound for diameter given $N,d$ and thus a well-defined $T_L^*(N,d)$. An algorithm is \textbf{Moore optimal} iff $T_L\!=\!T_L^*(N,d)$. 
Moore optimal topologies have the lowest diameter, which is also ideal for all-to-all throughput.
Appendix~\ref{app-sec:optimality} gives formal definitions of optimalities.

\emph{Ring allreduce}~\cite{ncclRAR} has a total-hop latency that is linear in $N$, while the BW runtime is optimal. \emph{Double binary trees} (DBT)~\cite{ncclDBT}, on the other hand, offers the advantage of logarithmic total-hop latency but has suboptimal BW performance. Our work offers a range of topologies that are Pareto-efficient in total-hop latency and BW performance.

%% file: network_fabric.tex
\subsection{Network Fabric}
\label{sec:fabric}

Our work identifies topologies and schedules helpful for a broad range of settings, such as \emph{switchless physical circuits}, \emph{patch-panel optical circuits}, and \emph{optical circuit switches}. While these options differ in cost and reconfigurability~\cite{topo-opt}, they are all significantly cheaper than packet-switch solution~\cite{tpuv4,tpuv4fabric,topo-opt} and can benefit from our work.

\emph{Switchless physical circuits} require the least amount of fabric hardware.  However, the topology must remain reasonably static for long periods, as the reconfiguration is manual. \emph{Patch-panel optical circuits} provide a higher degree of reconfigurability by using a mechanical solution (e.g., robotic arms) to perform physical reconfigurations through a patch panel. The reconfigurations occur on the scale of minutes, but the patch panel itself can scale to a large number of duplex ports and is reasonably cheap (e.g., 1008 ports at \$100 per port~\cite{patch-panel}).
Both options can benefit from a carefully curated topology optimized for the workload, but they require it to remain static for a job given the reconfiguration costs.

Commercial \emph{optical circuit switches (OCS)} can perform reconfigurations in $\approx$10ms, are more expensive than patch panels, and scale to fewer ports~\cite{polatis,calient} (e.g., Polatis 3D-MEMS switch has 384 ports at \$520 per port~\cite{polatis}). Though OCSes support faster reconfigurations, the delays are still too high to support the rewiring of the circuits \emph{during} a typical collective operation.\footnote{Research prototypes~\cite{rotornet,sirius} support $\mu s$ to $ns$ reconfigurations using overlay-hop relays and Valiant load balancing (VLB). While this is valuable for generic workloads, collectives have structured communication patterns, and it would be ideal to realize them without incurring the VLB overheads.} Thus, they cannot take advantage of algorithms designed for full-bisection switches, such as recursive halving/doubling \cite{thakur2005optimization,coll-survey}, that exploit high logical degree over time to provide both latency and bandwidth optimality. Thus, OCSes can also benefit from the custom-built and low-degree topologies synthesized by our approach.  

All of these optical technologies allow for a shared cluster to be split into multiple subclusters for running separate jobs~\cite{tpuv4}, so each job can be configured with its own topology. Further, unidirectional topologies are technically feasible on optical testbeds. 
Unidirectionality gives greater freedom in topology design and can enable lower-diameter networks.

\noindent \textbf{Evaluation Target:} In this paper, we use a reconfigurable optical patch panel to configure and evaluate different topologies. Given the high reconfiguration costs for the patch panel, we identify an efficient topology that will remain static for the duration of a job. Nevertheless, our techniques could be used to derive topologies for finer reconfiguration timescales if the hardware can efficiently support them.

%% file: related_work.tex
\subsection{Related Work}\label{sec:relatedwork}
Several existing optical-circuit-based ML systems~\cite{sipml, sip-switch,x-nest,hybrid-opt,topo-opt} fit the direct-connect model; however, they rely on existing implementations of collectives. Typically, communication libraries for ML training~\cite{sergeev2018horovod,nccl,gibiansky2017bringing,noordhuis2017accelerating} offer either ring collective for high-latency bandwidth-optimal transfers or tree collective, which has logarithmic latency but suffers from load-imbalances across links. Other topologies such as mesh, tori, hypercubes, etc., have also been explored in HPC systems~\cite{bokhari1992complete, ho1986distributed, barnett1994interprocessor, patarasuk2009bandwidth,conc_comp, chan2006collective, faizian2017random}, but their bandwidth-latency tradeoff choices are limited as well. Bandwidth and latency optimal collectives for switch networks such as recursive-doubling, Bruck algorithm~\cite{coll-survey,thakur2005optimization}, BlueConnect~\cite{cho2019blueconnect}, etc., are unsuitable for direct-connect networks, because their one-to-one communication patterns fail to utilize all available links, and they assume a fully connected network.

Our work uniquely considers joint optimization of \emph{both} the network topology and the corresponding collective communication schedule at a large scale, while prior work either optimizes one or the other. For instance, TopoOpt~\cite{topo-opt} generates customized shifted-ring topologies to optimize concurrent collective and non-collective communications for hybrid data-parallel~\cite{li2020pytorch,li2014scaling,dean2012large} and model-parallel~\cite{shoeybi2019megatron,narayanan2021efficient,jia2019beyond} DNN training, respectively. The collective communications in TopoOpt still use existing ring collectives. Consequently, when data or expert parallelism, for example, dominate the workload, TopoOpt's performance suffers from similar latency issues present in existing ring collectives (see \S\ref{sec:arperf}).
Our effort is complementary as it synthesizes new topologies and schedules for collectives that span the entire cluster, but we do not optimize sub-cluster communications for hybrid parallelism. Extending our work to jointly optimize topologies and schedules for hybrid parallelism is future work.

Recent work like Blink~\cite{wang2020blink}, SCCL~\cite{SCCL}, and TACCL~\cite{shah2021synthesizing} also focus on generating a collective schedule for a given topology. However, they all involve NP-hard optimizations that severely limit their scalability. SCCL is capable of generating optimal schedules, but it fails to generate a schedule in a reasonable time when the topology size is beyond 30 nodes. TACCL improves the scalability of SCCL by using communication sketches and also handles switch networks, but it sacrifices schedule performance and is still limited in scalability (see \S\ref{sec:schedeval}).
In our approach, we either synthesize the schedule along with the topology or rely on a polynomial-time schedule generation technique that is provably optimal for networks with certain symmetry properties.

Generic large-scale topologies are typically not optimized for collective communications but for general datacenter traffic~\cite{jellyfish,xpander,dragonfly,slimfly,polarfly,spectralfly}.
Our framework can incorporate any degree-constrained regular topology (e.g., low-diameter expander graphs~\cite{rolim1998bisecting,genkautz}) and generate candidate schedules.

\subsection{All-to-All Throughput}\label{sec:alltoallmain}

While we optimize collectives like allreduce, reduce-scatter, and allgather, the performance of all-to-all communication is also crucial for training DNN models like Mixture of Experts (MoE)~\cite{gshard,switchtran,deepspeedmoe,lina} and Deep Learning Recommendation Model (DLRM)~\cite{naumov2019deep,naumov2020deep}. Unlike other collectives, the scheduling of all-to-all can be easily formulated and efficiently solved as multi-commodity flow (MCF) problem~\cite{basu2024efficient,alltoall1,alltoall2,alltoall3,xpander}. However, the graph diameter of the underlying topology is critical for all-to-all throughput~\cite{basu2024efficient,diameter1,diameter2,topo-opt,rotornet}.
The intuition behind this is that if the nodes are far from each other, then all-to-all flows cost more \textit{bandwidth tax}~\cite{rotornet2,topo-opt} (i.e., the bandwidth of the flow multiplied by the length of the flow). With a fixed \textit{network capacity} (i.e., the total number of links times link bandwidth), longer flows reduce the available bandwidth for each flow, thus decreasing all-to-all throughput. In this work, we construct topologies and associated schedules that are high-performance in both allreduce-type collectives and all-to-all by deriving efficient allreduce-type schedules on either existing or our synthesized low-diameter topologies.

%% file: figures/example_actual.tex
\begin{tikzpicture}[roundnode/.style={circle,draw=black,minimum size=7mm}]
	\node[roundnode,label=170:{$S=\begin{bmatrix} C_1 \\ C_2 \end{bmatrix}$}]	(0)	at (-1.5*0.7, 1.5*0.7) {$a$};
	\node[roundnode]	(2)	at (1.5*0.7, 1.5*0.7)	{$c$};
	\node[roundnode]	(1) at (1.5*0.7, -1.5*0.7)	{$b$};
	\node[roundnode]	(3) at (-1.5*0.7, -1.5*0.7)	{$d$};
	
	\path[-latex,anchor=south,line width=0.4mm] (0) edge [bend left=25] node {$((a,S),(a,c),1)$} (2);
	\path[-latex] (2) edge [bend left=25] (0);
	\path[-latex,anchor=west,line width=0.4mm] (2) edge [bend left=25] node {$((a,C_1),(c,b),2)$} (1);
	\path[-latex] (1) edge [bend left=25] (2);
	\path[-latex,anchor=north,line width=0.4mm] (3) edge [bend right=25] node {$((a,C_2),(d,b),2)$} (1);
	\path[-latex] (1) edge [bend right=25] (3);
	\path[-latex,anchor=east,line width=0.4mm] (0) edge [bend right=25] node {$((a,S),(a,d),1)$} (3);
	\path[-latex] (3) edge [bend right=25] (0);
\end{tikzpicture}

%% file: topo.tex
\section{Overview of Our Approach}
\label{sec:overview}

Direct-connect topologies can typically be categorized as either \textbf{low-hop} topologies, which have low diameters (e.g., expander graphs) suited for all-to-all throughput and small-data allgather/reduce-scatter/allreduce, or \textbf{load-balanced} topologies, which have simplistic structure (e.g., ring, torus) with easy load-balanced schedule for large-data allgather/reduce-scatter/allreduce (see Table~\ref{tab:tradeoffs}). We seek to jointly identify network topologies and schedules that achieve high performance in both categories to the extent possible. Specifically, this entails the challenging task of constructing load-balanced allgather\footnote{Recall that we construct allreduce and reduce-scatter by transforming allgather schedules.} schedules for low-hop topologies.

\begin{table}[tb]
    \centering
    \resizebox{\columnwidth}{!}{
    \begin{tabular}{|l|c|c|c|}
        \hline
        \multirow{2}{*}{\shortstack[c]{Topology Type}} & \multirow{2}{*}{\shortstack[c]{Small-Data Allreduce \\ (Total-Hop Latency $T_L$)}} & \multirow{2}{*}{\shortstack[c]{Large-Data Allreduce \\ (BW Perf $T_B$)}} & \multirow{2}{*}{\shortstack[c]{All-to-All \\ Throughput}} \\
        & & & \\
        \hline
        \textbf{Low-Hop} & \checkmark & -- & \checkmark \\
        \hline
        \textbf{Load-Balanced} & -- & \checkmark & -- \\
        \hline
    \end{tabular}
    }
    \caption{The tradeoffs of low-hop topology vs. load-balanced topology. \normalfont Reduce-scatter and allgather perform similarly to allreduce.}
    \label{tab:tradeoffs}
\end{table}

At a small scale, one could handpick a topology such as the complete bipartite graph $K_{2,2}$ defined at $N\!=\!4,d\!=\!2$. $K_{2,2}$ is both low-hop and load-balanced that a Moore- and BW-optimal allgather could be manually constructed (Figure~\ref{fig:notationexample}).
But how do we scale the topology and the schedule to larger sizes? Our work approaches this problem with two tools: \textit{expansion techniques}~(\S\ref{sec:expansion}) and \textit{BFB schedule generation}~(\S\ref{sec:shortestpathschedule}).

\textbf{Expansion Techniques:} Given a base topology and its schedule, expansion techniques can expand it into a larger topology and associated schedule with minimal loss in performance. We call the resulting topologies \textbf{synthesized topologies}. The base topologies are small in scale, such as $K_{2,2}$ in Figure~\ref{fig:notationexample}, for which straightforward schedules exist or an exhaustive search for the schedule is feasible. The line graph expansion, for example, can then expand $K_{2,2}$ and its schedule in Figure~\ref{fig:notationexample} to an allgather for $N\!=\!4\!\cdot\! 2^n$, for arbitrarily large $n$, while retaining a node degree of 2. Multiple expansion techniques can be composed to achieve the desired $N$ and $d$.

\textbf{Breadth-First-Broadcast (BFB) Schedule Generation:} Besides synthesized topologies, we can use known topologies from graph theory (e.g., twisted torus, expander graphs). We call them \textbf{generative topologies} as they can be instantiated at various $N$ and $d$. Generative topologies are often low-hop, beneficial for total-hop latency and all-to-all throughput. The problem, though, is that efficient load-balanced collective schedules are not known for many of these topologies, and existing schedule generators~\cite{SCCL, shah2021synthesizing} are intractable at moderate to large scale.
Our work offers BFB, a polynomial-time schedule generator that can yield high-performance schedules for large-scale topologies. For allgather, it performs a breadth-first broadcast from each node and uses linear programs to balance the workload on links. Although not always optimal, BFB schedules are provably optimal for many topologies exhibiting certain symmetries. For instance, BFB can generate a schedule with the lowest total-hop latency and BW optimality on any torus, including those with unequal dimensions.

With expansion techniques and BFB schedules, our \textit{topology finder} (\S\ref{sec:topofinder}) assembles a large pool of topologies and schedules, identify \textit{Pareto-efficient} ones from a low-hop vs. load-balanced perspective, and select from them for a given workload.
When two options are Pareto-efficient, one must be better than the other in either low-hop (i.e., total-hop latency) or load-balanced (i.e., BW performance) but not in both.
We choose low-hop options for workloads requiring all-to-all throughput and small-data allgather/reduce-scatter/allreduce, and load-balanced options for large-data allgather/reduce-scatter/allreduce.
Finally, the \textit{compiler} (\S\ref{sec:implementation}) lowers the chosen topology and schedule to the runtime and hardware.

%% file: expansions.tex
\begin{figure*}[htb]
    \centering
    \begin{minipage}{0.65\textwidth}
        \centering
        \begin{subfigure}{.32\textwidth}
            \centering
            \input{figures/complete_bipartite_2.tex}
            \caption{$K_{2,2}$ $(N\!=\!4,d\!=\!2)$}
        \end{subfigure}
        \hfill
        \begin{subfigure}{.33\textwidth}
            \centering
            \input{figures/L_complete_bipartite.tex}
            \caption{$L(K_{2,2})$ $(N\!=\!8,d\!=\!2)$}
            \label{fig:linegraphtopology}
        \end{subfigure}
        \hfill
        \begin{subfigure}{.33\textwidth}
            \centering
            \input{figures/D_complete_bipartite.tex}
            \caption{Broadcast from $ca$}
            \label{fig:linegraphschedule}
        \end{subfigure}
        \caption{The complete bipartite topology $K_{2,2}$ with its line graph $L(K_{2,2})$. \normalfont Figure (a) shows the base topology and broadcast paths from $a$ to $c,b,d$ in $A_{K_{2,2}}$ (see figure \ref{fig:notationexample}). The number next to edge shows the comm step using the edge. Figure (b) shows the expanded topology. Observe that every edge in $K_{2,2}$ becomes a vertex in $L(K_{2,2})$, and two vertices are connected if the corresponding edges in $K_{2,2}$ have one's head node being the other's tail node. Figure (c) shows the broadcast paths of node $ca$, transformed from the broadcast paths of $a$ in figure (a). At the 1st comm step, by step 1 of def \ref{def:linegraph}, $ca$ broadcasts its shard to all its neighbors: $((ca,S),(ca,ac),1)$, $((ca,S),(ca,ad),1)$. The rest of the broadcast paths are transformed from $A_{K_{2,2}}$ by step 2 of def \ref{def:linegraph}, e.g. $((a,C_1),(c,b),2)\mapsto\{((ca,C_1),(cb,bc),3),((ca,C_1),(cb,bd),3)\}$. Each of the nodes $bc$ and $bd$ receives $C_1,C_2$ from its two in-neighbors, just like $b$ does in $A_{K_{2,2}}$.}
        \label{fig:linegraph}
    \end{minipage}
    \hfill
    \begin{minipage}{0.33\textwidth}
        \centering
        \includegraphics[width=\textwidth]{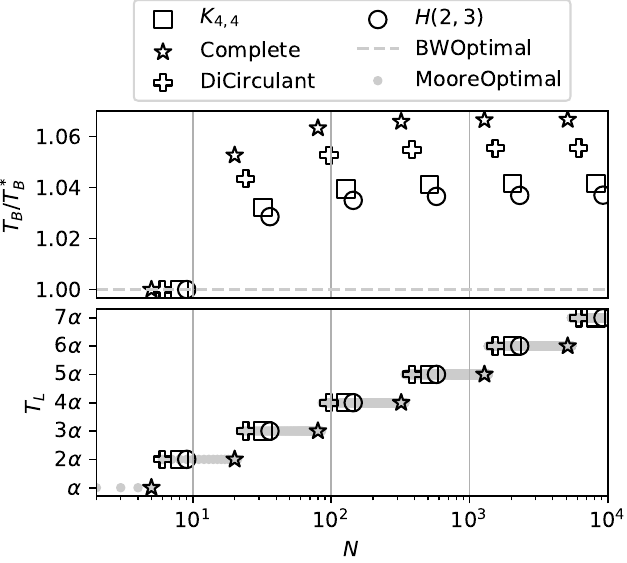}
        \caption{Line graph expansion on Moore and BW optimal degree-4 base graphs: \normalfont complete bipartite graph $K_{4,4}$, complete graph, directed circulant graph, and Hamming graph $H(2,3)$. $T_B^*\!=\!\frac{M}{B}\!\cdot\!\frac{N-1}{N}$ is the optimal BW runtime.}
        \label{fig:asymptotictrend}
    \end{minipage}
\end{figure*}

\section{Expansion Techniques}\label{sec:expansion}

We present three techniques that can be applied to construct near-optimal large-scale \emph{synthesized topologies} and schedules by expanding small-scale topologies and associated schedules.
The three techniques provide different options for increasing the size of the topology and the per-node degree, while preserving either total-hop latency or BW optimality of the base graph and schedule (Table \ref{tab:expansiontechniques}). While we describe the techniques in the context of allgather, corollary \ref{thm:scheduleswitchmap} in \S\ref{sec:reducescatterallgather} implies equivalent constructions for \textbf{reduce-scatter} and \textbf{allreduce}.

\subsection{Line Graph Expansion}\label{sec:linegraph}

We borrow the line graph transformation from graph theory~\cite{linegraph}, which transforms an input graph $G$ into a larger graph $L(G)$ as follows: every edge in $G$ becomes a node in $L(G)$, and two nodes in $L(G)$ are adjacent if the corresponding edges are adjacent in $G$ (Definition~\ref{linegraphdef}).

\noindent \textbf{Intuition:} Line graph expands an $N$-node degree-$d$ topology into a $dN$-node topology. The degree $d$ remains the same, which is crucial since the degree is often limited by hardware constraints like the number of available ports. While the number of nodes grows by $d$-fold, the diameter of the topology only increases by one, which is also optimal for total-hop latency and all-to-all performance. In addition, the paths in the base topology are mapped into the expanded topology, allowing the comm schedule for the base to be expanded as well. Line graph expansion can be applied multiple times to scale the topology and schedule to arbitrarily large sizes.

Figure~\ref{fig:linegraphtopology} gives an example of the line graph of the complete bipartite graph $K_{2,2}$. 
Any (shortest) path $w_0\!\shortarrow w_1\!\shortarrow\dots\!\shortarrow w_n$ in $K_{2,2}$ also becomes a (shortest) path $w_{-\!1}w_0\shortarrow w_0w_1\shortarrow\dots\shortarrow w_{n-1}w_n\shortarrow w_nw_{n+1}$ in $L(K_{2,2})$ from $w_{-\!1}w_0$ to $w_nw_{n+1}$, for any $w_{-\!1},w_{n+1}$ provided that $w_{-\!1}w_0\!\neq\! w_nw_{n+1}$.

Given an allgather schedule $A_G$ for $G$, we construct schedule $A_{L(G)}$ for $L(G)$. Pick any node $v'\!v$ in $L(G)$. It needs to broadcast its shard to every other node in $L(G)$. Pick any other node, say, $uu'$. For each element $x$ of $v'\!v$'s shard, we want to send $x$ to $uu'$. Since $v$ broadcasts $x$ to every other node in $A_G$, there is a path $v\shortarrow w_1\shortarrow\dots\shortarrow w_{n-1}\shortarrow u$ in $G$ along which $x$ is sent to $u$ in $A_G$. Thus, the path $v'\!v\shortarrow vw_1\shortarrow w_1w_2\shortarrow\dots\shortarrow w_{n-1} u\shortarrow uu'$ can be utilized to send $x$ from $v'\!v$ to $uu'$ in $L(G)$. 
\begin{restatable}[Schedule of Line Graph]{definition}{linescheduledef}\label{def:linegraph}
Given an allgather schedule $A_G$ for topology $G$, let $A_{L(G)}$ be the schedule for line graph $L(G)$ containing:
\begin{enumerate}[label=\arabic*.]
    \item $((v'\!v,S),(v'\!v,vu),1)$ for each edge $(v'\!v,vu)\in E_{L(G)}$ with $v'\!v\neq vu$. {\normalfont \bf [Insert the 1st comm step in $A_{L(G)}$.]}
    \item $((v'\!v,C),(uw,ww'),t+1)$ for each $((v,C),(u,w),t)\!\in\! A_G$ and $v'\!v \neq ww'$. {\normalfont \bf [Adapt $A_G$ to form $A_{L(G)}$.]}
\end{enumerate}
\end{restatable}
At the 1st comm step, $x$ is broadcasted by $v'\!v$ to every neighbor (e.g., $vw_1$). Then, for every $((v,C),(w_i,w_{i+1}),t)$ in $A_G$ with $x\in C$, there is $((v'\!v,C),(w_iw_{i+1},w_{i+1}w_{i+2}),t+1)$ in $A_{L(G)}$ that takes $x$ from $w_iw_{i+1}$ to $w_{i+1}w_{i+2}$ and, eventually, to $uu'$ ($v\!=\!w_0,u\!=\!w_n$). Since $x$ and $uu'$ are picked arbitrarily, $v'\!v$ broadcasts every element of its shard to all nodes in $L(G)$. Figure~\ref{fig:linegraphschedule} shows an example of schedule construction.

As for the performance of $A_{L(G)}$, we leave the mathematical details in Appendix~\ref{app-subsec:linegraph}. In practice, one can apply line graph expansion repeatedly to scale the topology and schedule indefinitely. The more optimal the base topology and schedule are, the more optimal the expanded topology and schedule will be. Figure~\ref{fig:asymptotictrend} shows how the performance evolves as we continuously apply line graph expansion to several Moore and BW optimal base graphs. The total-hop latency always remains Moore optimal. The BW performance deviates from optimality $T_B^*$ but remains a constant factor away asymptotically. A key observation in Figure~\ref{fig:asymptotictrend} is that the larger the size of the base graph is, the closer the expanded schedule is to BW optimality. 
Line graph expansion is notable for its ability to construct indefinitely large-scale topologies without increasing degree $d$. The expansion is also efficient in terms of diameter, making it ideal for all-to-all as well.

\subsection{Degree Expansion}

\textbf{Intuition:} While line graph expansion expands the number of nodes, degree expansion additionally expands the topology degree. Taking a base topology $G$, we make $n$ copies of it and connect two nodes in different copies if they are adjacent in $G$. This process forms an expanded topology $G*n$, which multiplies both the number of nodes and the degree of $G$ by $n$. Because the connections in $G*n$ are derived from $G$, similar to line graph expansion, we can map paths from $G$ to $G*n$ to expand the communication schedule of $G$ as well.

Figure~\ref{fig:degreeexampletopology} gives an example of expanding a 4-node unidirectional ring into an 8-node degree-2 topology (see formal definition of degree expanded topology in Definition~\ref{def:degreeexp}). 
Based on the input schedule $A_G$ for $G$, we construct a schedule $A_{G*n}$ for $G*n$. For any data traveling along $v\shortarrow w^{(1)}\shortarrow\dots\shortarrow w^{(m)}\shortarrow u$ in $A_G$, $A_{G*n}$ has the data travel along $v_i\shortarrow w^{(1)}_i\shortarrow\dots\shortarrow w^{(m)}_i\shortarrow u_j$ for all $i,j$. That is, data is transmitted within the $i$-th copy of $G$, except at the last step. With this construction, any node $u_i$ has broadcasted the data to all other nodes except its own copies $u_j$s. We add an additional comm step for $u_j$ to collect the data from its in-neighbors (see Figure~\ref{fig:degreeexampleschedule}).

\begin{restatable}[Degree Expanded Schedule]{definition}{degexpscheduledef}\label{def:degexpschedule}
Given an allgather schedule $A_G$ for $G$, construct $A_{G*n}$ for $G*n$:
\begin{enumerate}[label=\arabic*.]
    \item For all $i,\!j$ including $i\!=\!j$ and for each $(\!(v,\!C),(u,\!w),t\!)\!\in\! A_G$, add $((v_j,C),(u_j,w_i),t)$ to $A_{G*n}$;
    \item\label{degexp1ststep} Divide shard $S$ into equal-sized chunks $C_1,\dots,C_{nd}$. Given $u_i,u_j\in V_{G*n}$ with $i\neq j$, add $((u_i,C_\alpha),(v_\alpha,u_j),t_{\max}+1)$ to $A_{G*n}$ for each $(v_1,u_j),\dots,(v_{nd},u_j)\in E_{G*n}$, where $t_{\max}$ is the max comm step in $A_G$.
\end{enumerate}
\end{restatable}

Unlike line graph expansion, degree expansion preserves BW optimality. This is because the expanded broadcast paths from copies of an original node are totally disjoint from each other (Figure~\ref{fig:degreeexampleschedule}). However, degree expansion does not preserve Moore optimality. While line graph expansion does not change degree, degree expansion increases it, reducing the number of comm steps required for Moore optimality (\S\ref{app-subsec:degree}).

\begin{figure}[tb]
    \begin{subfigure}{.24\columnwidth}
        \captionsetup{skip=-11pt}
        \centering
        \input{figures/UniRing.tex}
        \caption{$G(N\!=\!4,d\!=\!1\!)$}
    \end{subfigure}
    \hfill
    \begin{subfigure}{.37\columnwidth}
        \captionsetup{skip=-11pt}
        \centering
        \input{figures/D_UniRing_1.tex}
        \caption{$G*2$ $(N\!=\!8,d\!=\!2)$}
        \label{fig:degreeexampletopology}
    \end{subfigure}
    \begin{subfigure}{.37\columnwidth}
        \captionsetup{skip=-11pt}
        \centering
        \input{figures/D_UniRing_2.tex}
        \caption{Broadcasts from $a_1,a_2$}
        \label{fig:degreeexampleschedule}
    \end{subfigure}
    \caption{4-node unidirectional ring and its degree expansion to $d\!=\!2$. \normalfont Figure (a) shows the base topology and broadcast path from $a$ to $b,c,d$. Figure (b) shows the expanded topology. Figure (c) shows the broadcast paths from $a_1$ and $a_2$ to other nodes marked in red and blue respectively. For any $u\neq a$, the path from $a_i$ to $u_j$ stays in $i$ until the very last step when it jumps to $u_j$, e.g., $a_1\!\shortarrow\! b_1\!\shortarrow\! c_2$ and $a_2\!\shortarrow\! b_2\!\shortarrow\! c_2\!\shortarrow\! d_1$. For $a_i$ to $a_j$, each in-neighbor of $a_j$ is tasked with sending an equal portion of $a_i$'s shard to $a_j$ in the end; for example, $d_1$ and $d_2$ each send half of $a_1$'s shard to $a_2$ and also each send half of $a_2$'s shard to $a_1$. The broadcast paths are disjoint from each other, resulting in BW optimality.}
    \label{fig:degreeexample}
\end{figure}

\subsection{Cartesian Product Expansion}\label{sec:cartesianpower}

The Cartesian product of two graphs $G_1,G_2$ is an expanded graph $G_1\square G_2$ with size and degree equal to the product of $G_1,G_2$'s sizes and the sum of their degrees, respectively.
\begin{restatable}[Cartesian Product]{definition}{cartesianprodtopodef}\label{def:cartesianprodtopodef}
The Cartesian product digraph $G_1\square G_2$ of digraphs $G_1$ and $G_2$ has vertex set $V_{G_1}\times V_{G_2}$ with vertex $\mathbf u=(u_1,u_2)$ connected to $\mathbf v=(v_1,v_2)$ iff either $(u_1,v_1)\in E_{G_1}$ and $u_2=v_2$; or $u_1=v_1$ and $(u_2,v_2)\in E_{G_2}$.
\end{restatable}
This definition generalizes to the Cartesian product of $n$ digraphs: $G_1\square G_2\square\dots\square G_n$. When $G_1\!=\!\dots\!=\!G_n\!=\!G$, the product is denoted as Cartesian power $G^{\square n}$.
We use Cartesian power and product in our topology and schedule expansion.

\noindent \textbf{Intuition:} The Cartesian product $G_1\square G_2\square\dots\square G_n$ consists of $n$ dimensions, with connections in dimension $i$ identical to $G_i$. Taking the schedules of $G_1,G_2,\dots,G_n$, we can balance the amount of traffic going through each dimension to achieve high BW performance. Cartesian product expansion greatly expands the set of topologies we construct, because we can take several of our base or constructed topologies to create a new product topology and its efficient schedule.

\noindent \textbf{Cartesian Power Expansion:} This is a special case of performing Cartesian product on identical graphs. Given a $d$-regular $G$ and schedule $A_G$, we can construct a schedule $A_{G^{\square n}}$ for $G^{\square n}$, which is $nd$-regular and has $|V_G|^n$ nodes. This technique helps generate efficient topologies, including some well-known ones like hypercube and Hamming graph.
We describe how to construct allgather schedules for Cartesian power graphs by using $\ell\times\ell$ torus ($\text{$\ell$-ring}^{\square 2}$) as an example. Taking \text{$\ell$-ring} allgather schedule $A$, a typical allgather schedule on $\ell\times\ell$ torus, as in hierarchical ring allreduce~\cite{hierarring}, is to perform schedule $A$ along rings on one dimension first and then the other dimension. Consider two schedules: $A^{(1)}$ performs allgather on vertical rings first and then horizontal ones; $A^{(2)}$ performs allgather in the opposite order. $A^{(1)},A^{(2)}$ use disjoint set of links at any comm step. Thus, we divide each shard in $\ell\times\ell$ torus into two halves and let them be allgathered by $A^{(1)},A^{(2)}$ separately. The combined schedule, where $A^{(1)}$ and $A^{(2)}$ are performed \emph{simultaneously}, is BW-optimal. The total-hop latency would be $2T_L(A)$.

The above torus schedule has appeared in previous literature~\cite{multiporttorus}. It can be generalized to generate schedules for Cartesian power of arbitrary topologies (see Appendix~\ref{app-subsec:cartesian}).

\noindent \textbf{Cartesian Product Expansion}\label{sec:spathcartesianexp} One can also construct a schedule for the Cartesian product of distinct topologies. For example, an $a\times b\times c$ 3D torus is the Cartesian product of three rings with lengths $a,b,c$. This schedule cannot be directly generated by construction but requires an additional BFB schedule generation technique that we introduce in \S\ref{sec:shortestpathschedule}. If individual topologies have BW-optimal BFB schedules, as in the case of any torus, then the schedule generated for the Cartesian product is BW-optimal (Theorem~\ref{thm:spathcartesianexpopt}).

%% file: figures/complete_bipartite_2.tex
\scalebox{0.75}{
\begin{tikzpicture}[roundnode/.style={circle,draw=black,minimum size=7mm}]
	\node[roundnode,line width=0.4mm]	(0)	at (-1.5*0.9, 1.5*0.9) {$a$};
	\node[roundnode]	(2)	at (1.5*0.9, 1.5*0.9)	{$c$};
	\node[roundnode]	(1) at (1.5*0.9, -1.5*0.9)	{$b$};
	\node[roundnode]	(3) at (-1.5*0.9, -1.5*0.9)	{$d$};
	\path[-latex,anchor=south,line width=0.4mm] (0) edge [bend left=25] node {\textbf{1}} (2);
	\path[-latex] (2) edge [bend left=25] (0);
	\path[-latex,anchor=west,line width=0.4mm] (2) edge [bend left=25] node {\textbf{2}} (1);
	\path[-latex] (1) edge [bend left=25] (2);
	\path[-latex,anchor=south,line width=0.4mm] (3) edge [bend left=25] node {\textbf{2}} (1);
	\path[-latex] (1) edge [bend left=25] (3);
	\path[-latex,anchor=west,line width=0.4mm] (0) edge [bend left=25] node {\textbf{1}} (3);
	\path[-latex] (3) edge [bend left=25] (0);
\end{tikzpicture}
}

%% file: figures/L_complete_bipartite.tex
\scalebox{0.75}{
\begin{tikzpicture}[roundnode/.style={circle,draw=black,minimum size=1mm,line width=0.25mm,fill=white},roundnode2/.style={circle,draw=gray,minimum size=1mm}]
	\node[roundnode2]	(0)	at (-1.5*0.9, 1.5*0.9) {\textcolor{gray}{$a$}};
	\node[roundnode2]	(2)	at (1.5*0.9, 1.5*0.9)	{\textcolor{gray}{$c$}};
	\node[roundnode2]	(1) at (1.5*0.9, -1.5*0.9)	{\textcolor{gray}{$b$}};
	\node[roundnode2]	(3) at (-1.5*0.9, -1.5*0.9)	{\textcolor{gray}{$d$}};
	\path[-latex,draw=gray,fill=gray] (0) edge [bend left=35] (2);
	\path[-latex,draw=gray,fill=gray] (2) edge [bend left=35] (0);
	\path[-latex,draw=gray,fill=gray] (2) edge [bend left=35] (1);
	\path[-latex,draw=gray,fill=gray] (1) edge [bend left=35] (2);
	\path[-latex,draw=gray,fill=gray] (3) edge [bend left=35] (1);
	\path[-latex,draw=gray,fill=gray] (1) edge [bend left=35] (3);
	\path[-latex,draw=gray,fill=gray] (0) edge [bend left=35] (3);
	\path[-latex,draw=gray,fill=gray] (3) edge [bend left=35] (0);
	
	\node[roundnode] (20) at (0, 0.9*0.9) {$ca$};
	\node[roundnode] (02) at (0, 2.2*0.9) {$ac$};
	\node[roundnode] (03) at (-0.9*0.9, 0) {$ad$};
	\node[roundnode] (30) at (-2.2*0.9, 0) {$da$};
	\node[roundnode] (31) at (0, -0.9*0.9) {$db$};
	\node[roundnode] (13) at (0, -2.2*0.9) {$bd$};
	\node[roundnode] (12) at (0.9*0.9, 0) {$bc$};
	\node[roundnode] (21) at (2.2*0.9, 0) {$cb$};
	
	\path[-latex,line width=0.25mm] (02) edge [bend left=40] (20);
	\path[-latex,line width=0.25mm] (20) edge [bend left=40] (02);
	\path[-latex,line width=0.25mm] (03) edge [bend left=40] (30);
	\path[-latex,line width=0.25mm] (30) edge [bend left=40] (03);
	\path[-latex,line width=0.25mm] (13) edge [bend left=40] (31);
	\path[-latex,line width=0.25mm] (31) edge [bend left=40] (13);
	\path[-latex,line width=0.25mm] (12) edge [bend left=40] (21);
	\path[-latex,line width=0.25mm] (21) edge [bend left=40] (12);
	
	\path[-latex,line width=0.25mm] (02) edge [bend left=40] (21);
	\path[-latex,line width=0.25mm] (21) edge [bend left=40] (13);
	\path[-latex,line width=0.25mm] (13) edge [bend left=40] (30);
	\path[-latex,line width=0.25mm] (30) edge [bend left=40] (02);
	
	\path[-latex,line width=0.25mm] (20) edge [bend left=40] (03);
	\path[-latex,line width=0.25mm] (03) edge [bend left=40] (31);
	\path[-latex,line width=0.25mm] (31) edge [bend left=40] (12);
	\path[-latex,line width=0.25mm] (12) edge [bend left=40] (20);
\end{tikzpicture}
}

%% file: figures/D_complete_bipartite.tex
\scalebox{0.75}{
\begin{tikzpicture}[roundnode/.style={circle,draw=black,minimum size=1mm,line width=0.25mm,fill=white},roundnode2/.style={circle,draw=gray,minimum size=1mm}]
	\node[roundnode,line width=0.4mm] (20) at (0, 0.9*0.9) {$ca$};
	\node[roundnode] (02) at (0, 2.2*0.9) {$ac$};
	\node[roundnode] (03) at (-0.9*0.9, 0) {$ad$};
	\node[roundnode] (30) at (-2.2*0.9, 0) {$da$};
	\node[roundnode] (31) at (0, -0.9*0.9) {$db$};
	\node[roundnode] (13) at (0, -2.2*0.9) {$bd$};
	\node[roundnode] (12) at (0.9*0.9, 0) {$bc$};
	\node[roundnode] (21) at (2.2*0.9, 0) {$cb$};
	
	\path[-latex] (02) edge [bend left=40] (20);
	\path[-latex,line width=0.4mm,anchor=east] (20) edge [bend left=40] node {\textbf{1}} (02);
	\path[-latex,line width=0.4mm,anchor=north] (03) edge [bend left=40] node {\textbf{2}} (30);
	\path[-latex] (30) edge [bend left=40] (03);
	\path[-latex] (13) edge [bend left=40] (31);
	\path[-latex,line width=0.4mm,anchor=west] (31) edge [bend left=40] node {\textbf{3}} (13);
	\path[-latex] (12) edge [bend left=40] (21);
	\path[-latex,line width=0.4mm,anchor=north] (21) edge [bend left=40] node {\textbf{3}} (12);
	
	\path[-latex,line width=0.4mm,anchor=south west] (02) edge [bend left=40] node {\textbf{2}} (21);
	\path[-latex,line width=0.4mm,anchor=north west] (21) edge [bend left=40] node {\textbf{3}} (13);
	\path[-latex] (13) edge [bend left=40] (30);
	\path[-latex] (30) edge [bend left=40] (02);
	
	\path[-latex,line width=0.4mm,anchor=south east] (20) edge [bend left=40] node {\textbf{1}} (03);
	\path[-latex,line width=0.4mm,anchor=north east] (03) edge [bend left=40] node {\textbf{2}} (31);
	\path[-latex,line width=0.4mm,anchor=north west] (31) edge [bend left=40] node {\textbf{3}} (12);
	\path[-latex] (12) edge [bend left=40] (20);
\end{tikzpicture}
}

%% file: figures/UniRing.tex
\scalebox{0.7}{
\begin{tikzpicture}[roundnode/.style={circle,draw=black,minimum size=1mm,fill=white}]
	\node[roundnode]	(0)	at (-0.9, 0.9) {$a$};
	\node[roundnode]	(1)	at (0.9, 0.9)	{$b$};
	\node[roundnode]	(2) at (0.9, -0.9)	{$c$};
        \node[roundnode]	(3) at (-0.9, -0.9)	{$d$};
	
	\path[-latex,line width=0.5mm] (0) edge [bend left=30] (1);
	\path[-latex,line width=0.5mm] (1) edge [bend left=30] (2);
        \path[-latex,line width=0.5mm] (2) edge [bend left=30] (3);
	\path[-latex] (3) edge [bend left=30] (0);

        \node[roundnode,opacity=0] (21) at (0, -1.5) {$c_1$};
\end{tikzpicture}
}

%% file: figures/D_UniRing_1.tex
\scalebox{0.68}{
\begin{tikzpicture}[roundnode/.style={circle,draw=black,minimum size=1mm,fill=white}]
        \node[roundnode] (01) at (-1.5, 1.5) {$a_1$};
	\node[roundnode] (11) at (1.5, 1.5) {$b_1$};
	\node[roundnode] (21) at (1.5, -1.5) {$c_1$};
        \node[roundnode] (31) at (-1.5, -1.5) {$d_1$};
	\node[roundnode]	(02) at (-0.65, 0.65) {$a_2$};
	\node[roundnode] (12) at (0.65, 0.65)	{$b_2$};
	\node[roundnode] (22) at (0.65, -0.65) {$c_2$};
        \node[roundnode] (32) at (-0.65, -0.65) {$d_2$};
	
	\path[-latex] (01) edge [bend left=35] (11);
	\path[-latex] (11) edge [bend left=35] (21);
        \path[-latex] (21) edge [bend left=35] (31);
	\path[-latex] (31) edge [bend left=35] (01);
 
	\path[-latex] (02) edge [bend left=25] (12);
	\path[-latex] (12) edge [bend left=25] (22);
	\path[-latex] (22) edge [bend left=25] (32);
	\path[-latex] (32) edge [bend left=25] (02);
 
	\path[-latex] (01) edge [bend left=40] (12);
	\path[-latex] (11) edge [bend left=40] (22);
        \path[-latex] (21) edge [bend left=40] (32);
	\path[-latex] (31) edge [bend left=40] (02);
 
	\path[-latex] (02) edge [bend left=40] (11);
	\path[-latex] (12) edge [bend left=40] (21);
	\path[-latex] (22) edge [bend left=40] (31);
	\path[-latex] (32) edge [bend left=40] (01);
\end{tikzpicture}
}

%% file: figures/D_UniRing_2.tex
\scalebox{0.68}{
\begin{tikzpicture}[roundnode/.style={circle,draw=black,minimum size=1mm,fill=white}]
        \node[roundnode,draw=red,line width=0.4mm] (01) at (-1.5, 1.5) {$a_1$};
	\node[roundnode] (11) at (1.5, 1.5) {$b_1$};
	\node[roundnode] (21) at (1.5, -1.5) {$c_1$};
        \node[roundnode] (31) at (-1.5, -1.5) {$d_1$};
	\node[roundnode,draw=blue,line width=0.4mm]	(02) at (-0.65, 0.65) {$a_2$};
	\node[roundnode] (12) at (0.65, 0.65)	{$b_2$};
	\node[roundnode] (22) at (0.65, -0.65) {$c_2$};
        \node[roundnode] (32) at (-0.65, -0.65) {$d_2$};
	
	\path[-latex,line width=0.5mm,red] (01) edge [bend left=35] (11);
	\path[-latex,line width=0.5mm,red] (11) edge [bend left=35] (21);
        \path[-latex,line width=0.5mm,red] (21) edge [bend left=35] (31);
	\path[-latex,line width=0.5mm,blue] (31) edge [bend left=35] (01);
 
	\path[-latex,line width=0.5mm,blue] (02) edge [bend left=25] (12);
	\path[-latex,line width=0.5mm,blue] (12) edge [bend left=25] (22);
	\path[-latex,line width=0.5mm,blue] (22) edge [bend left=25] (32);
	\path[-latex,line width=0.5mm,red] (32) edge [bend left=25] (02);
 
	\path[-latex,line width=0.5mm,red] (01) edge [bend left=40] (12);
	\path[-latex,line width=0.5mm,red] (11) edge [bend left=40] (22);
        \path[-latex,line width=0.5mm,red] (21) edge [bend left=40] (32);
	\path[-latex,line width=0.5mm,red] (31) edge [bend left=40] (02);
 
	\path[-latex,line width=0.5mm,blue] (02) edge [bend left=40] (11);
	\path[-latex,line width=0.5mm,blue] (12) edge [bend left=40] (21);
	\path[-latex,line width=0.5mm,blue] (22) edge [bend left=40] (31);
	\path[-latex,line width=0.5mm,blue] (32) edge [bend left=40] (01);
\end{tikzpicture}
}

%% file: finder.tex
\begin{table}[tb]
    \centering
    \resizebox{\columnwidth}{!}{
    \begin{tabular}{|l|c|c|c|c|c|}
        \hline
        Expansion Techniques & \# of Nodes & Deg & Moore & BW & Perf \\
        \hline
        Line Graph Exp $L^n(G)$ & $d^nN$ & $d$ & \checkmark & $\times$ & Thm \ref{thm:linegraphconclusion} \\
        \hline
        Degree Exp $G*n$ & $nN$ & $nd$ & $\times$ & \checkmark & Thm \ref{thm:degexp} \\
        \hline
        Cartesian Power $G^{\square n}$ & $N^n$ & $nd$ & $\times$ & \checkmark & Thm \ref{thm:cartesianpowperformance} \\
        \hline
        Cartesian Prod $G_1\square\!\dots\!\square G_n$ & $\prod_i N_i$ & $\sum_i d_i$ & $\times$ & \checkmark & Thm \ref{thm:spathcartesianexpopt} \\
        \hline
    \end{tabular}
    }
    \caption{Summary of expansion techniques. \normalfont The table shows the characteristics of the resulting topology and schedule after applying expansion techniques to an $N$-node degree-$d$ base topology. ``$\checkmark,\times$'' show whether the expansion preserves Moore/BW optimality. The last column refers to the theorems that give the exact performance of expanded schedules.}
    \label{tab:expansiontechniques}
    \resizebox{\columnwidth}{!}{
    \begin{tabular}{|l||c|c|c||c|c||}
	\hline
	Topology & $T_L$ & $T_B$ & \scalebox{0.9}[1]{$2(T_L\!+\!T_B)$} & \scalebox{0.9}[1]{$D(G)$} & \scalebox{0.9}[1]{All-to-All} \\
	\hline
	$\Pi_{4,1024}$ & $5\alpha$ & $1.332\sfrac{M}{B}$ & 323.5us & $5$ & 409.1us \\
	$L^3(C(16,\{3,4\}))$ & $6\alpha$ & $1.020\sfrac{M}{B}$ & 291.0us & $6$ & 403.5us \\
	$L^2(\text{Diamond}^{\square 2})$ & $8\alpha$ & $1.004\sfrac{M}{B}$ & 328.4us & $8$ & 446.6us \\
	$L(\text{DBJMod}(2,4)^{\square 2})$ & $11\alpha$ & $1.000\sfrac{M}{B}$ & 387.8us & $9$ & 529.9us \\
	\scalebox{0.9}[1]{$(\text{UniRing}(1,4)\square \text{UniRing}(1,8))^{\square 2}$} & $20\alpha$ & $0.999\sfrac{M}{B}$ & 567.6us & $20$ & 1174.4us \\
	\hline
	\textbf{Theoretical Bound} & $\mathbf{5}\alpha$ & $\mathbf{0.999}\sfrac{M}{B}$ & \textbf{267.6us} & $\mathbf{5}$ & \textbf{382.3us} \\
	\hline
    \end{tabular}
    }
    \caption{Pareto-efficient topologies at $N\!=\!1024$, $d\!=\!4$. \normalfont The $2(T_L\!+\!T_B)$ column shows the allreduce runtimes for $\alpha\!=\!10\mu s$ and $M/B\!=\!1{\normalfont MB}/100{\normalfont Gbps}$. We multiply $T_L\!+\!T_B$ by 2 because allreduce is performed by combining reduce-scatter and allgather. The all-to-all time is computed via multi-commodity flow (Appendix~\ref{app-sec:all-to-all}) with each node having 1MB of data to send (i.e., sending $1/N$ MB to each node). For comparison, the baselines Shifted Ring and Double Binary Tree (\S\ref{sec:exptsetup}) have allreduce times of $20640$us and $1434$us, and all-to-all times of $10738$us and $21475$us, respectively. Table \ref{table:toposummary} shows the details of the base topologies.}
    \label{table:paretooptimal}
\end{table}

\subsection{Topology Finder}\label{sec:topofinder}

The goal of Topology Finder is to produce the best topologies and schedules for a target $N$ and $d$. If we aim for asymptotic performance (fixed $d$, $N\!\to\!\infty$), as in Theorem~\ref{thm:linebwopt}, we want \textbf{the base topology to be as large as possible and the base schedule to be as optimal as possible.} However, for a target $N$ and $d$, only base topologies with certain sizes (e.g., divisors of $N$) and degrees can be expanded to the target. Thus, we keep a database of known base topologies and their schedules (Table~\ref{table:toposummary}). These topologies and schedules are highly optimized and cover a wide range of $N$ and $d$.

Given base topologies, we perform a bottom-up search for the combinations of expansion techniques to reach the target $N$ and $d$. We iteratively apply expansions to candidates.
At intermediate sizes, we prune candidates with inferior performance and keep the best ones for further expansion. 
Because each expansion multiplies the topology size (Table~\ref{tab:expansiontechniques}), the number of expansions that can be applied before the size gets too large---and hence the number of possible combinations---is limited, making the search feasible.

While we expand the topologies, proved theorems (Table~\ref{tab:expansiontechniques}) allow us to predict the performance of expanded topologies. This is vital for the search because it is intractable to construct schedules for every possible topology and compare their performance. Being able to predict with an analytical formula allows us to compare different topologies and prune inferior ones. We keep a Pareto frontier of topologies for a given $N$ and $d$. A topology is inferior to another only if it is worse in both total-hop latency and BW runtime. Ultimately, the search finds all Pareto-efficient topologies for the target $N$ and $d$.
Depending on the testbed, we convert unidirectional topologies to bidirectional ones (Appendix~\ref{sec:unitobidirectional}).
We then determine the best-performing topology for a specific workload. 

Table~\ref{table:paretooptimal} shows the result for $N\!=\!1024$ and $d\!=\!4$. From top to bottom, the Pareto frontier exhibits an increasing $T_L$ and a decreasing $T_B$, with the top and bottom being Moore and BW optimal, respectively. On the all-to-all side, the diameters of the topologies also follow the same trend as $T_L$ because of $T_L\!\geq\!\alpha\!\cdot\! D(G)$ (Theorem~\ref{thm:diameter}). Table~\ref{table:paretooptimal} also shows the allreduce and all-to-all times calculated based on specific $\alpha,M,B$. Notably, the line graph of circulant graph $L^3(C(16,\{3,4\}))$ has both the lowest allreduce and all-to-all times, within 9\% and 6\% of the theoretical bounds. Table~\ref{tab:allourtopos} in appendix contains more results for $N\!=\!32,64,\dots,1024$. 

While low-hop/diameter indicates high all-to-all throughput, other metrics like average distance between nodes~\cite{diameter1, diameter2} also play a role. Thus, despite having a lower diameter, $\Pi_{4,1024}$ underperforms $L^3(C(16,\{3,4\}))$. Including other metrics makes the search process more complex and computationally expensive. In practice, $T_L$ and $D(G)$ are feasible and accurate enough for predicting all-to-all throughput.

In DNN training experiments, we use one topology for the entire training due to the high reconfiguration latency of our target patch panel platform. In such a case, we select the best option based on the distribution of collective sizes $M$s, which could be the layer sizes or a fixed value like the bucket size in PyTorch DDP~\cite{pytorch-dist}, depending on the training framework. With faster reconfiguration, one could change topology to optimize for different collective runs during training.

Our implementation runs under a minute for all $d\!=\!2,4,8,16$ and $N$ up to 2000. While this can be sped up, we find it acceptable for now given that the search is performed once for all $N$s and $d$s, and results can be saved for future use.

%% file: BFS_schedule.tex
\section{Breadth-First-Broadcast (BFB) Schedule}\label{sec:shortestpathschedule}

We now present a scalable algorithm for finding schedules for \emph{generative topologies} that are directly instantiated based on graph theory as well as topologies obtained from Cartesian Product expansion, which is the only expansion technique that doesn't yield a schedule.
State-of-the-art schedule generations (e.g., Blink~\cite{wang2020blink}, SCCL~\cite{SCCL}, and TACCL~\cite{shah2021synthesizing}) can scale only to a modest number of nodes because they involve NP-hard optimization.
To ensure polynomial-time generation, we impose a \textit{breadth-first} broadcast order from each node such that (1) all communications between a pair of nodes rely only on the shortest paths between them; (2) the schedule is structured as a series of comm steps, where each comm step is responsible for eagerly transmitting data to a set of nodes that is one additional hop away.
Our BFB schedule generation technique does not guarantee optimality in an arbitrary topology, given these constraints prohibit the use of longer paths or delayed (non-eager) transmissions along paths, but these constraints enable polynomial-time generation.

Despite these constraints, BFB schedules guarantee the following: (1) The schedules have the lowest total-hop latency as all data is eagerly communicated over shortest paths. (2) For Cartesian product topologies, BFB schedule generation yields a BW-optimal solution if the underlying product components admit BW-optimal BFB schedules (thus yielding optimal schedules for networks such as torus with arbitrary dimensions). (3) BFB schedules are also provably BW-optimal for many generative topologies with inherent symmetry.

\subsection{BFB Schedule Generation Linear Program}\label{sec:bfblp}

\textbf{Intuition:} Allgather is a simultaneous broadcast from every node in the topology. A BFB allgather schedule, as the name suggests, enforces a \textit{breadth-first} broadcast from every node. At each comm step $t$, every node $u$ needs to receive the data from all nodes $v$ for which the distance from $v$ to $u$ is $t$, through the ingress links. BFB uses a linear program to balance the traffic across the ingress links to optimize BW performance.

At each comm step $t$, the BFB schedule requires that for every node $v$, all nodes at a distance $t$ from $v$, i.e., $N^+_t(v)$, receive $v$'s data shard within the comm step. To achieve this, all nodes at distance $t\!-\!1$, i.e. $N^+_{t-1}(v)$, need to collectively multicast the data shard to nodes $N^+_t(v)$ in comm step $t$.

Given any $v$ and $u\!\in\! N^+_t(v)$, $u$ may have multiple in-neighbor $w$s in $N^+_{t-1}(v)$. All of them can provide $v$'s data shard because they have received it in comm step $t-1$. Since the BW runtime of a comm step equals the transmission time of the most congested link, a question is \textbf{how to allocate the amount of data $u$ receives from each $w$ to balance the workload on links?} Figure~\ref{fig:shortestpath-illustration} shows an example. Here, $u_1$ needs to get $v_1$'s shard from $w_1,w_2$ and $v_2$'s shard from $w_2$. The solution is simple: since $u_1$ can only get $v_2$'s shard from $w_2$, we let $w_1$ send $v_1$'s shard and $w_2$ send $v_2$'s shard, achieving a perfectly balanced workload. For $u_2$, it is more complicated. We formulate such a problem as a linear program:
\begin{equation}\label{lpmodel}
    \begin{array}{lr@{}ll}
    \text{minimize}  & & U_{u,t} &\\
    \text{subject to}& \displaystyle\sum_{v} & x_{v,(w,u),t}\leq U_{u,t}, & \makebox[2.6cm][l]{$\forall w\!\in\!N^-\!(u)\!=\!N^-_1\!(u)$}\\
                     & \displaystyle\sum_{w} & x_{v,(w,u),t}=1, & \forall v\!\in\! N^-_t\!(u)\\
                     & 0\leq &x_{v,(w,u),t}\leq 1. & \forall w,v
    \end{array}
\end{equation}
$x_{v,(w,u),t}$ is the proportion of $v$'s shard that is sent from $w$ to $u$ and is defined for every $v,w$ such that $w\!\in\! N^-(u)$ and $d(v,u)\!=\!d(v,w)\!+\!1\!=\!t$. $U_{u,t}$ is the max workload among links to $u$, i.e., $(\!w_1\!,\!u_2\!)\!,\!(\!w_2\!,\!u_2\!)\!,\!(\!w_3\!,\!u_2\!)$ in the case of $u_2$. \textbf{Minimizing $U_{u,t}$ is equivalent to minimizing \mbox{$\frac{M/N}{B/d}\!\cdot\! U_{u,t}$}, the max transmission time among links to $u$ at comm step $t$.} The 1st and 2nd constraints ensure correct max workload and $u$ receiving all data shards, respectively. Appendix~\ref{app:shortestpath} gives the specific LP for $u_2$, and the solution is shown in blue in Figure~\ref{fig:shortestpath-illustration}. The workload is also balanced with each link sending $2/3$ shard and hence BW runtime $\frac{M/N}{B/d}\!\cdot\!\frac{2}{3}$.

\begin{figure}[tb]
    \centering
    \input{figures/BFS_example}
    \caption{Example of BFB allgather schedule at comm step $t$. \normalfont Here, $w_1,w_2\in N^+_{t-1}(v_1)$ and $w_2,w_3\in N^+_{t-1}(v_2)$. $u_1,u_2$ are at distance $t$ from both $v_1,v_2$, so they both need to receive the data shards of $v_1,v_2$ in comm step $t$. Note that $u_1$ cannot get $v_2$'s shard from $w_3$ because $w_3$ is not an in-neighbor of $u_1$. The figure also shows the solutions to LPs (\ref{lpmodel}). The red and blue are independent LPs optimizing $U_{u_1,t}$ and $U_{u_2,t}$ respectively.}
    \label{fig:shortestpath-illustration}
\end{figure}

SCCL~\cite{SCCL} and TACCL~\cite{shah2021synthesizing} use NP-hard optimizations with discrete variables used to ensure each chunk is received before being sent. In contrast, we do not need discrete variables. A key observation from Figure~\ref{fig:shortestpath-illustration} is that because $w_1,w_2,w_3$ all receive the entire shard of $v_1$ at comm step $t\!-\!1$, the $x_{v_1,(w_1,u_2),t}\!=\!2/3$ and $x_{v_1,(w_2,u_2),t}\!=\!1/3$ in the solution can be any portions of the data shard, as long as their union is the entire shard. Assuming $[0,1]$ is the entire shard of $v_1$, no matter the $2/3$ sent by $w_1$ to $u_2$ is $[0,\frac{2}{3}]$ or $[\frac{1}{3},1]$, the $1/3$ sent by $w_2$ can simply be $[\frac{2}{3},1]$ or $[0,\frac{1}{3}]$ accordingly. Thus, \textbf{we only need to decide the amount of data sent on each link, which are continuous variables}, enabling polynomial-time schedule generation. To obtain a complete schedule, one needs to solve an LP~(\ref{lpmodel}) for each $u\!\in\! V_G$ and $t\!\in\!\{1,\dots,D(G)\}$. The BW runtime of the generated schedule is
\begin{equation}\label{eq:lpT_B}
    T_B=\frac{M/N}{B/d}\sum_{t=1}^{D(G)}\max_{u\in V_G}U_{u,t}.
\end{equation}
One could create an LP incorporating all $U_{u,t}$s and minimize $T_B$ (\ref{eq:lpT_B}) \textit{``globally''}. However, the result is equivalent to individually solving small LPs~(\ref{lpmodel}) for each $u$ and $t$. This is because the LPs are independent of each other, e.g., the decisions made to minimize $U_{u_2,t}$ in Figure~\ref{fig:shortestpath-illustration} do not affect $U_{u_1,t}$, and vice versa. The advantage of small LPs is that they can be solved in parallel utilizing multi-core processors.
Due to the breadth-first nature of BFB, data always follows the shortest paths between source and destination.
Thus, the number of comm steps of BFB schedule always equals the graph diameter, i.e., $T_L\!=\!\alpha\!\cdot\! D(G)$, the lowest possible $T_L$ given $G$.

Appendix~\ref{app:shortestpath} analyzes the BFB schedule and includes modifications to generate \textbf{discrete chunked schedules}~(\S\ref{sec:shortestIP}) and schedules for \textbf{heterogeneous link bandwidths}~(\S\ref{sec:heterogeneous}). Corollary \ref{thm:scheduleswitchfunc} shows how to convert to a \textbf{reduce-scatter} schedule.

\subsection{Generative Topologies}\label{sec:generateivetopo}

Generative topologies, unlike synthesized ones, are large graphs directly borrowed from graph theory. They are too large for manual or NP-hard schedule generation. Thus, we use the BFB linear program to generate schedules. Since a BFB schedule always has the lowest $T_L$ for a topology, if it is also BW-optimal, then it is \textbf{the optimal schedule} for that topology. Generative topologies often have symmetries that allow us to prove optimality mathematically. Their low diameters are also ideal for all-to-all throughput.

\textbf{Torus} is a widely used topology in parallel computing systems. Our torus schedule generated by BFB is theoretically optimal and represents a significant improvement over traditional torus schedules~\cite{multiporttorus}. Given a $d_1\!\times\! d_2\!\times\!\dots\!\times\! d_n$ torus, the traditional schedule, which performs parallel ring collectives on dimensions, only works (or is efficient) when dimensions are equal, i.e., $d_1\!=\!d_2\!=\!\dots\!=\!d_n$, and has $T_L\!=\!\sum_i(d_i-1)\alpha$. BFB torus schedule, however, is BW-optimal with any $d_i$s and $T_L\!=\!\sum_i\lfloor d_i/2\rfloor\alpha$. The BW optimality is due to torus being the \textit{Cartesian product} of rings, each of which has a BW-optimal BFB schedule (Theorem \ref{thm:spathcartesianexpopt}). BFB torus opens up many more construction possibilities since $d_i$s can be any combination. In our evaluations (\S\ref{sec:supercomp}), we compare BFB and traditional torus schedules on supercomputing clusters.

\textbf{Generalized Kautz Graph} (\S\ref{sec:genkautz}) and \textbf{Circulant Graph} (\S\ref{sec:circulant}) are a pair of versatile graphs in our toolbox. The former can be constructed for any $N$ and $d$, while the latter can be constructed for any $N$ and even-value $d$. Furthermore, the BFB schedule of the former is at most one $\alpha$ away from Moore optimality, making it the topology with the lowest $T_L$, while the latter always has a BW-optimal BFB schedule. Thus, they can fill gaps in $N$ and $d$ that expansion techniques fail to cover (e.g., prime $N$) or provide good candidates.

Besides the aforementioned ones, the following graphs also have optimal schedules by BFB. \textbf{Distance-Regular Graph} (\S\ref{sec:dist-reg}) is a family of large highly-symmetric graphs that are both BW-optimal and low-hop at the same time. The \textbf{Twisted Torus}~\cite{twistedtorus} used by TPU v4~\cite{tpuv4} is also computationally verified to be BW-optimal for at least $N\!\leq\! 10^4$. A \textbf{BFB Ring Schedule} with half the $T_L$ of traditional one is shown in \S\ref{sec:spathbiring}.

%% file: figures/BFS_example.tex
\scalebox{0.8}{
\begin{tikzpicture}[roundnode/.style={circle,draw=black,minimum size=7mm}]
	\node[roundnode] (0) at (-3.5, 0.75) {$v_1$};
	\node[roundnode] (1) at (-3.5, -0.75) {$v_2$};
	\node[roundnode] (2) at (0, 1.25) {$w_1$};
	\node[roundnode] (3) at (0, 0) {$w_2$};
	\node[roundnode] (4) at (0, -1.25) {$w_3$};
	\node[roundnode,color=red] (5) at (3.5, 1.25) {$u_1$};
        \node[anchor=west,color=red] at (5.north east) {$\substack{U_{u_1,t}=1}$};
	\node[roundnode,color=blue] (6) at (3.5, -1.25) {$u_2$};
        \node[anchor=west,color=blue] at (6.south east) {$\substack{U_{u_2,t}=2/3}$};
	
	\path[-latex,dashed,anchor=south] (0) edge [bend left=10] node {$t-1$} (2);
	\path[-latex,dashed,anchor=south] (0) edge [bend right=10] node {$t-1$} (3);
	\path[-latex,dashed,anchor=north] (1) edge [bend left=10] node {$t-1$} (3);
	\path[-latex,dashed,anchor=north] (1) edge [bend right=10] node {$t-1$} (4);
	
	\path[-latex,anchor=south,color=red] (2) edge [bend left=10] node {$\substack{x_{v_1,(w_1,u_1),t}=1}$} (5);
	\path[-latex,anchor=180,color=red] (3) edge [bend right=10] node {$\hspace{7mm}\substack{x_{v_1,(w_2,u_1),t}=0\\ x_{v_2,(w_2,u_1),t}=1}$} (5);
	\path[-latex,anchor=163,color=blue] (2) edge [bend left=10] node[pos=0.7] {$\hspace{9mm}\substack{x_{v_1,(w_1,u_2),t}=2/3}$} (6);
	\path[-latex,anchor=15,color=blue] (3) edge [bend left=15] node[pos=0.65] {$\substack{x_{v_1,(w_2,u_2),t}=1/3\\ x_{v_2,(w_2,u_2),t}=1/3}$} (6);
	\path[-latex,anchor=north,color=blue] (4) edge [bend right=10] node {$\substack{x_{v_2,(w_3,u_2),t}=2/3}$} (6);
\end{tikzpicture}
}

%% file: implementation2.tex
\section{Schedule Compilation}
\label{sec:implementation}
We implemented two compilers for lowering communication schedules to both GPU and CPU clusters, given the significance of collective communication for both ML and HPC workloads. We lowered over 1K schedules for various topologies and configurations. The compilers enable us to evaluate the performance of our topologies and schedules on hardware and to validate our mathematical model.

For GPUs, our compiler lowers a mathematically defined schedule to an XML file that can be executed by the MSCCL runtime~\cite{msccl}. MSCCL is an open-source collective communication library that extends NCCL~\cite{nccl} with an interpreter providing the ability to program custom schedules. Communication schedules are defined in XML as instructions (send/receive/reduce/copy) for each GPU threadblock. Our compiler also performs certain optimizations, such as consolidating non-contiguous sends using a scratch buffer and evenly distributing the computational workload across threadblocks.

For CPU-based supercomputers, we use Intel oneCCL~\cite{oneccl} + libfabric~\cite{libfabric} to execute schedules. We extended oneCCL with an interpreter that executes XMLs. The mathematical schedules are lowered into instructions (send/recv/reduce/copy/sync) for CPUs in an XML file and then executed.

%% file: eval.tex
\section{Evaluation}\label{sec:evaluation}

We present performance evaluation results on a 12-node direct-connect optical GPU testbed and a supercomputing torus CPU cluster with up to 54 nodes. We also present analytical and simulation results at larger scales.

\noindent \textbf{Collective Communication:} On 12-node testbed, topologies from our topology finder consistently outperform baselines in allreduce, reduce-scatter, and allgather (\S\ref{sec:arperf}, Fig~\ref{fig:arexptres}, Fig~\ref{fig:rs_ag_results}). Analytical model shows order-of-magnitude improvements in allreduce and all-to-all performance at larger scales (Fig~\ref{fig:comparisonwithringdbt}).

\noindent \textbf{DNN Training:} In data-parallel training experiments, our topologies achieve the best performance for both small models and \mbox{GPT-2}~\cite{gpt2} (\S\ref{sec:dnntrain}, Fig~\ref{fig:testbedtraining}). In simulated large-scale training, our topologies demonstrate order-of-magnitude improvements in all-to-all involved expert-parallel training (Fig~\ref{fig:simmoetraining}).

\noindent \textbf{Schedule Generation:} While generating optimal schedules, BFB is orders of magnitude faster and more scalable than SCCL~\cite{SCCL} and TACCL~\cite{shah2021synthesizing} (\S\ref{sec:scalcomp}, Table~\ref{tab:sccltacclbfs} \& Fig~\ref{fig:sccltacclbfs}). On supercomputing torus clusters, BFB outperforms traditional scheduling~\cite{multiporttorus}, SCCL, and TACCL (\S\ref{sec:supercomp}, Fig~\ref{fig:supercomp}).

Finally, we also conducted experiments on our testbed to compare BFB against widely adopted solutions for switch networks (e.g., NCCL and recursive halving \& doubling) (\S\ref{sec:switchcomp}), and to validate the $\alpha$-$\beta$ cost model (\S\ref{sec:cost-model}).

\input{eval_testbed}

%% file: eval_testbed.tex
\begin{figure*}[tb]
    \begin{minipage}{0.36\textwidth}
        \centering
        \vspace{0.5cm}
        \scalebox{0.75}{
        \begin{tabular}{|c|l|c|}
            \hline
            $N$ & \textbf{Topology} & $T_L$ \\
            \hline
            5 & Complete Graph: $K_5$ & $2\alpha$ \\
            \hline
            \multirow{2}{*}{6} & \multirow{2}{*}{\shortstack[l]{Degree Expansion of \\ Complete graph: $K_3*2$}} & \multirow{2}{*}{$4\alpha$} \\
            & & \\
            \hline
            7 & Circulant Graph: $C(7,\{2,3\})$ & $4\alpha$ \\
            \hline
            8 & Complete Bipartite Graph: $K_{4,4}$ & $4\alpha$ \\
            \hline
            9 & Hamming Graph: $H(2,3)$ & $4\alpha$ \\
            \hline
            \multirow{2}{*}{10} & \multirow{2}{*}{\shortstack[l]{Degree Expansion of BFB augmented \\ Bidirectional Ring: $\text{BiRing}(2,5)*2$}} & \multirow{2}{*}{$4\alpha$} \\
            & & \\
            \hline
            11 & Circulant Graph: $C(11,\{2,3\})$ & $4\alpha$ \\
            \hline
            12 & Circulant Graph: $C(12,\{2,3\})$ & $4\alpha$ \\
            \hline
        \end{tabular}
        }
        \captionof{table}{OurBestTopo at $d\!=\!4$ generated by topology finder~(\S\ref{sec:topofinder}). \normalfont All topologies listed above are BW-optimal.}
        \label{tab:ourtopologies}
    \end{minipage}
    \hfill
    \begin{minipage}{0.62\textwidth}
        \centering
        \vspace{0.5cm}
        \includegraphics[scale=0.49]{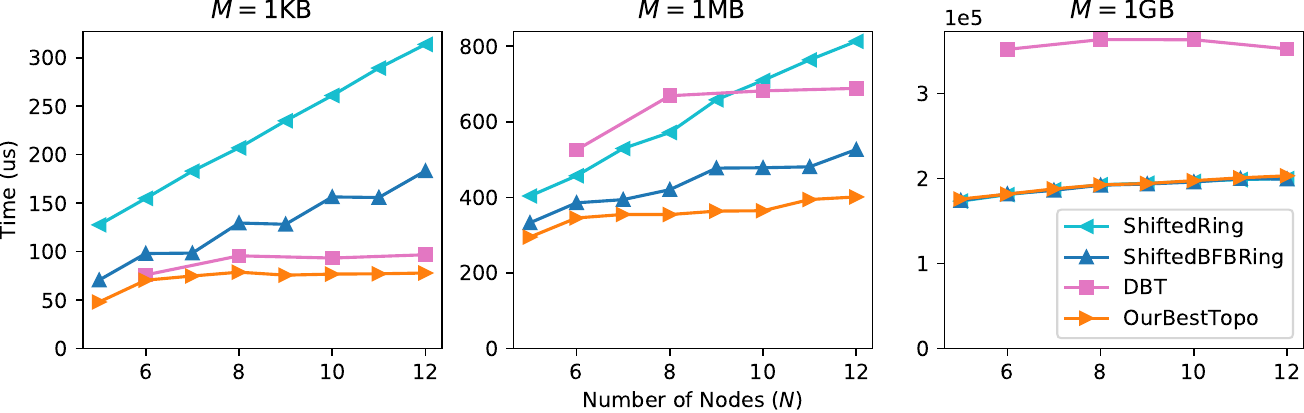}
        \caption{Allreduce experiment results on testbed at $M\!=\!1{\normalfont KB},1{\normalfont MB},1{\normalfont GB}$. \normalfont ``OurBestTopo'' topologies are listed in Table \ref{tab:ourtopologies}. Reduce-scatter and allgather results are in Figure~\ref{fig:rs_ag_results} of Appendix.}
        \label{fig:arexptres}
    \end{minipage}
\end{figure*}

\subsection{Direct-Connect Optical Testbed}
\label{sec:testbed-eval}
Our testbed consists of 12 servers, each with an NVIDIA A100-PCIe-40GB GPU~\cite{gpu} and a 100 Gbps HP NIC~\cite{nic}, configured as 4x25Gbps breakout interfaces~\cite{breakout}.
The NICs are directly connected via a Telescent optical patch panel~\cite{patchpanel}.
Our testbed can realize topologies by reconfiguring the patch panel.
We limit our evaluation to bidirectional topologies due to limitations in our testbed (discussed in Appendix~\ref{sec:unitobidirectional}).

\subsection{Experiment Setup}\label{sec:exptsetup}

\textbf{Baselines:} We evaluate against the two baselines at $d\!=\!4$: (1)~\emph{ShiftedRing}, which improves upon NCCL ring~\cite{nccl}, is used by TopoOpt~\cite{topo-opt} for data-parallel training. The topology is a superposition of two bidirectional rings, each allreducing half of the data. (2)~\emph{Double Binary Tree} (DBT), also implemented in NCCL, uses the topology and schedule from~\cite{dbtpaper,ncclDBT}.

\noindent\textbf{Methodology:} We use the MSCCL runtime~\cite{SCCL,gc3,msccl} to evaluate the topologies and schedules. We sweep through runtime parameters, such as the protocol (Simple or LL), number of channels (1, 2, 4, or 8), degrees of pipelining for the DBT baseline, etc., and choose the best-performing schedule for each data size. For DNN training, we pass our schedules to PyTorch with the MSCCL backend.

\subsection{Collective Communication Evaluation}\label{sec:arperf}

\begin{figure*}[htb]
    \begin{minipage}{0.33\textwidth}
        \centering
        \includegraphics[width=\textwidth]{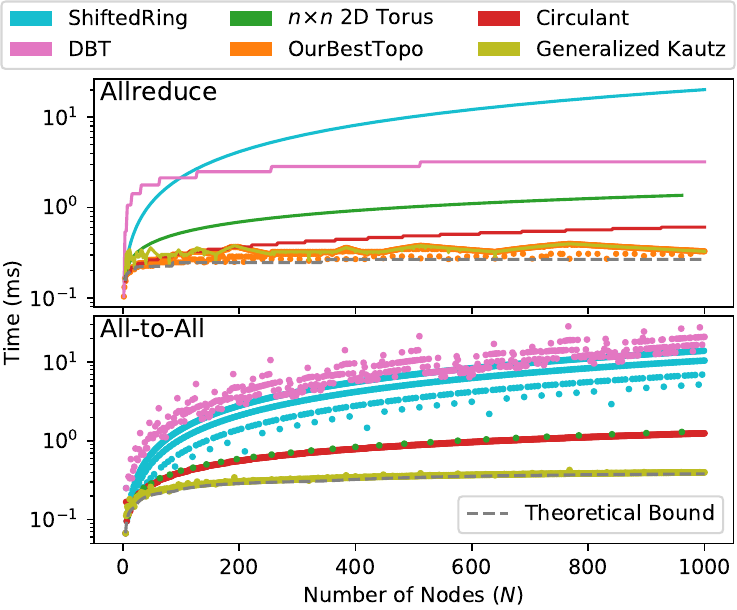}
        \caption{Comparing theoretical allreduce and all-to-all runtimes analytically at large $N$ for $d\!=\!4$, $\alpha\!=\!10\mu s$, and $M/B\!=\!1{\normalfont MB}/100{\normalfont Gbps}$. {\normalfont The all-to-all times are computed via multi-commodity flow (Appendix~\ref{app-sec:all-to-all}) with each node having 1MB of data to send (i.e., sending $1/N$ MB to each node).}}
        \label{fig:comparisonwithringdbt}
    \end{minipage}
    \nextfloat
    \hfill
    \begin{minipage}{0.65\textwidth}
        \begin{subfigure}{.63\textwidth}
            \centering
            \includegraphics[scale=0.5]{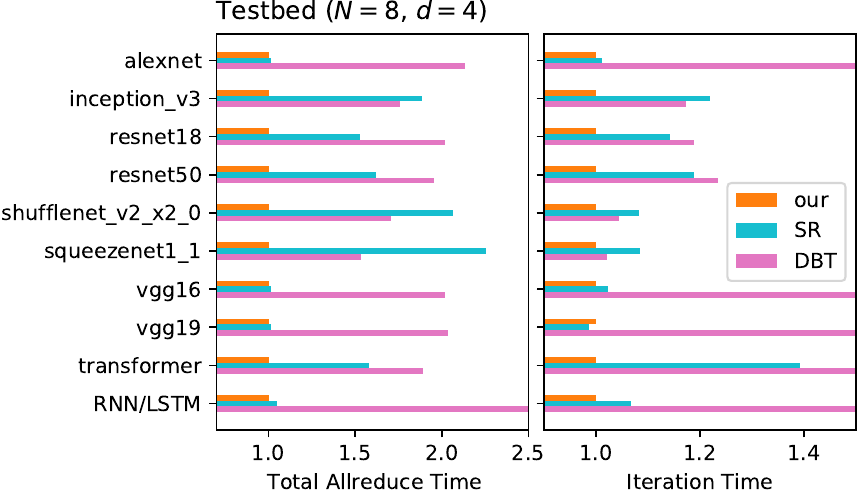}
            \caption{Small Model testbed training results (normalized by $K_{4,4}$)}
            \label{fig:smalltraining}
        \end{subfigure}
        \hfill
        \begin{subfigure}{.355\textwidth}
            \centering
            \includegraphics[scale=0.5]{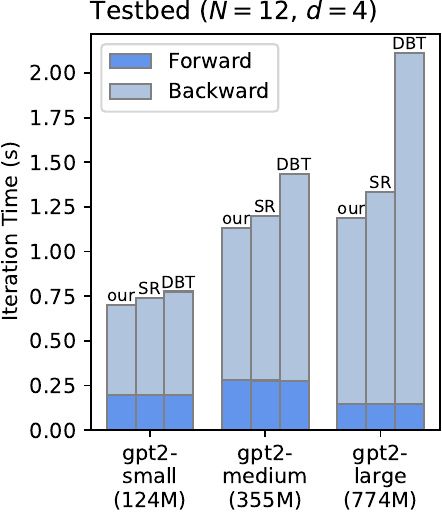}
            \caption{GPT2 testbed training results}
            \label{fig:gpt2training}
        \end{subfigure}
        \caption{Testbed data-parallel training results with different topologies. {\normalfont We compare our topologies (Table~\ref{tab:ourtopologies}) against ShiftedRings (SR) and double binary trees (DBT) at 8 and 12-node scale. (a) shows results of training small models on 8 A100 GPUs of our testbed, all using batch size 64. The total allreduce time is the sum of the allreduce times for all layers in the model. (b) shows results of training GPT-2 with 12 A100 GPUs. The per-GPU batch sizes are selected to max out the 40GB GPU memory with the small, medium, and large models having per-GPU batch sizes of 8, 4, and 1, respectively.}}
        \label{fig:testbedtraining}
    \end{minipage}
\end{figure*}

Figure~\ref{fig:arexptres} shows allreduce results for varying topology sizes $N$ and data sizes $M$. Table~\ref{tab:ourtopologies} shows the topologies generated by our topology finder (\S\ref{sec:topofinder}). We also add \emph{ShiftedBFBRing}, which is ShiftedRing topology but with our BFB generated schedule.
We observe that in the small data regime ($M\!=\!1$KB), our topology beats ShiftedRing by a significant margin ($\sim\!\! 75 \%$ at $N\!=\!12$) and also outperforms DBT  ($\sim\!\! 20 \%$ at $N\!=\!8,10,12$). Our ShiftedBFBRing beats ShiftedRing (\mbox{$\sim\!\! 40 \%$} at $N\!=\!12$) despite using the same topology. At small data sizes, the runtime is dominated by total-hop latency $T_L$, and hence, we can significantly outperform ShiftedRing, which has linear instead of logarithmic $T_L$ growth with respect to $N$.

In the large data regime ($M\!=\!1$GB), our topology beats DBT by a significant margin ($\sim\!\! 50 \%$ lower at $N\!=\!8,10,12$) and matches the performance of ShiftedRing. At large data sizes, the runtime is dominated by BW runtime $T_B$. Since the ShiftedRing is BW-optimal, we can only match its performance. Due to the influence of both total-hop latency and BW runtime at intermediate data sizes ($M\!=\!1$MB), our topology outperforms ShiftedRing ($\sim\!\! 50 \%$ at $N\!=\!12$) and DBT ($\sim\!\! 45 \%$ at $N\!=\!8,10,12$) in this regime. Our ShiftedBFBRing also outperforms ShiftedRing ($\sim\!\! 35 \%$ at $N\!=\!12$).
Note that although our gains over ShiftedRing diminish as $M$ grows, future increases in hardware bandwidth will enhance gains at large $M$ due to $T_L$ playing a more significant role. Appendix Figure~\ref{fig:rs_ag_results} shows the reduce-scatter and allgather results, which demonstrate trends and conclusions similar to those in Figure~\ref{fig:arexptres}.

Figure~\ref{fig:comparisonwithringdbt} shows the allreduce and all-to-all runtime comparison for large $N$ based on our analytical model. Topologies generated by our topology finder perform orders of magnitude faster in both allreduce and all-to-all. In allreduce, our best topologies outperform ShiftedRing and DBT by $56\times$ and $10\times$, respectively, near $N\!=\!1000$, due to the former's linear growth in $T_L$ and the latter's poor BW performance. When compared against 2D torus, our topologies also achieve $4\times$ better allreduce performance near $N\!=\!1000$ (see \S\ref{sec:paretoanalysis} for a detailed analysis of our topologies at large $N$ for different $\alpha,M/B$). As for all-to-all, we compare baseline topologies against our lowest-diameter topology, generalized Kautz, and our highest-diameter topology, circulant, from our Pareto-frontier for any $N$ and $d$. These two represent our best and worst all-to-all topologies, respectively, while also serving as the worst and best BW-efficient allreduce topologies. Nevertheless, circulant still outperforms all baselines in all-to-all: $9\times$ and $14\times$ better than ShiftdRing and DBT, respectively, on average from $N\!=\!900$ to $1000$. It is barely matched by the 2D torus, which is limited to the square number $N$s. Our lowest-diameter topology, generalized Kautz, outperforms ShiftedRing and DBT by $28\times$ and $42\times$, respectively, and is within 5.2\% of the theoretical bound from $N\!=\!900$ to $1000$.

\subsection{DNN Training Evaluation}
\label{sec:dnntrain}

\begin{figure}[htb]
    \centering
    \includegraphics[width=\columnwidth]{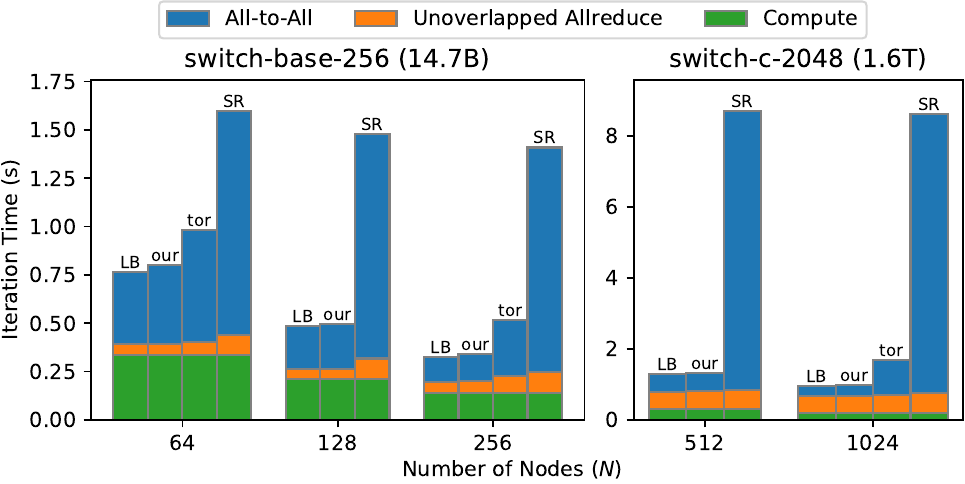}
    \caption{Simulated expert-parallel training of Switch Transformers across various topologies of different sizes. {\normalfont The simulation is conducted assuming $\alpha\!=\!10\mu$s, $B\!=\!100\text{Gbps}$. All-to-all time is computed via multi-commodity flow (Appendix~\ref{app-sec:all-to-all}). We detailed our topologies and the setup of simulation in Appendix~\ref{app-sec:simtrainingdetails}.}}
    \label{fig:simmoetraining}
\end{figure}

We compare our topologies against ShiftedRings and double binary trees in DNN training. On our small-scale testbed, we demonstrate improvements in data-parallel training across various small models and also \mbox{GPT-2}~\cite{gpt2}. For full-scale LLM training involving up to 1024 nodes and all-to-all communications, we simulate expert-parallel training of a Mixture of Experts (MoE) model to show our improvements at scale.

\noindent\textbf{Testbed Training:} We run PyTorch Distributed Data Parallel (DDP)~\cite{pytorch-dist} training experiments on our testbed. Figure~\ref{fig:testbedtraining} shows the results of training both small DNN models and \mbox{GPT-2}. We compare our topologies (from Table~\ref{tab:ourtopologies}) against ShiftedRings and DBT. In training small models (Figure~\ref{fig:smalltraining}), our topology improves total allreduce time by 30\% and 50\% on average against ShiftedRing and DBT, respectively. With optimizations such as compute-communication overlap, our topology still secures a 10\% and 25\% average improvement in iteration time over the baselines. In \mbox{GPT-2} training (Figure~\ref{fig:gpt2training}), despite the limited scale of our testbed, our topology enhances iteration time by 7\% and 25\% on average compared to ShiftedRing and DBT, respectively.

\begin{figure*}
    \begin{minipage}{0.625\textwidth}
        \centering
        \resizebox{\textwidth}{!}{
        \begin{tabular}{|c||c|c|c|c||c|c|c|c||c|c|c|c||c|}
            \hline
            \multirow{2}{*}{$N$} & \multicolumn{4}{c||}{SCCL} & \multicolumn{4}{c||}{TACCL w/o Symmetry} & \multicolumn{4}{c||}{TACCL w/ Symmetry} & \multirow{2}{*}{BFB} \\
            \cline{2-13}
            & $c\!=\!1$ & $c\!=\!2$ & $c\!=\!3$ & $c\!=\!4$ & $c\!=\!1$ & $c\!=\!2$ & $c\!=\!3$ & $c\!=\!4$ & $c\!=\!1$ & $c\!=\!2$ & $c\!=\!3$ & $c\!=\!4$ & \\
            \hline
            \multicolumn{14}{|c|}{Hypercube} \\
            \hline
            4 & 0.59 & 0.64 & 0.68 & 0.72 & 0.89 & 0.50 & 0.83 & 0.75 & 0.62 & 0.51 & 0.71 & 0.60 & $<\!\!0.01$ \\
            8 & 0.86 & 1.22 & 1.86 & 2.48 & 96.9 & 807 & 63.2 & 1800 & 7.97 & 645 & 7.39 & 1801 & $<\!\!0.01$ \\
            16 & 21.4 & 48.4 & 130 & 573 & 1801 & 1801 & 1801 & 1802 & 1801 & n/a & n/a & n/a & $<\!\!0.01$ \\
            32 & $>\!\!10^4$ & $>\!\!10^4$ & $>\!\!10^4$ & $>\!\!10^4$ & 1802 & n/a & n/a & n/a & n/a & n/a & n/a & n/a & 0.03 \\
            64 & $>\!\!10^4$ & $>\!\!10^4$ & $>\!\!10^4$ & $>\!\!10^4$ & n/a & n/a & n/a & n/a & n/a & n/a & n/a & n/a & 0.17 \\
            1024 & $>\!\!10^4$ & $>\!\!10^4$ & $>\!\!10^4$ & $>\!\!10^4$ & n/a & n/a & n/a & n/a & n/a & n/a & n/a & n/a & 52.7 \\
            \hline
            \multicolumn{14}{|c|}{2D Torus ($n\times n$)} \\
            \hline
            4 & 0.61 & 0.63 & 0.67 & 0.76 & 0.68 & 0.50 & 0.82 & 0.72 & 0.45 & 0.51 & 0.76 & 0.64 & $<\!\!0.01$ \\
            9 & 1.00 & 1.51 & 2.22 & 3.44 & 1801 & 189 & 67.8 & 262 & 88.6 & 71.1 & 67.8 & 105 & $<\!\!0.01$ \\
            16 & 17.5 & 60 & 131 & 603 & 1801 & 1801 & 1801 & 1802 & 1801 & 1801 & 1801 & n/a & $<\!\!0.01$ \\
            25 & 3286 & 5641 & $>\!\!10^4$ & $>\!\!10^4$ & 1802 & 1802 & 1803 & n/a & 1802 & n/a & n/a & n/a & 0.01 \\
            36 & $>\!\!10^4$ & $>\!\!10^4$ & $>\!\!10^4$ & $>\!\!10^4$ & n/a & n/a & n/a & n/a & n/a & n/a & n/a & n/a & 0.03 \\
            2500 & $>\!\!10^4$ & $>\!\!10^4$ & $>\!\!10^4$ & $>\!\!10^4$ & n/a & n/a & n/a & n/a & n/a & n/a & n/a & n/a & 61.1 \\
            \hline
        \end{tabular}
        }
        \captionof{table}{Comparing allgather schedule generation runtimes (in seconds) of SCCL, TACCL, and BFB. {\normalfont The setup of SCCL is to generate schedules with least number of comm steps. Both SCCL and TACCL were run with chunks=1,2,3,4 (number of chunks per shard), and TACCL was run w/ and w/o manually set topology symmetry. ``n/a'' indicates where TACCL reports an error due to failure to generate a solution within its 1800s time limit for MILP solver.}}
        \label{tab:sccltacclbfs}
    \end{minipage}
    \hfill
    \begin{minipage}{0.355\textwidth}
        \includegraphics[width=\textwidth]{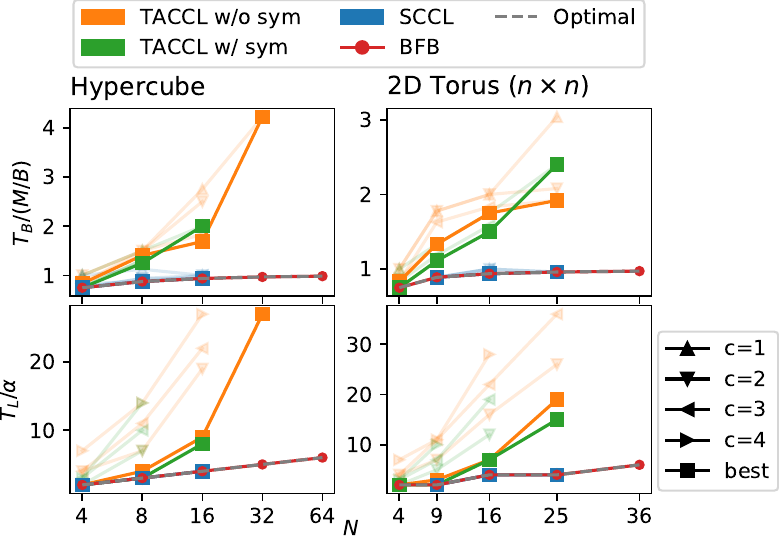}
        \caption{Comparing theoretical performances of schedules from Table \ref{tab:sccltacclbfs}. {\normalfont We show both $T_L$ and $T_B$ of the schedules, along with the theoretical optimal. For SCCL and TACCL, the solid lines show the best results from parameter sweeps. The inferior ones are dimmed.}}
        \label{fig:sccltacclbfs}
    \end{minipage}
\end{figure*}

%% file: analysis.tex
\noindent\textbf{Large-Scale Simulation:} While improvements over ShiftedRing and DBT have been shown in testbed training, full-scale LLM training is performed on much larger clusters. In Figure~\ref{fig:simmoetraining}, we simulate expert-parallel training of Switch Transformers~\cite{switchtran} on a much larger scale with parameter sizes up to 1.6 trillion. We collect execution timestamps from one A100 GPU to derive the compute times for each layer. Communication times are then added to simulate training, ensuring compute-communication overlap/dependency. Appendix~\S\ref{app-sec:simtrainingdetails} provides further details of the simulation.

Expert-parallel training involves not only data-parallel allreduce for non-expert laryers, but also all-to-all communications to transfer tokens to and from the routed experts, which are in the critical path of compute~\cite{switchtran,deepspeedmoe,gshard,lina}. In Figure~\ref{fig:simmoetraining}, we break down the iteration time into compute time, unoverlapped allreduce time, and all-to-all time for better understanding of performance. As previously analyzed in Figure~\ref{fig:comparisonwithringdbt}, ShiftedRing exhibit all-to-all performance that is order-of-magnitude worse than our topologies. At 256-node training of 14.7B MoE model, ShiftedRing has $8\!\times$ greater total all-to-all time, resulting in $4\!\times$ longer iteration time compared to our topology. We also include 2D torus for comparison due to its relatively better all-to-all performance. However, it still has all-to-all and iteration times that are $2\!\times$ and $1.5\!\times$ greater, respectively, than our topology. The disparity is even larger at 1024-node training of 1.6T MoE model, where ShiftedRing and 2D torus show all-to-all times that are $27\!\times$ and $3.3\!\times$ greater, and iteration times that are $9\!\times$ and $1.7\!\times$ longer, respectively. At this scale, ShiftedRing and 2D torus spend 91\% and 58\% of iteration time on all-to-all communications, while our topology only spends 30\%. We omit DBT in Figure~\ref{fig:simmoetraining} due to its significantly worse performance ($\sim\!2\!\times$ of ShiftedRing). Due to high performance in both allreduce and all-to-all, our topologies consistently remain within 5\% of the theoretical lower bound for iteration time. 

Because large models involve large data-parallel allreduce sizes and both torus and ShiftedRing are BW-optimal, the allreduce performances are similar across these topologies. For a broader spectrum of all-to-all efficient low-hop topologies like expander graphs, due to the lack of efficient allreduce schedules prior to our work, they are unable to be utilized in allreduce-involved training. 

\subsection{BFB Schedule Evaluation}\label{sec:schedeval}

We evaluate schedule generation from two aspects: schedule generation runtime and the performance of generated schedules. In \S\ref{sec:scalcomp}, we compare BFB with state-of-the-art schedule generations: SCCL~\cite{SCCL} and TACCL~\cite{shah2021synthesizing}, in both generation runtime and theoretical schedule performance. In \S\ref{sec:supercomp}, we compare on supercomputing torus clusters (with up to 54 nodes) the performance of torus schedules generated by BFB, traditional torus scheduling~\cite{multiporttorus}, SCCL, and TACCL.

\subsubsection{Schedule Generation}\label{sec:scalcomp}

In schedule generation, SCCL and TACCL are the closest in spirit to BFB schedule generation. Table \ref{tab:sccltacclbfs} shows the runtime comparison between SCCL, TACCL, and BFB when generating allgather schedules for hypercube and 2D torus. Both SCCL and TACCL use NP-hard optimization to generate schedules. SCCL, which uses an SMT solver, fails to generate a schedule within $10^4$ seconds beyond $N\!=\!30$. TACCL formulates the scheduling problem as a mixed integer linear program (MILP). It sets an 1800s time limit for its MILP solver, after which it will return the best solution found up to that point. However, for larger topologies, TACCL's solver fails to find a solution within the time limit, resulting in an error. In comparison, BFB schedule generation is faster by orders of magnitude due to its polynomial-time generation.

In terms of theoretical schedule performance, Figure~\ref{fig:sccltacclbfs} compares the total-hop latency and BW runtime of generated schedules. Given a topology, SCCL and TACCL need to perform a sweep across parameters such as the number of chunks and symmetry. They have to generate schedules for different parameter sets to identify the high-performance ones, unlike BFB, which has no parameter. In Figure~\ref{fig:sccltacclbfs}, the schedules of SCCL and BFB can both achieve exact optimality, but TACCL's have significantly worse performance, especially at large $N$s. SCCL is uniquely capable of generating all Pareto-efficient schedules. However, due to the prohibitive runtime of parameter sweep, SCCL can only do so for very small $N$s.

\begin{figure}[tb]
    \centering
    \vspace{0.1in}
    \includegraphics[width=\columnwidth]{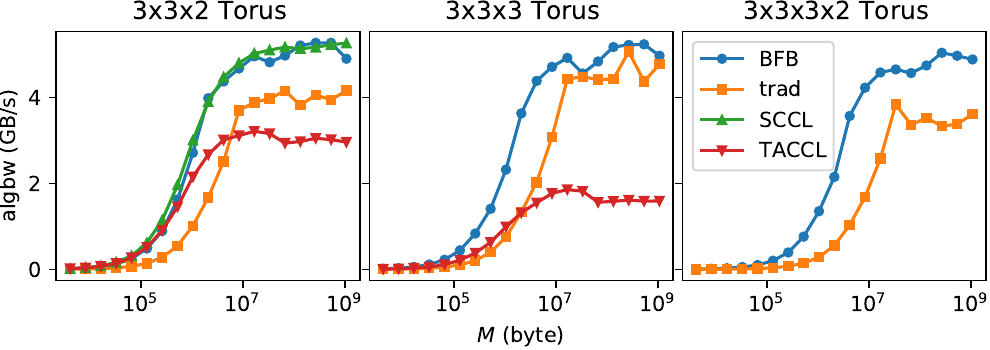}
    \caption{Comparing allreduce performances of torus schedules generated by BFB, traditional torus scheduling~\cite{multiporttorus}, SCCL, and TACCL on Frontera~\cite{frontera} supercomputer. {\normalfont The $y$-axis is algorithmic bandwidth (algbw), computed as $M$ divided by end-to-end runtime. SCCL fails to generate a schedule for $3\!\!\times\!\!3\!\!\times\!\!3$ and $3\!\!\times\!\!3\!\!\times\!\!3\!\!\times\!\!2$ tori, and TACCL fails to generate a schedule for $3\!\!\times\!\!3\!\!\times\!\!3\!\!\times\!\!2$ torus within the time limits.}}
    \label{fig:supercomp}
\end{figure}

\subsubsection{Supercomputing Allreduce Experiments}\label{sec:supercomp}

In the supercomputing setting, we run torus schedules generated by BFB, traditional torus scheduling~\cite{multiporttorus}, SCCL, and TACCL on Frontera~\cite{frontera} supercomputer at the Texas Advanced Computing Center (TACC)~\cite{tacc}. The cluster consists of 396 nodes in a 6D torus direct-connect topology. Each node is equipped with an Intel Xeon Platinum 8280 CPU and a Rockport NC1225 network card, capable of delivering 25 Gbps per link, with degree 12. At this degree, however, the total BW of a single node is bottlenecked by the 100 Gbps host BW of PCIe Gen3 x16. Finally, the schedules are lowered and run using Intel oneCCL~\cite{oneccl} + libfabric~\cite{libfabric}.

We run schedules on two types of sub-torus within the cluster: equal-dimension ($3\!\!\times\!\!3\!\!\times\!\!3$) and unequal-dimension ($3\!\!\times\!\!3\!\!\times\!\!2$ \& $3\!\!\times\!\!3\!\!\times\!\!3\!\!\times\!\!2$). As shown in Figure~\ref{fig:supercomp}, BFB schedules achieve the highest performance in all settings. As mentioned in \S\ref{sec:generateivetopo}, the traditional torus schedule can only achieve high BW performance in tori with equal dimensions. At large $M$, it matches BFB's performance in $3\!\!\times\!\!3\!\!\times\!\!3$ torus but significantly underperforms in $3\!\!\times\!\!3\!\!\times\!\!2$ and $3\!\!\times\!\!3\!\!\times\!\!3\!\!\times\!\!2$, where BFB has 29\% and 42\% higher algbw on average for $M\!\geq\!100$MB. At small to intermediate $M$ ($<\!100$MB), BFB outperforms traditional schedules by 3.1$\times$ on average in all settings due to its 2$\times$ lower in total-hop latency and higher BW performance.

As for SCCL and TACCL, we adhere to the same time limits and parameter sweeps as in \S\ref{sec:scalcomp} and select the best result at each $M$ from all parameter sets. In $3\!\!\times\!\!3\!\!\times\!\!2$ torus, SCCL is able to generate an optimal schedule, matching BFB's performance across all $M$. However, it fails to generate a schedule within $10^4$ seconds for other larger tori. TACCL can only generate schedules in $3\!\!\times\!\!3\!\!\times\!\!2$ and $3\!\!\times\!\!3\!\!\times\!\!3$, and its schedules underperform BFB's by a large margin. One additional observation is that the algbw of BFB at large $M$ hardly changes from 18-node ($3\!\!\times\!\!3\!\!\times\!\!2$) to 54-node ($3\!\!\times\!\!3\!\!\times\!\!3\!\!\times\!\!2$) torus. This can be explained by the fact that BFB schedules have mathematically achieved allreduce BW optimality ($\frac{2M}{B}\!\cdot\!\frac{N-1}{N}$), which remains nearly constant as $N$ increases.

%% file: conclusion.tex
\section{Concluding Remarks}
\label{sec:conclusion}

Collectives are critical to both ML training and HPC workloads. Current solutions often rely solely on existing topologies and schedules, resulting in high total-hop latency, bandwidth inefficiency, or low all-to-all throughput. We presented a general, highly scalable, and automated algorithmic framework for optimizing topology and schedule generation for collectives by leveraging scalable graph-theoretic approaches. Our evaluation demonstrates significant performance benefits using multiple testbeds and large-scale simulations for both collectives and training jobs.

%% file: appendix.tex
\tolerance=3000 %

\clearpage

\noindent {\Large \textbf{Appendix}}

\vspace{0.1in}
\noindent In this appendix, we give formal mathematical definitions and analysis of various techniques and concepts mentioned in the main text. To summarize,
\squishlist
    \item \S\ref{app-sec:evalapp} provides supplementary materials to evaluation section.
    \item \S\ref{sec:reducescatterallgather} gives formal definitions of reduce-scatter/allgather schedule and how one can be transformed into another.
    \item \S\ref{app-sec:optimality} gives formal definitions of total-hop latency and bandwidth optimality, along with discussions on optimal allreduce schedule and computational cost of reduction.
    \item \S\ref{app-sec:expansionperf} provides formal definitions of expansion techniques and optimality analysis of their expanded schedules.
    \item \S\ref{app:shortestpath} provides optimality analysis of BFB schedule generation and discusses variant formulations that support generating schedules for a fixed number of chunks and for heterogeneous network topology.
    \item \S\ref{sec:genetopoappen} discusses various generative topologies and the performance of their generated BFB schedules.
    \item \S\ref{sec:proof} provides proofs of all theorems in this paper.
    \item \S\ref{sec:suppfigtab} contains supplementary tables and figures. In particular, Table~\ref{table:toposummary} gives a summary of topologies in this paper.
\squishend

\input{eval_sup}

\section{Reduce-Scatter \& Allgather}\label{sec:reducescatterallgather}

We use tuple $((v,C),(u,w),t)$ to denote that $u$ sends $v$'s chunk $C$ to $w$ at comm step $t$. Node $v$ is the source and destination node of chunk $C$ in allgather and reduce-scatter respectively. A communication schedule is thus a collection of tuples.
\begin{restatable}[Allgather]{definition}{defallgather}\label{def:allgather}
An algorithm $(G,A)$ is an allgather algorithm if for arbitrary $x\in S$ and distinct $u,v\in V_G$, there exists a sequence in $A$:
\begin{multline*}
	((v,C_1),(w_0,w_1),t_1),((v,C_2),(w_1,w_2),t_2),\dots\\
	((v,C_n),(w_{n-1},w_n),t_n),
\end{multline*}
where $w_0\!=\!v$, $w_n\!=\!u$, $t_1\!<\!t_2\!<\!\dots\!<\!t_n$, and $x\!\in\! C_1\!\cap C_2\!\cap\dots\cap C_n$.
\end{restatable}
This sequence serves to broadcast $x$ from $v$ to $u$. A reduce-scatter algorithm has the same definition except $w_0\!=\!u$, $w_n\!=\!v$. In reduce-scatter, we assume any chunk received by a node is immediately reduced with the node's local chunk.

In this paper, many of the techniques are discussed under allgather only. We will show that anything holds in either reduce-scatter or allgather has an equivalent version for the other collective operation. To do so, we use the concept of \emph{transpose graph} from graph theory and define \emph{reverse schedule}. We say a schedule $A$ is \emph{for} topology $G$ if every $((v,C),(u,w),t)\in A$ satisfies $u,v,w\in V_G$ and $(u,w)\in E_G$.
\begin{restatable}[Reverse Schedule]{definition}{defreverseschedule}
    Suppose $A$ is a schedule for $G$. A reverse schedule $A^T$ of $A$ is a schedule for transpose graph $G^T$ such that $((v,C),(u,w),t_{\max}-t+1)\in A^T$ iff $((v,C),(w,u),t)\in A$, where $t_{\max}$ is the max comm step in $A$.
\end{restatable}
It is trivial to see that $T_L(A)=T_L(A^T)$ and $T_B(A)=T_B(A^T)$. Note that $(u,w)\in E_{G^T}$ if and only if $(w,u)\in E_{G}$ by definition of transpose graph.
\begin{restatable}{theorem}{thmscheduleswitch}\label{thm:scheduleswitch}
    If $A$ is a reduce-scatter/allgather schedule for $G$, then $A^T$ is an allgather/reduce-scatter schedule for $G^T$.
\end{restatable}
Theorem~\ref{thm:scheduleswitch} has the following two corollaries:
\begin{restatable}{corollary}{thmscheduleswitchfunc}\label{thm:scheduleswitchfunc}
    Suppose $G\mapsto f(G)$ is a function to construct reduce-scatter/allgather schedule given graph $G$, then $G\mapsto f(G^T)^T$ is a function to construct allgather/reduce-scatter schedule given graph $G$.
\end{restatable}
\begin{restatable}{corollary}{thmscheduleswitchmap}\label{thm:scheduleswitchmap}
    Suppose $(G,A)\mapsto (f(G),f(A))$ is a mapping within reduce-scatter/allgather algorithms, then $(G,A)\mapsto (f(G^T)^T,f(A^T)^T)$ is a mapping within allgather/reduce-scatter algorithms.
\end{restatable}
For example, the line graph expansion in \S\ref{sec:linegraph} can be seen as a mapping within allgather, and the BFB linear program~(\ref{lpmodel}) can be seen as a function to construct allgather schedule. Thus, Corollary~\ref{thm:scheduleswitchfunc} and \ref{thm:scheduleswitchmap} have shown that they both have equivalent versions in reduce-scatter.

In undirected topology, it is well-known that reduce-scatter and allgather are a pair of dual operations such that one can be transformed into another by reversing the communication in schedule~\cite{conc_comp}. It is similar for directed topology but with extra requirement and more complicated transformation. We define the following property for directed graphs:
\begin{definition}[Reverse-Symmetry]
    A digraph $G$ is reverse-symmetric if it is isomorphic to its own transpose graph $G^T$.
\end{definition}
In graph theory, there is a similar concept called \emph{skew-symmetric graph}. Reverse-symmetry is a weaker condition than skew-symmetry.

We define a way to transform the schedule for $G$ into a schedule for $G^T$ based on graph isomorphism:
\begin{restatable}[Schedule Isomorphism]{definition}{defscheduleisomorphism}\label{def:scheduleisomorphism}
    Suppose $G$ and $G'$ are isomorphic. Let $f:V_G\to V_{G'}$ be the graph isomorphism and $A$ be a schedule for $G$, then $f(A)$ is a schedule for $G'$ that $((f(v),C),(f(u),f(w)),t)\!\in\! f(A)$ iff $((v,C),(u,w),t)\!\in\! A$.
\end{restatable}
\begin{restatable}{theorem}{thmrescallather}\label{thm:rescallather}
Suppose $G$ is reverse-symmetric. Let $G^T$ be the transpose graph, and let $f:V_{G^T}\to V_G$ be the isomorphism from $G^T$ to $G$. If $(G,A)$ is a reduce-scatter/allgather algorithm, then $(G,f(A^T))$ is an allgather/reduce-scatter algorithm with $T_L(f(A^T))=T_L(A)$ and $T_B(f(A^T))=T_B(A)$.
\end{restatable}
Theorem~\ref{thm:rescallather} establishes that given any reverse-symmetric topology, if we have either reduce-scatter or allgather, then we can construct both reduce-scatter and allgather. Since allreduce can be achieved by applying a reduce-scatter followed by an allgather, we only need one of reduce-scatter and allgather to construct a complete allreduce algorithm. Furthermore, if the reduce-scatter or allgather algorithm has runtime $T$, then the resulting allreduce algorithm has runtime $2T$.

Most of our base topologies are reverse-symmetric (Table~\ref{table:toposummary}). In addition, all of our expansion techniques also preserve reverse-symmetry. Thus, one can almost always use Theorem~\ref{thm:rescallather} to derive reduce-scatter and allreduce schedules from allgather schedule on our synthesized topologies. For non-reverse-symmetric topologies like generalized Kautz graph, one can apply Corollary~\ref{thm:scheduleswitchfunc} or \ref{thm:scheduleswitchmap} to construct reduce-scatter and allgather separately.

\section{Topology-Schedule Optimality}
\label{app-sec:optimality}

Because our cost model is only concerned with total-hop latency and BW runtime, the optimality of reduce-scatter/allgather algorithm is only related to total-hop latency optimality and BW optimality in this paper. Note that we also consider topology as a dimension that can be optimized, so optimality is discussed in the space of all topology-schedule combinations, i.e., \textit{algorithms} by our definition.

\subsection{Total-Hop Latency Optimality}\label{sec:nodelatencyoptimality}

\begin{definition}[Total-Hop Latency Optimal]
    Given an $N$-node degree-$d$ reduce-scatter/allgather algorithm $(G,A)$, if any other $N$-node degree-$d$ reduce-scatter/allgather algorithm $(G',A')$ satisfies $T_L(A')\!\geq\! T_L(A)$, then $(G,A)$ is total-hop latency optimal.
\end{definition}
Because in reduce-scatter/allgather, every node needs to send a shard of data to every other node, the number of comm steps is lower bounded by the graph diameter:
\begin{restatable}{theorem}{thmdiameter}\label{thm:diameter}
    Every reduce-scatter/allgather algorithm $(G,A)$ satisfies $T_L(A)\!\geq\!\alpha\cdot D(G)$, where $D(G)$ is the diameter of $G$.
\end{restatable}
Because we can always construct a BFB schedule $A$ for topology $G$ with $T_L(A)=\alpha\cdot D(G)$, it follows the corollary:
\begin{restatable}{corollary}{thmlatencyoptimaldiameter}
    An $N$-node degree-$d$ reduce-scatter/allgather algorithm $(G,A)$ is total-hop latency optimal if and only if
    \[
        T_L(A)=\alpha\cdot D(G)=\alpha\cdot\min\{D(G'):|V_{G'}|=N,\deg(G')=d\}.
    \]
\end{restatable}
The minimum diameter of a directed graph given a number of vertices and degree is still an open question. One can check \textit{degree/diameter problem}~\cite{moore_bound} for more information. However, as a close upper bound of number of vertices given degree and diameter, the \emph{Moore bound} for digraph is sufficient to tell the total-hop latency optimality in most cases.
\begin{definition}[Moore Bound]
    Let $G$ be any degree-$d$ digraph of diameter $k$. The Moore bound is an upper bound on the number of vertices in $G$:
    \[
        M_{d,k}=\sum_{i=0}^k d^i=\frac{d^{k+1}-1}{d-1}.
    \]
\end{definition}
\begin{definition}[Moore Optimal]\label{def:mooreopt}
     Let $(G,A)$ be an $N$-node degree-$d$ reduce-scatter/allgather algorithm with $T_L(A)=k\alpha$, then $(G,A)$ is Moore optimal if $N>M_{d,k-1}$.
\end{definition}
Because for any degree-$d$ digraph $G$, $D(G)\geq k$ must be true as long as $|V_G|>M_{d,k-1}$, Moore optimality is a stronger condition than total-hop latency optimality. We define a function $T^*_L$ such that $T^*_L(N,d)$ is equal to the Moore optimal total-hop latency of $N$-node degree-$d$ reduce-scatter/allgather algorithms.

\subsection{Bandwidth Optimality}

\begin{definition}[Bandwidth Optimal]
    Given an $N$-node degree-$d$ reduce-scatter/allgather algorithm $(G,A)$, if any other $N$-node degree-$d$ reduce-scatter/allgather algorithm $(G',A')$ satisfies $T_B(A')\geq T_B(A)$, then $(G,A)$ is BW-optimal.
\end{definition}
In reduce-scatter/allgather, each node needs to send/receive at least $M\cdot\frac{N-1}{N}$ amount of data. Thus, the following holds:
\begin{restatable}{theorem}{thmbandwidthlowerbound}
    $\frac{M}{B}\cdot\frac{N-1}{N}$ is a lower bound of $T_B(A)$ for any $N$-node reduce-scatter/allgather algorithm $(G,A)$.
\end{restatable}
Note that one can always construct a ring of degree $d$ by sending $d$ parallel edges from one node to the next node. The trivial ring reduce-scatter/allgather schedule has $\frac{M}{B}\cdot\frac{N-1}{N}$ BW runtime. Therefore, we have:
\begin{restatable}{corollary}{thmbwoptequation}\label{thm:bwoptequation}
    An $N$-node reduce-scatter/allgather algorithm $(G,A)$ is BW-optimal if and only if $T_B(A)=\frac{M}{B}\cdot\frac{N-1}{N}$.
\end{restatable}
We define a function $T^*_B$ such that $T^*_B(N)=\frac{M}{B}\cdot\frac{N-1}{N}$ is the optimal BW runtime of $N$-node reduce-scatter/allgather algorithms. From Corollary~\ref{thm:bwoptequation}, we have the following necessary and sufficient condition for BW optimality:
\begin{restatable}{theorem}{thmbwopt}\label{thm:bwopt}
    An allgather algorithm $(G,A)$ is BW-optimal if and only if:
    \begin{enumerate}
        \item\label{thm:bwopt1} $\frac{1}{B/d}\!\sum_{((v,C),(u,w))\in A_t}\!|C|\!=\!T_B(A_t)$ for all $(u,v)\!\in\! E_G$ and $t\!\in\!\{1,\dots,t_{\max}\}$. $A_t$ is the subschedule of $A$ at comm step $t$.
        \item\label{thm:bwopt2} Pick any distinct $u,v\in V_G$. For each $x\in S$, there exists a unique $((v,C),(w,u),t)\in A$ such that $x\in C$.
    \end{enumerate}
\end{restatable}
Condition~\ref{thm:bwopt1} ensures that at each comm step, every link of topology $G$ has equal workload, so no link finishes early and results in waste of bandwidth. Condition~\ref{thm:bwopt2} ensures that no piece of data is received twice by some node, so no duplicated send exists.

\subsection{Allreduce Optimality}\label{sec:allreduceopt}

In this paper, we construct an allreduce algorithm through a reduce-scatter followed by allgather. In such construction, the lower bound of allreduce algorithm is $2(T^*_L(N,d)+T^*_B(N))$. To compare this with the lower bound of any allreduce construction, in \cite{patarasuk2009bandwidth}, the authors have proved that $2T^*_B(N)$ is indeed the lower bound of BW runtime of any allreduce algorithm. As for total-hop latency, a reduce-scatter followed by allgather has at least $2D(G)$ number of comm steps, so $2T^*_L(N,d)$ is the lower bound of total-hop latency. Although one can use all-to-all to construct an allreduce with number of comm steps equal to one diameter $D(G)$ (lower bound being $T^*_L(N,d)$ instead of $2T^*_L(N,d)$),
the lower bound of BW runtime for all-to-all is $\frac{M}{B}\cdot(N-1)=N\cdot T^*_B(N)$, which is much worse than $2T^*_B(N)$.

There is also another way of constructing allreduce: reduce followed by broadcast. In such an approach, the number of comm steps can be twice the radius of $G$ instead of twice the diameter. However, the Moore bound for graph diameter also applies to graph radius, so $2T^*_L(N,d)$ is still a lower bound of allreduce via reduce+broadcast. By Theorem~\ref{thm:genkautzlatency}, the total-hop latency optimal allreduce via reduce+broadcast is at most $2\alpha$ lower than the total-hop latency of generalized Kautz graph can do with reduce-scatter plus allgather. Furthermore, reduce+broadcast is usually poor in BW performance.

\subsection{Computational Cost}\label{sec:computeopt}

In this paper, we omit the computational cost of reduction operation in performance analysis. While this approach is commonly adopted in previous literature \cite{wang2020blink, SCCL, shah2021synthesizing, dbtpaper}, we give a formal reasoning why this approach is legitimate. It is not only because computational cost is generally orders of magnitude lower than network cost, but also because computational cost can be incorporated into network cost.

Assume a cost model where computation and network communication do not overlap at each node.\footnote{Otherwise, the computational cost would be even more negligible.} In particular, at each comm step of reduce-scatter, the computation to reduce chunks happens immediately after the node receives all chunks and before the node starts to send out chunks for the next comm step. We adopt notations from \cite{conc_comp}, where $\gamma$ denotes the computational time cost per size of data. Like total-hop latency and BW runtime, we also let $T_C(A)$ be the total time spent on computation by schedule $A$. As argued in \cite{conc_comp}, a lower bound of computational cost is $T_C\geq M\cdot\gamma\cdot\frac{N-1}{N}$ for both reduce-scatter and allreduce, which is identical to the BW optimality of reduce-scatter and half of that of allreduce. The following theorem shows that BW runtime of a schedule can act as an upper bound for the computational time.
\begin{restatable}{theorem}{thmcomputebound}\label{thm:computebound}
    Given a reduce-scatter algorithm $(G,A)$, suppose $T_B(A)=\frac{M}{B}\cdot y$, then $T_C(A)\leq M\cdot\gamma\cdot y$.
\end{restatable}
The rationale behind Theorem~\ref{thm:computebound} is that the amount of computation for any node at a given comm step equals the amount of data the node receives during that comm step. Thus, \textbf{as we balance network transmission, it naturally leads to a more balanced computation.} With Theorem~\ref{thm:computebound}, if the BW runtime of some allreduce schedule $A$ is $T_B(A)=2\frac{M}{B}\cdot y$, then $T_B(A)+T_C(A)\leq M\cdot(\frac{2}{B}+\gamma)\cdot y$. We can thus simply define $B'=(\frac{1}{B}+\frac{\gamma}{2})^{-1}$, and then $2\frac{M}{B'}\cdot y$ can represent the sum of BW runtime and computational runtime altogether. The value of $y$ is all that matters. The following corollary shows that if an algorithm is BW-optimal, then such representation is exact.
\begin{restatable}{corollary}{thmbwopttocomputeopt}
    If allreduce algorithm $(G,A)$ is BW-optimal, i.e., $T_B(A)\!=\!2\frac{M}{B}\!\cdot\!\frac{N-1}{N}$, then $T_C(A)\!=\!M\!\cdot\!\gamma\!\cdot\!\frac{N-1}{N}$ and $T_B(A)+T_C(A)\!=\!2M\!\cdot\!(\frac{1}{B}+\frac{\gamma}{2})\!\cdot\!\frac{N-1}{N}$.
\end{restatable}
When profiling a testbed, one can simply derive the value of $\frac{1}{B}+\frac{\gamma}{2}$ using BW-optimal topologies and use it as the new $1/B$ to apply the results of this paper. While it is still possible for two schedules with the same BW runtime to have different computational runtimes, such difference is bounded by the aforementioned theorems and orders of magnitude smaller than BW runtime.

\section{Optimality of Expansion Techniques}
\label{app-sec:expansionperf}

In this section, we provide formal definitions and detailed performance analysis of expansion techniques.

\subsection{Line Graph Expansion}
\label{app-subsec:linegraph}

\begin{definition}[Line Graph]\label{linegraphdef}
    Given a directed graph (or digraph) $G$, each edge $(u,v)\in E_G$ corresponds to a vertex $uv$ in the line graph $L(G)$. For every $uv,vw$ pair in $V_{L(G)}$, there exists an edge $(uv,vw)\in E_{L(G)}$.
\end{definition}
\linescheduledef*

The following theorem gives the performance of the expanded schedule:
\begin{restatable}{theorem}{thmlinegraphperformance}\label{thm:linegraphperformance}
    Given a $d$-regular topology $G$, if $(G,A_G)$ is an $N$-node allgather algorithm, then $(L(G),A_{L(G)})$ is a $dN$-node allgather algorithm satisfying:
    \begin{gather}
        T_L(A_{L(G)})=T_L(A_G)+\alpha,\label{eq:linelatency}\\
        T_B(A_{L(G)})\leq T_B(A_G)+\frac{M}{B}\cdot\frac{1}{N}.\label{eq:linebw}
    \end{gather}
\end{restatable}
From Theorem~\ref{thm:linegraphperformance}, one can see that the performance of the expanded schedule depends on that of the base schedule.
Theorem~\ref{thm:linegraphperformance} makes an implicit assumption that $T_B(A_G,M,B)\!=\!\tau(M/B)$ for some constant $\tau$. This assumption, suggesting that $T_B$ scales linearly with data size and inversely with bandwidth, should hold for any reasonably designed schedule.

Consequently, if we apply line graph expansion $n$ times, the performance of the expanded schedule is:
\begin{restatable}{corollary}{thmlinegraphconclusion}\label{thm:linegraphconclusion}
    Given a $d$-regular topology $G$, if $(G,A_G)$ is an $N$-node allgather algorithm with $T_B(A_G,M,B)\!=\!\tau(M/B)$ for some constant $\tau$, then $(L^n(G),A_{L^n(G)})$ is a $d^nN$-node allgather algorithm satisfying:
    \begin{gather}
        T_L(A_{L^n(G)})=T_L(A_G)+n\alpha,\\
        T_B(A_{L^n(G)})\leq T_B(A_G)+\frac{M}{B}\cdot\frac{d}{d-1}\left(\frac{1}{N}-\frac{1}{d^nN}\right).\label{eq:linenbw}
    \end{gather}
\end{restatable}

Speaking of optimality:
\begin{restatable}{theorem}{thmlinemooreopt}\label{thm:linemooreopt}
    $(L^n(G),A_{L^n(G)})$ is Moore optimal if and only if $(G,A_G)$ is Moore optimal.
\end{restatable}

\begin{restatable}{theorem}{thmlinebwopt}\label{thm:linebwopt}
    If $(G,A_G)$ is BW-optimal with $N$ nodes, then $T_B(A_{L^n(G)})/T_B^*(d^nN)\leq 1+[(d-1)N]^{-1}$ for all $n$.
\end{restatable}

As mentioned in the main text, by Theorem~\ref{thm:linebwopt}, the key metric for the quality of base graph is how large it is while achieving both Moore and BW optimality. Currently, our largest such base graph that works for any even degree is Hamming graph $H(2,\!1\!+\!d/2)$, which has $(1\!+\!d/2)^2\!=\!\Theta(d^2)$ number of nodes. The corresponding line graph expanded topology is always Moore optimal and at most $O(1/d^3)$ away from BW optimality by Theorem~\ref{thm:linebwopt}.

Line graph expansion is closely related to BFB schedule for two reasons: (1) most of our base topologies like complete bipartite graph and Hamming graph use BFB schedule as the base schedule, and (2) the line graph expansion of BFB schedule is still a BFB schedule. To see the performance bound in Theorem~\ref{thm:linegraphperformance} is tight, we have the following results in the context of BFB schedule:
\begin{restatable}{theorem}{thmshortestpathlinegraph}\label{thm:shortestpathlinegraph}
    Let $A_G$ be a BFB allgather schedule for $G$ with $|N^+(u)|\!>\!1$ for all $u\in V_G$, then the expanded schedule $A_{L(G)}$ is a BFB allgather schedule for $L(G)$. In particular, if $A_G$ is the optimal BFB schedule for $G$, then $A_{L(G)}$ is the optimal BFB schedule for $L(G)$ satisfying:
    \begin{equation}\label{eq:shortestpathT_B}
        T_B(A_{L(G)})=T_B(A_G)+\frac{M}{B}\cdot\frac{1}{N}.
    \end{equation}
\end{restatable}
\begin{restatable}{corollary}{thmshortestpathlinegraphcoro}\label{thm:shortestpathlinegraphcoro}
    Let $A_G$ be a BFB allgather schedule for $G$ with $|N^+(u)|\!>\!1$ for all $u\in V_G$, then the expanded schedule $A_{L^n(G)}$ is a BFB allgather schedule for $L^n(G)$. In particular, if $A_G$ is the optimal BFB schedule for $G$, then $A_{L^n(G)}$ is the optimal BFB schedule for $L^n(G)$ satisfying:
    \[
        T_B(A_{L^n(G)})=T_B(A_G)+\frac{M}{B}\cdot\frac{d}{d-1}\left(\frac{1}{N}-\frac{1}{d^nN}\right).
    \]
\end{restatable}

\subsection{Degree Expansion}
\label{app-subsec:degree}

\begin{restatable}[Degree Expanded Topology]{definition}{degexptopodef}\label{def:degreeexp}
Given an $N$-node $d$-regular topology $G$ without self-loops, construct the degree expanded $nN$-node $nd$-regular topology $G*n$:
\begin{enumerate}[label=\arabic*.]
    \item For each vertex $v\in V_G$, add $v_1,\dots,v_n$ to $V_{G*n}$,
    \item For each edge $(u,v)\in E_G$, add $(u_i,v_j)$ to $E_{G*n}$ for all $i,j$ including $i=j$.
\end{enumerate}
\end{restatable}
\degexpscheduledef*
\begin{restatable}{theorem}{thmdegexp}\label{thm:degexp}
    Given a $d$-regular topology $G$ without self loops, if $(G,A_G)$ is an $N$-node allgather algorithm with $T_B(A_G,M,B)=\tau(M/B)$ for some constant $\tau$, then $(G*n,A_{G*n})$ is an $nN$-node allgather algorithm satisfying:
    \begin{gather}
        T_L(A_{G*n})=T_L(A_G)+\alpha,\label{eq:degexplatency}\\
        T_B(A_{G*n})=T_B(A_G)+\frac{M}{B}\cdot\frac{n-1}{nN}.\label{eq:degexpbw}
    \end{gather}
\end{restatable}
\begin{restatable}{corollary}{thmdegexpbwopt}
    If $(G,A_G)$ is BW-optimal and $T_B(A_G,M,B)\!=\!\tau(M/B)$ for some $\tau$, then $(G*n,A_{G*n})$ is BW-optimal.
\end{restatable}
Degree expansion preserves BW optimality. As for total-hop latency of degree expanded topology, observe that $T_L^*(N,d)=\Theta(\log_d N)$ and $\log_{nd} nN<\log_d N$, so $T_L^*$ decreases as we apply degree expansion. Since $T_L$ increases in degree expansion, the total-hop latency optimality is not preserved.

\subsection{Cartesian Product Expansion}
\label{app-subsec:cartesian}

\cartesianprodtopodef*
Definition \ref{def:cartesianprodtopodef} generalizes to Cartesian product of multiple digraphs: $G_1\square G_2\square G_3\!=\!(G_1\square G_2)\square G_3$. The Cartesian product of $n$ identical digraphs is denoted as Cartesian power $G^{\square n}$.

\begin{definition}[Schedule of Cartesian Power]
Given an allgather schedule $A_G$ for topology $G$ and $n\in\mathbb{N}$, construct the schedule $A_{G^{\square n}}$ for $G^{\square n}$:
\begin{enumerate}[label=\arabic*.]
    \item Construct the schedule $A^{(1)}$ as follows:
    \item For $j=1,\dots,n$, for each $((w,C),(u,v),t)\in A_G$, add
    \[
        (((\mathbf{x},w,\mathbf{z}),C),((\mathbf{y},u,\mathbf{z}),(\mathbf{y},v,\mathbf{z})),t+(j-1)t_{\max})
    \]
    to $A^{(1)}$ for all $\mathbf{x},\mathbf{y}\in V_G^{j-1}$ and $\mathbf{z}\in V_G^{n-j}$. $t_{\max}$ is the max comm step in $A_G$.
    \item Similarly, construct $A^{(i)}$ for $i\!=\!2,\dots,n$ that each vertex $\mathbf{v}$ in $A^{(1)}$ is shifted by $i\!-\!1$ to $(\mathbf{v}[n\!-\!i\!+\!2\!:\!n],\mathbf{v}[1\!:\!n\!-\!i\!+\!1])$.
    \item Divide each shard into $n$ equal-sized subshards. Construct schedule $A_{G^{\square n}}$ such that $A^{(i)}$ performs allgather over the $i$-th subshards of all nodes.
\end{enumerate}
\end{definition}

\begin{restatable}{theorem}{thmcartesianpowperformance}\label{thm:cartesianpowperformance}
    Given a $d$-regular topology $G$, if $(G,A_G)$ is an $N$-node allgather algorithm with $T_B(A_G,M,B)=\tau(M/B)$ for some constant $\tau$, then $G^{\square n}$ is an $nd$-regular topology, and $(G^{\square n},A_{G^{\square n}})$ is an $N^n$-node allgather algorithm satisfying:
    \begin{gather}
        T_L(A_{G^{\square n}})=n\cdot T_L(A_G),\label{eq:cartesianlatency}\\
        T_B(A_{G^{\square n}})=T_B(A_G)\cdot\frac{N}{N-1}\cdot\frac{N^n-1}{N^n}.\label{eq:cartesianbw}
    \end{gather}
\end{restatable}
We then have the following corollary:
\begin{restatable}{corollary}{thmcartesianpowerbwopt}
   If $(G,A_G)$ is BW-optimal and $T_B(A_G,M,B)\!=\!\tau(M/B)$ for some $\tau$, then $(G^{\square n},A_{G^{\square n}})$ is BW-optimal.
   \label{coro:CartesianPowerExpansion}
\end{restatable}
Like degree expansion, Cartesian power expansion does not preserve total-hop latency optimality.

We use BFB schedule generation when dealing with Cartesian product of distinct topologies:
\begin{restatable}{theorem}{thmspathcartesianexpopt}\label{thm:spathcartesianexpopt}
    Let $G_1,G_2,\dots,G_n$ be topologies that
    \begin{enumerate}
        \item $G_1,\dots,G_n$ are nontrivial simple digraphs;
        \item Every $G_i$ has BW-optimal BFB allgather schedule.
    \end{enumerate}
    Then, the optimal BFB allgather schedule, i.e. the schedule generated by BFB LP {\normalfont (\ref{lpmodel})}, for $G_1\square\dots\square G_n$ is also BW-optimal. The total-hop latency of the schedule equals $\alpha\cdot D(G_1\square\dots\square G_n)\!=\!\alpha\cdot\sum_iD(G_i)$.
\end{restatable}
The BFB schedule generation can also be used when individual topologies do not have BW-optimal BFB schedules; however, in such a case, we do not have performance bound for the schedule of the Cartesian product.

\section{BFB Schedule Generation}\label{app:shortestpath}

The LP formulation corresponding to $u_2$ in Figure~\ref{fig:shortestpath-illustration} is:
\begin{equation*}
    \begin{array}{lr@{}ll}
    \text{minimize}  & U_{u_2,t}\phantom{x_{v_1,w_1}} & &\\
    \text{subject to}& \displaystyle x_{v_1,(w_1,u_2),t}&\leq U_{u_2,t}, & \\
                      & \displaystyle x_{v_1,(w_2,u_2),t}+x_{v_2,(w_2,u_2),t}&\leq U_{u_2,t}, & \\
                      & \displaystyle x_{v_2,(w_3,u_2),t}&\leq U_{u_2,t}, & \\
                     & \displaystyle x_{v_1,(w_1,u_2),t}+x_{v_1,(w_2,u_2),t}& =1, & \\
                      & \displaystyle x_{v_2,(w_2,u_2),t}+x_{v_2,(w_3,u_2),t}& =1, & \\
                     & 0\leq x_{v,(w,u_2),t}&\leq 1. & \forall v,w
    \end{array}
\end{equation*}

\begin{restatable}[BFB schedule]{definition}{shortestpathdef}\label{def:shortestpath}
    An allgather schedule $A$ for $G$ is a BFB schedule if $A$ satisfies: $((v,C),(w,u),t)\in A$ only if $d(v,u)=d(v,w)+1=t$.
\end{restatable}
\begin{restatable}{theorem}{thmshortestpath}\label{thm:shortestpath}
    A schedule $A$ for $G$ is a BFB allgather schedule if and only if the following are satisfied:
    \begin{enumerate}
        \item\label{thm:shortestpath1} If $((v,C),(w,u),t)\in A$, then $d(v,u)=d(v,w)+1=t$;
        \item\label{thm:shortestpath2} For any distinct $u,v\in V_G$, the collection of chunks $\mathcal{C}_v=\{C\ |\ ((v,C),(w,u),t)\in A\}$ satisfies $S=\bigcup_{C\in\mathcal{C}_v}C$.
    \end{enumerate}
\end{restatable}
Condition \ref{thm:shortestpath1} ensures the schedule follows the breadth-first broadcast order. Condition \ref{thm:shortestpath2} ensures every node receives the entire shard from every other node and thus a valid allgather.

\subsection{Optimality}\label{sec:bfsoptimality}

\begin{restatable}{theorem}{thmbfslatency}
    If $A$ is a BFB schedule for $G$, then $T_L(A)=\alpha\cdot D(G)$.
\end{restatable}
There may exist many BFB schedules for a given topology $G$. They all have the same $T_L$ but may have different $T_B$s. Thus, the optimal BFB schedule is the one with the lowest $T_B$. Since every BFB schedule can be expressed as a solution to linear program (\ref{lpmodel}), we have the following result:
\begin{restatable}{theorem}{thmshortestlpopt}\label{thm:shortestlpopt}
    Given any topology $G$, linear program {\normalfont (\ref{lpmodel})} gives the optimal BFB schedule of $G$.
\end{restatable}
An important implication of Theorem~\ref{thm:shortestlpopt} is that \textbf{if we can show a BW-optimal BFB schedule exists for a topology $G$, then linear program $(\ref{lpmodel})$ is guaranteed to generate one.} This has become an important tool for us to prove that BFB schedule generation can always generate BW-optimal schedules for some families of topologies. For the rest of this section, we show conditions that, if met by a topology, ensure it has a BW-optimal BFB schedule.

The following theorem shows the necessary and sufficient conditions for a BFB allgather schedule to be BW-optimal:
\begin{restatable}{theorem}{thmshortestbwopt}\label{thm:shortestbwopt}
    Suppose $(G,A)$ is a BFB allgather schedule. $(G,A)$ is BW-optimal if and only if:
    \begin{enumerate}
        \item\label{thm:shortestbwopt1} There exists a sequence $N^-_1,N^-_2,\dots,N^-_{D(G)}\in\mathbb N$ such that for any $x\in\mathbb N$ and $u\in V_G$, $|N^-_t(u)|=N^-_x$.
        \item\label{thm:shortestbwopt2} For any $(w,u)\in E_G$, $\sum_{((v,C),(w,u))\in A_t}|C|=\frac{M}{N}|N^-_t(u)|/d=\frac{M}{N}N^-_t/d$.
    \end{enumerate}
\end{restatable}
We assume $G$ is $d$-regular. Condition~\ref{thm:shortestbwopt1} and \ref{thm:shortestbwopt2} together ensure that at each comm step, all links have perfectly balanced workloads. In Theorem~\ref{thm:spathcartesianexpopt}, we have already proven that \textit{a Cartesian product graph has BW-optimal BFB schedule if it is the product of graphs that each have a BW-optimal BFB schedule.} Here, Theorem~\ref{thm:shortestbwopt} also leads to the following sufficient condition for a bidirectional topology to have a BW-optimal BFB schedule:
\begin{restatable}{theorem}{thmshortestbwoptsuff}\label{thm:shortestbwoptsuff}
    There exists a BW-optimal BFB schedule for undirected graph $G$ if for every distance $x$, two of the following constants exist:
     \begin{enumerate}
        \item\label{thm:shortestbwoptconst1} $N_x=|N_x(u)|$ for any $u\in V_G$;
        \item\label{thm:shortestbwoptconst2} $a_x=|N_x(u)\cap N_{x-1}(w)|$ for any $u\in V_G$ and $w\in N(u)$;
        \item\label{thm:shortestbwoptconst3} $b_x=|N(u)\cap N_{x-1}(v)|$ for any $u\in V_G$ and $v\in N_x(u)$.
    \end{enumerate}
    Moreover, if two of $N_x,a_x,b_x$ exist, then the third one must also exist with $N_x=da_x/b_x$.
\end{restatable}
Note that in undirected graphs, we have $N^+_x(u)=N^-_x(u)=N_x(u)$. To understand these constants, $N_x$ is the number of data shards $u$ needs to receive at comm step $x$; $a_x$ is the number of data shards that can be transmitted by each link $(w,u)$ at comm step $x$; $b_x$ is the number of link $(w,u)$s that each data shard can use to transmit the data to $u$ at comm step $x$. These three constants collectively ensure that links are perfectly balanced with each link transmitting $\frac{M}{N}N_x/d=\frac{M}{N}a_x/b_x$ amount of data at comm step $x$.

Now, we give a \textit{necessary and sufficient condition} for any topology to have a BW-optimal BFB schedule. The condition is derived based on the observation that the BFB optimization problem is equivalent to a job scheduling problem. In each comm step $t$, for each node $u$, we have a set of jobs $\{j_1,j_2,\dots,j_m\}$ (data from the source nodes $v\!\in\! N^-_t(u)$) and a set of processors $\{p_1,p_2,\dots,p_d\}$ (links from in-neighbors $w\!\in\! N^-(u)$). There exists a map $f$ from any job to a set of processors that $j_i$ can only be scheduled to the processors in $f(j_i)$ (in-neighbor $w$s satisfying $d(v,u)=d(v,w)+1=t$). Assuming jobs can be arbitrarily divided into subjobs for parallel execution on multiple processors, the problem is how to schedule these jobs to processors so that workloads are balanced across all processors. We have the following result:
\begin{restatable}{theorem}{thmbalancedbfs}\label{thm:balancedbfs}
    The workloads can be balanced if and only if there exists no subset $J\subseteq\{j_1,j_2,\dots,j_m\}$ such that
    \[
        \frac{|J|}{\left|\bigcup_{j\in J} f(j)\right|}>\frac{m}{d}.
    \]
\end{restatable}
Note that there is an independent scheduling problem for each comm step $t$ and node $u$. Therefore, topology $G$ has a BW-optimal BFB schedule if and only if:
\begin{enumerate}
    \item At each comm step $t$, $|N^-_t(u)|$ is the same for all $u\in V_G$.
    \item The scheduling problem w.r.t. each $t$ and $u$ satisfies the condition in Theorem~\ref{thm:balancedbfs}.
\end{enumerate}

\subsection{Discrete Chunked BFB Schedule}\label{sec:shortestIP}
The BFB LP (\ref{lpmodel}) makes an assumption that shards can be divided arbitrarily and infinitesimally. However, to compile the schedule into an executable form, one often needs a \emph{discrete chunked schedule}, where each shard is divided into a fixed number of equal-sized chunks. In practice, $x_{v,(w,u),t}$s are usually solved to be rational numbers. We can divide each shard into a number of chunks equal to the LCM of $x_{v,(w,u),t}$s' denominators so that each $x_{v,(w,u),t}$ represents some integer number of chunks. This approach has worked for us without any issues. However, there exists the case where each shard of the data can only be divided into $P$ equal chunks (i.e., the whole data $M$ can only be divided into $PN$ equal chunks). In such a case, we show that \emph{we can approximate the optimal discrete chunked BFB schedule in polynomial time}.

Consider the following integer program given $u,t$:
\begin{equation}\label{integerprog}
    \begin{array}{lr@{}ll}
        \text{min}  & &W_{u,t} &\\
        \text{s.t.}& \displaystyle\sum_{v} & y_{v,(w,u),t}\leq W_{u,t}, & \forall w\in N^-(u)\\
                        & \displaystyle\sum_{w} & y_{v,(w,u),t}=P, & \forall v\in N^-_{t}(u)\\
                        & y_{v,(w,u),t}&\in\{0,1,\dots,P\}, & \forall w,v.
    \end{array}
\end{equation}
Compared with (\ref{lpmodel}), one can easily see that the optimal solution of (\ref{integerprog}) gives the optimal BFB allgather schedule when each shard of the data can only be divided into $P$ chunks. One can also easily solve the LP relaxation of (\ref{integerprog}) in polynomial time. Let $T_B^{\text{OPT}}$ be the optimal BW runtime of the schedule obtained by directly solving integer program (\ref{integerprog}). Suppose the LP relaxation gives a schedule with BW runtime $T_B^{\text{LP}}$, then it holds that $T_B^{\text{LP}}\leq T_B^{\text{OPT}}$.

Let $y^{\text{LP}}_{v,(w,u),t}$s be the solution to the LP relaxation of (\ref{integerprog}). We can obtain an integral solution $y_{v,(w,u),t}$s of (\ref{integerprog}) by rounding $y^{\text{LP}}_{v,(w,u),t}$s into integers. For each $v$, we have
\[
    \sum_{w}\left\lfloor y^{\text{LP}}_{v,(w,u),t}\right\rfloor\leq P\leq\sum_{w}\left\lceil y^{\text{LP}}_{v,(w,u),t}\right\rceil.
\]
Thus, it is trivial to round $y^{\text{LP}}_{v,(w,u),t}$s to integer $y_{v,(w,u),t}$s that $\sum_w y_{v,(w,u),t}\!=\!P$ and $y_{v,(w,u),t}\!<\!y^{\text{LP}}_{v,(w,u),t}\!+\!1$. We give the following approximation bound for the resulting schedule:
\begin{restatable}{theorem}{thmshortestIP}\label{thm:shortestIP}
    Rounding LP gives a solution with BW runtime $T_B\!\leq\! T_B^{\text{OPT}}\!+\!\frac{M}{B}\cdot\frac{d(d^{D(G)}-1)}{(d-1)PN}$. In addition, if topology $G$ is Moore optimal, then $T_B\!\leq\! T_B^{\text{OPT}}\!+\!\frac{M}{B}\cdot\frac{d}{P}$.
\end{restatable}
The cost $\frac{M}{B}\cdot\frac{d}{P}$ is negligible since $P$ can easily be hundreds or even thousands while degree $d$ is usually a small integer.

\subsection{Heterogeneous BFB Schedule}\label{sec:heterogeneous}

The BFB LP (\ref{lpmodel}) assumes a homogeneous network. It turns out that with little modification, (\ref{lpmodel}) can become an LP for heterogeneous network too:
\begin{equation}\label{lpmodelhetero}
    \begin{array}{lr@{}ll}
    \text{min}  & &U_{u,t} &\\
    \text{s.t.}& \alpha_{w,u}+\frac{M/N}{B_{w,u}}\displaystyle\sum_{v} & x_{v,(w,u),t}\leq U_{u,t}, & \forall w\in N^-(u)\\
                     & \displaystyle\sum_{w} & x_{v,(w,u),t}=1, & \forall v\in N^-_{t}(u)\\
                     & 0\leq &x_{v,(w,u),t}\leq 1, & \forall w,v.
    \end{array}
\end{equation}
$\alpha_{w,u}$ and $B_{w,u}$ are the total-hop latency and bandwidth of link $(w,u)$. In some cases, the $\alpha$ of some link $(w,u)$ is so high that $\alpha_{w,u}$ alone dominates $U_{u,t}$ in (\ref{lpmodelhetero}) even though $\sum_{v}x_{v,(w,u),t}=0$. This is problematic because one should not pay $\alpha_{w,u}$ if link $(w,u)$ is not used. However, such a scenario can be easily detected after solving LP (\ref{lpmodelhetero}). One can avoid the issue by simply removing link $(w,u)$ and solving the LP again.

\section{Generative Topologies}\label{sec:genetopoappen}

In this section, we introduce several topologies for which applying BFB schedule generation yields high-performance communication schedules.

\subsection{Bidirectional Ring}\label{sec:spathbiring}

\begin{figure}[tb]
    \centering
    \scalebox{0.7}{
    \begin{tikzpicture}[roundnode/.style={circle,draw=black,minimum size=7mm}]
        \node[roundnode] (0) at ({-90}:1.5) {$r$};
        \node[roundnode] (1) at ({-90+360*1/5}:1.5) {};
        \node[roundnode] (2) at ({-90+360*2/5}:1.5) {};
        \node[roundnode] (3) at ({-90+360*3/5}:1.5) {};
        \node[roundnode] (4) at ({-90+360*4/5}:1.5) {};

        \path[-latex,anchor=north] (0) edge [bend right=20] node {$S$} (1);
        \path[-latex,anchor=west] (1) edge [bend right=20] node {$S$} (2);

        \path[-latex,anchor=north] (0) edge [bend left=20] node {$S$} (4);
        \path[-latex,anchor=east] (4) edge [bend left=20] node {$S$} (3);
    \end{tikzpicture}
    }
    \hspace{1cm}
    \scalebox{0.7}{
    \begin{tikzpicture}[roundnode/.style={circle,draw=black,minimum size=7mm}]
        \node[roundnode] (0) at ({-90}:1.5) {$r$};
        \node[roundnode] (1) at ({-90+360*1/6}:1.5) {};
        \node[roundnode] (2) at ({-90+360*2/6}:1.5) {};
        \node[roundnode] (3) at ({-90+360*3/6}:1.5) {};
        \node[roundnode] (4) at ({-90+360*4/6}:1.5) {};
        \node[roundnode] (5) at ({-90+360*5/6}:1.5) {};

        \path[-latex,anchor=north] (0) edge [bend right=20] node {$S$} (1);
        \path[-latex,anchor=west] (1) edge [bend right=20] node {$S$} (2);

        \path[-latex,anchor=north] (0) edge [bend left=20] node {$S$} (5);
        \path[-latex,anchor=east] (5) edge [bend left=20] node {$S$} (4);

        \path[-latex,anchor=south] (4) edge [bend left=20] node {$C_1$} (3);
        \path[-latex,anchor=south] (2) edge [bend right=20] node {$C_2$} (3);
    \end{tikzpicture}
    }
    \caption{The broadcast paths of ring BFB allgather schedule. {\normalfont The left and right figures respectively show the broadcast patterns for odd- and even-sized bidirectional rings. Edges of the rings are omitted. $C_1$ and $C_2$ are two halves of shard $S$.}}
    \label{fig:BFBringallgather}
\end{figure}

Ring is the most common topology for allreduce. The traditional schedule on ring is to make each shard go a full circle to do reduce-scatter/allgather. In a bidirectional ring, one can simply make half the shard go clockwise and the other half go counterclockwise to utilize both directions of the links. Such a reduce-scatter/allgather schedule is BW-optimal but poor in total-hop latency with $T_L\!=\!(N-1)\alpha$. With BFB schedule generation, we discovered a new ring reduce-scatter/allgather schedule that achieves half the total-hop latency ($T_L\!=\!\lfloor N/2\rfloor\alpha$) while maintaining BW optimality. From each node, the BFB allgather schedule broadcasts the \emph{entire} shard clockwise and counterclockwise in parallel. Thus, each direction only needs to go half a circle instead of a full circle. If $N$ is even, then the farthest node directly across the ring receives each half of the shard from each of its two neighbors in the end. Figure~\ref{fig:BFBringallgather} shows examples in odd- and even-sized rings respectively.

\subsection{Generalized Kautz Graph}\label{sec:genkautz}

\begin{figure*}[tb]
    \centering
    \includegraphics[width=\textwidth]{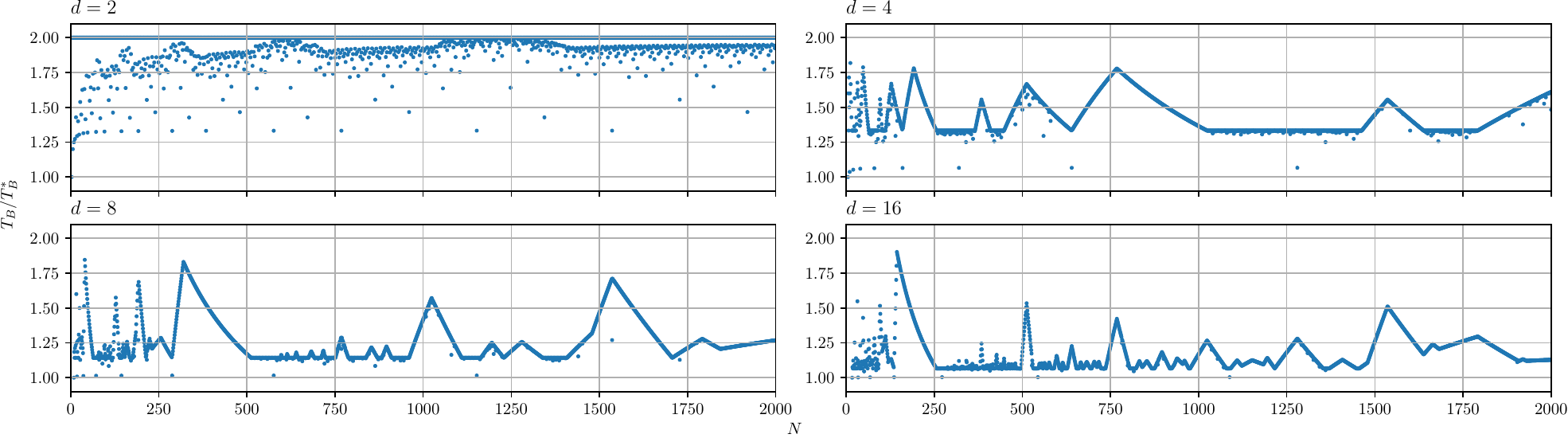}
    \caption{$T_B/T_B^*$ of generalized Kautz graph $\Pi_{d,N}$ up to $N=2000$. {\normalfont As can be seen, the BW runtime of $\Pi_{d,N}$ is less than or equal to $2T_B^*$ at all times for $d=2,4,8,16$. In particular, the higher the degree is, the closer $T_B$ is to optimal. As for total-hop latency, Theorem~\ref{thm:genkautzdiameter} shows that $T_L\leq T_L^*(N,d)+\alpha$.}}
    \label{fig:genkautzbwtime}
\end{figure*}

Generalized Kautz graph~\cite{genkautz,genkautz2} is a low-$T_L$ unidirectional topology that can be constructed for every $N$ and $d$.
\begin{restatable}[Generalized Kautz Graph]{definition}{generalizedkautzdef}\label{thm:genkautzlatency}
    The $\Pi_{d,m}$ digraph has the set of integers modulo $m$ as vertex set. Its arc set $A$ is defined as follows:
    \[
        A=\{(x,y)\ |\ y\equiv -dx-a,1\leq a\leq d\}.
    \]
    If $m=d^{n+1}+d^n$, then $\Pi_{d,m}=K(d,n)$, where $K(d,n)$ is the Kautz graph $L^n(K_{d+1})$.
\end{restatable}
We apply BFB schedule generation to generalized Kautz graph. The resulting schedule is not always Moore optimal, but the following theorem shows that it is at most one $\alpha$ away from Moore optimality, i.e., $T_L\leq T_L^*(N,d)+\alpha$:
\begin{restatable}{theorem}{thmgenkautzdiameter}\label{thm:genkautzdiameter}
    Suppose $D(\Pi_{d,m})=k$, then $m>M_{d,k-2}$.
\end{restatable}
Remember Moore optimality is stricter than total-hop latency optimality, so it is possible that generalized Kautz graph is total-hop latency optimal. The special case, Kautz graph $K(d,n)$, is always Moore optimal and is, in fact, the largest known digraph in \textit{degree/diamter problem} for any degree $d>2$~\cite{moore_bound}.

As for BW performance, from Figure~\ref{fig:genkautzbwtime}, one can see that generalized Kautz graph is also close to BW optimality, especially at higher degrees.

\subsection{Distance-Regular Graph}\label{sec:dist-reg}

In graph theory, distance-regular graphs are a family of highly symmetric undirected graphs. We can show that there exists a BW-optimal BFB schedule for any distance-regular graph, and thus LP (\ref{lpmodel}) can always generate one. We borrow the following definition from \cite{distregdef}:
\begin{definition}[Distance-Regular Graph]
    A connected graph $G$ is distance-regular if for any vertices $x,y\in V_G$ and integers $i,j$, the number of vertices at distance $i$ from $x$ and distance $j$ from $y$ depends only on $i,j$ and $d(x,y)$.
\end{definition}
In other words, there exists a constant $s^{h}_{i,j}$ for every $h,i,j$ such that $s^{h}_{i,j}=|N_i(x)\cap N_j(y)|$ whenever $x,y\in V_G$ satisfy $d(x,y)=h$. Thus, we can apply Theorem~\ref{thm:shortestbwoptsuff} with $N_x=s^0_{x,x}$, $a_x=s^{1}_{x,x-1}$, and $b_x=s^{x}_{1,x-1}$.

The significance of distance-regular graph is not only about BW optimality. Many of distance-regular graphs have low diameters, so their schedules are not only BW-optimal but also close to, and in some cases exactly, Moore optimal. Table~\ref{tab:exampledistreg} gives examples of distance-regular graphs at $d=4$. In addition, many of the base graphs mentioned in this paper are also distance-regular like complete bipartite graphs~(Figure~\ref{fig:notationexample}) and Hamming graphs. One can refer to \cite{distregdb} for a repository of distance-regular graphs.

\begin{table}[tb]
    \centering
    \resizebox{\columnwidth}{!}{
    \begin{tabular}{|l|c|c|c|c|c|c|}
        \hline
        Graph Name & $N$ & $T_L$ & $T_L^*$ & $T_L\!-\!T_L^*$ & $T_L^{**}$ & $T_L\!-\!T_L^{**}$ \\
        \hline
        Octahedron J(4,2) & 6 & 2 & 2 & 0 & 2 & 0 \\
        \hline
        Paley graph P9$\cong$H(2,3) & 9 & 2 & 2 & 0 & 2 & 0 \\
        \hline
        K5,5-I & 10 & 3 & 2 & 1 & 2 & 1 \\
        \hline
        \multirow{2}{*}{\shortstack[l]{Distance-3 graph of \\ Heawood graph}} & \multirow{2}{*}{14} & \multirow{2}{*}{3} & \multirow{2}{*}{2} & \multirow{2}{*}{1} & \multirow{2}{*}{2} & \multirow{2}{*}{1} \\
        & & & & & & \\
        \hline
        Line graph of Petersen graph & 15 & 3 & 2 & 1 & 2 & 1 \\
        \hline
        4-cube Q4$\cong$H(4,2) & 16 & 4 & 2 & 2 & 2 & 2 \\
        \hline
        Line graph of Heawood graph & 21 & 3 & 2 & 1 & 3 & 0 \\
        \hline
        Incidence graph of PG(2,3) & 26 & 3 & 3 & 0 & 3 & 0 \\
        \hline
        \multirow{2}{*}{\shortstack[l]{Incidence graph of AG(2,4) \\ minus a parallel class}} & \multirow{2}{*}{32} & \multirow{2}{*}{4} & \multirow{2}{*}{3} & \multirow{2}{*}{1} & \multirow{2}{*}{3} & \multirow{2}{*}{1} \\
        & & & & & & \\
        \hline
        Odd graph O4 & 35 & 3 & 3 & 0 & 3 & 0 \\
        \hline
        Line graph of Tutte's 8-cage & 45 & 4 & 3 & 1 & 3 & 1 \\
        \hline
        Doubled Odd Graph D(O4) & 70 & 7 & 3 & 4 & 4 & 3 \\
        \hline
        Incidence graph of GQ(3,3) & 80 & 4 & 3 & 1 & 4 & 0 \\
        \hline
        Line graph of Tutte's 12-cage & 189 & 6 & 4 & 2 & 5 & 1 \\
        \hline
        Incidence graph of GH(3,3) & 728 & 6 & 5 & 1 & 6 & 0 \\
        \hline
    \end{tabular}
    }
    \caption{Examples of distance-regular graphs at $d=4$ \cite{distregdb}. {\normalfont $T_L^{**}$ is the bidirectional Moore optimality.}}
    \label{tab:exampledistreg}
\end{table}

\subsection{Circulant Graph}\label{sec:circulant}

Circulant graph is a well-studied topology in both graph theory and network design. Many popular network topologies like shifted ring, chordal ring, and loop network are either part of or closely related to circulant graphs. The definition of circulant graph is as follows:
\begin{definition}
    The circulant graph $C(n,\{a_1,\dots,a_k\})$ is a bidirectional graph with vertex set $\{0,1,\dots,n-1\}$ and each node $i$ is adjacent to nodes $i\pm a_1,\dots,i\pm a_k\pmod{n}$.
\end{definition}
Note that in this paper, we only consider connected circulant graphs, and $C(n,\{a_1,\dots,a_k\})$ is connected if and only if $\gcd(n,a_1,\dots,a_k)=1$ \cite{circulantSurvey, circulantThesis}. It is easy to see that $C(n,\{a_1,\dots,a_k\})$ is an $n$-node $2k$-regular topology.

We have found that the BFB schedule generation seems to give BW-optimal schedules for all circulant graphs. In particular, we have the following conjecture:
\begin{restatable}{conjecture}{circulantconjecture}
    For any circulant graph $C(n,\{a_1,\dots,a_k\})$, there exists a BW-optimal BFB schedule.
\end{restatable}
While we leave a complete proof or disproof of this conjecture for future work, we have proved the conjecture holds when $k=2$, which corresponds to the graph having degree 4.

Circulant graph revolutionized our Pareto frontier of topologies since it can be constructed for every $N$ and even value $d$. It can provide a BW-optimal topology if our expansion techniques fail to produce one at some $N$ and $d$. Since all circulant graphs seem to be BW-optimal, the question is what choices of $a_1,\dots,a_k$ result in minimum total-hop latency, or equivalently, minimum diameter for a given $n$ and $k$. While this remains largely an open question in graph theory \cite{circulantSurvey}, the case of $k=2$ has been solved in \cite{circulantk2}:
\begin{restatable}{theorem}{thmdegfourconstruct}\label{thm:deg4construct}
    Given $n\!>\!6$ and $m\!=\!\lceil(-\!1\!+\!\sqrt{2n\!-\!1})\!/\!2\rceil$, circulant graph $C(n,\{m,m\!+\!1\})$ has a diameter equal to $m$, which is the minimum diameter over all circulant graphs $C(n,\{a_1,a_2\})$.
\end{restatable}
We can certainly use multiedge to apply this construction for any even degree that is $\geq 4$. The resulting topology has $\Theta(\sqrt{N})$ diameter, which is a significant improvement in terms of total-hop latency when BW optimality is required. Previously, the only topology that is known to be BW-optimal for any $N$ and $d$ is ring, which has $\Theta(N)$ diameter.

\input{proofs}

\section{Supplementary Tables and Figures}\label{sec:suppfigtab}

\begin{figure}[h]
    \centering
    \scalebox{0.8}{
        \input{figures/diamond}
    }
    \caption{Diamond Topology $(N=8,d=2)$.}
    \label{fig:diamondtopo}
\end{figure}

\begin{figure}[h]
    \centering
    \includegraphics[width=\columnwidth]{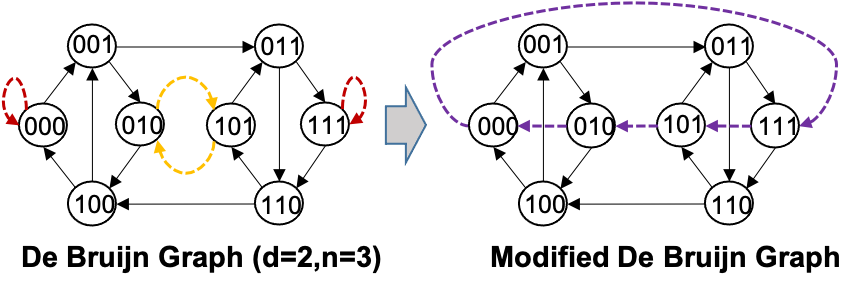}
    \caption{An Example of Modified de Bruijn Graph $(N=8,d=2)$. {\normalfont The modification rewires the self loops and 2-cycles in de Bruijn graph to form a single long cycle without violating degree constraint.}}
    \label{fig:modifieddbjg}
\end{figure}

\input{basegraphtable}

%% file: eval_sup.tex
\section{Evaluation Appendix}\label{app-sec:evalapp}

\squishlist
    \item \S\ref{sec:switchcomp} presents experiment results that compare BFB schedule generation with communication solutions for switch networks: NCCL~\cite{nccl} and recursive halving \& doubling.
    \item \S\ref{sec:cost-model} shows experiment results to validate $\alpha$-$\beta$ cost model.
    \item \S\ref{sec:paretoanalysis} gives an analysis of Pareto-efficient topologies/schedules under different hardware and workload specifications.
    \item \S\ref{app-sec:simtrainingdetails} details setup of simulated DNN training and the topologies generated by our topology finder.
    \item \S\ref{app-sec:all-to-all} provides the multi-commodity flow (MCF) formulation used to compute all-to-all throughput.
    \item \S\ref{sec:unitobidirectional} shows how to convert unidirectional topologies/schedules into bidirectional ones.
\squishend

\begin{figure}[htb]
    \centering
    \vspace{0.5cm}
    \includegraphics[width=\columnwidth]{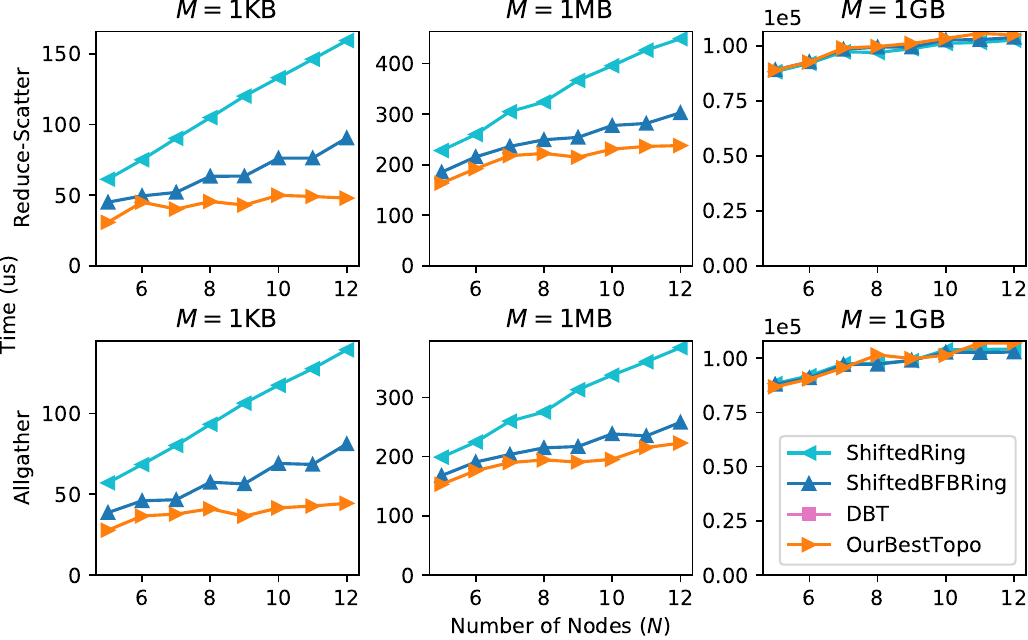}
    \caption{Comparing reduce-scatter and allgather runtimes of topologies. {\normalfont This figure shows the corresponding reduce-scatter and allgather results of Figure~\ref{fig:arexptres}.}}
    \label{fig:rs_ag_results}
\end{figure}

\subsection{Comparison Against Switch Solutions}\label{sec:switchcomp}

NCCL~\cite{nccl} and recursive halving \& doubling (RH\&D) are widely adopted collective communication solutions on switch networks. We assess the schedule performance of BFB against these solutions over two direct-connect 8-node topologies: hypercube and twisted hypercube~\cite{twistedncube}. Hypercube is widely used in HPC settings, and its connections perfectly match the communication pattern of RH\&D. Twisted hypercube is a variant of hypercube with a lower diameter.

Figure~\ref{fig:hypercubecompare} compares the baselines against our BFB schedule when run over either hypercube or twisted hypercube with $N\!=\!8$, $d\!=\!3$ on the testbed.
At small $M$, all schedules and topologies perform roughly the same, except BFB can take advantage of the lower diameter of twisted hypercube and achieve $\sim\!\! 20\%$ lower runtime.
At large $M$, because BFB achieves BW optimality on both topologies, it performs even better with 60\% lower runtime.
RH\&D and NCCL perform poorly as $M$ grows because they cannot utilize all $d\!=\!3$ links simultaneously. At every comm step of RH\&D, a node only communicates with one of the three neighbors, utilizing at most $1/3$ of the total bandwidth (similarly with NCCL).
Also, because the schedule is not matched to the twisted hypercube, some nodes communicate with nodes multiple hops away, occupying more link capacities and causing congestion.

In summary, \textbf{BFB schedules can fully exploit direct-connect topologies and provide better performance than solutions designed for switch networks}.

\begin{figure}[tb]
    \centering
    \includegraphics[scale=0.5]{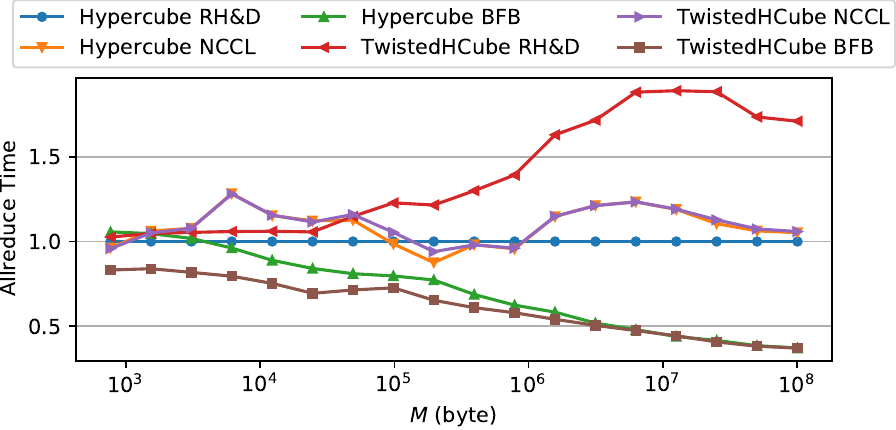}
    \caption{Comparing switch allreduce solutions (recursive halving \& doubling (RH\&D), NCCL) against BFB schedule on hypercube and twisted hypercube on $N\!=\!8,d\!=\!3$ testbed. {\normalfont The runtimes are normalized by the runtime of recursive halving \& doubling on hypercube.}}
    \label{fig:hypercubecompare}
\end{figure}

\subsection{Cost Model Validation}
\label{sec:cost-model}

Despite the wide acceptance of $\alpha$-$\beta$ cost model by previous literature \cite{hockney1994communication, conc_comp, SCCL, shah2021synthesizing, dbtpaper}, we also did a linear regression analysis to verify the cost model on our testbed. In particular, we want to verify that (1) total-hop latency follows $T_L=\alpha\cdot x+\epsilon$ and (2) BW runtime follows $T_B=\frac{M}{B}\cdot y$, where $x$ and $y$ are the number of comm steps and bandwidth factor respectively ($y=2\cdot\frac{N-1}{N}$ if BW-optimal). $\epsilon$ is the constant latency\footnote{This part of latency is a fixed constant for all topologies and schedules, so it is omitted earlier.} including time costs such as GPU kernel launching. Here, we use linear regression to derive the values of $\alpha$, $\epsilon$, and $1/B$, and compute the relative errors between observed runtimes and fitted runtimes. We fit the allreduce runtimes at 1KB to the total-hop latency, since BW runtime is negligible at such a small $M$. Similarly, we fit runtimes at 1GB to the BW runtime, since total-hop latency is negligible at such a large $M$.

Figure~\ref{fig:latencyfit} and \ref{fig:betafit} show our results of linear regression analysis to verify our cost model. For total-hop latency, we obtain estimates $\alpha\!\approx\! 13.33$us and $\epsilon\!\approx\! 21.60$us with low errors (average and maximum relative errors of 1.71\% and 6.21\% respectively). As one can see from Figure~\ref{fig:latencyfit}, ShiftedRing and ShiftedBFBRing have a straight and a stair-step shape of runtime growth respectively, which match the expected numbers of comm steps $2(N-1)$ and $2\lfloor N/2\rfloor$ respectively. For BW runtime, we get an estimate $1/B\!\approx\! 1.018\!\times\! 10^{-4}$us/byte or $B\!\approx\! 79$Gbps with low errors (average and maximum relative errors of 0.47\% and 1.32\% respectively). As one can see from Figure~\ref{fig:betafit}, all three topologies follow the fitted curve $2\frac{\text{1GB}}{B}\cdot\frac{N-1}{N}\!=\!2T^*_B(N)$ since they are all BW-optimal. However, there is a gap between $B\!\approx\! 79$Gbps and the hardware theoretical bandwidth 4x25Gbps$=$100Gbps. Besides inevitable loss of bandwidth in actual communication, the gap can also be explained by the fact that computational cost of reduction also accounts for part of $1/B$ as discussed in \S\ref{sec:computeopt}.

\begin{figure}[tb]
    \centering
    \begin{subfigure}{\columnwidth}
        \centering
        \includegraphics[scale=0.5]{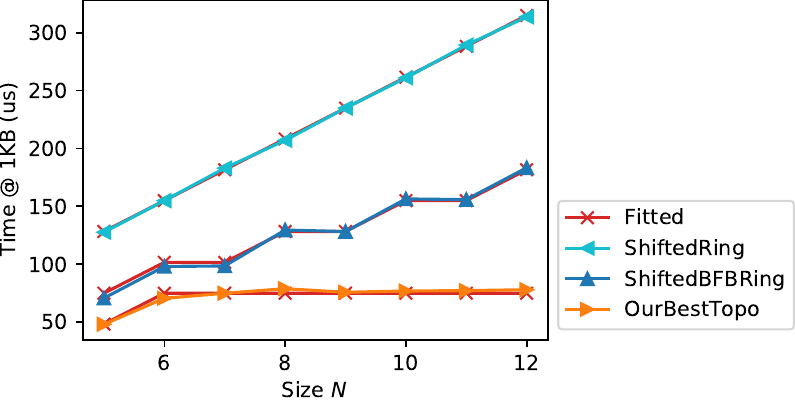}
        \caption{Total-Hop Latency $T_L$ {\normalfont (Relative error: avg 1.71\%, max 6.21\%)}}
        \label{fig:latencyfit}
    \end{subfigure}
    \begin{subfigure}{\columnwidth}
        \centering
        \includegraphics[scale=0.5]{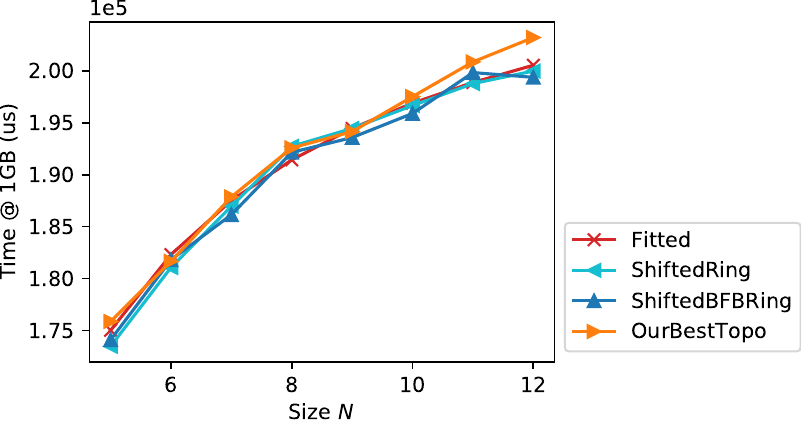}
        \caption{BW Runtime $T_B$ {\normalfont (Relative error: avg 0.47\%, max 1.32\%)}}
        \label{fig:betafit}
    \end{subfigure}
    \caption{Linear regression results.}
\end{figure}

\subsection{Pareto-Efficiency Analysis}\label{sec:paretoanalysis}

\begin{figure*}[tb]
    \centering
    \includegraphics[width=\textwidth]{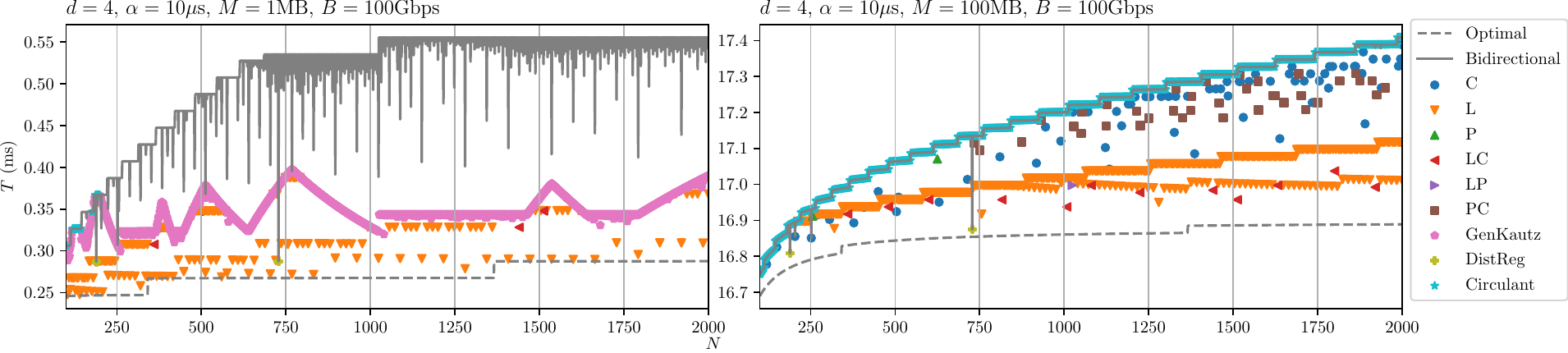}
    \caption{The minimum allreduce runtimes at different $N$ for $d=4$, $\alpha=10\mu s$, and $M/B=1{\normalfont MB}/100{\normalfont Gbps},100{\normalfont MB}/100{\normalfont Gbps}$. {\normalfont ``L'', ``P'', and ``C'' stand for line graph, Cartesian power, and Cartesian product (of different graphs) respectively. For example, ``LC'' means the runtime is achieved by a topology whose construction involves line graph expansion and Cartesian product. ``GenKautz'', ``DistReg'', and ``Circulant'' stand for generalized Kautz graph (\S\ref{sec:genkautz}), distance regular graph (\S\ref{sec:dist-reg}), and circulant graph (\S\ref{sec:circulant}) respectively. The figures also show the best bidirectional topology known at different $N$s. Degree expansion does not show up due to target $d=4$ being relatively small.}}
    \label{fig:opttopodiffn}
\end{figure*}

There could exist multiple Pareto-efficient topologies at given $N$ and $d$. For different $\alpha$ and $M/B$, the Pareto-efficient topology with minimum allreduce runtime is also different.
To see how $N$ affects the best choice of topologies, we use topology finder~(\S\ref{sec:topofinder}) to generate Pareto-efficient topologies at $d=4$ for $N$ up to 2000 and pick the best one based on specific values of $\alpha$ and $M/B$.

Figure~\ref{fig:opttopodiffn} shows two examples of such analysis.
At $M=1$MB, total-hop latency is more important than BW runtime. Thus, we see that generalized Kautz graph is the most popular one, being the best topology at many $N$s. On the contrary, at $M=100$MB, BW performance becomes the dominant factor, and thus circulant graph becomes the most popular one. Line graphs are also popular in both settings; however, line graph expansion requires target $N$ to be divisible by some power of $d$, so it does not work for any $N$.

\input{sim_details}

\input{undirected}

%% file: sim_details.tex
\subsection{Details of Simulated Distributed Training}\label{app-sec:simtrainingdetails}

\begin{figure}[tb]
    \centering
    \includegraphics[width=\columnwidth]{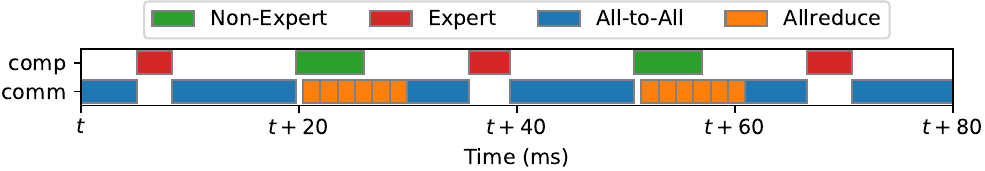}
    \caption{Example of a training timeline for Switch Transformers.}
    \label{fig:simmoetimeline}
\end{figure}

\begin{table}[tb]
    \centering
    \resizebox{\columnwidth}{!}{
    \begin{tabular}{|l||c|c||c|c||}
	\hline
	\textbf{Topology} & \multicolumn{2}{c||}{\textbf{Allreduce}} & \multicolumn{2}{c||}{\textbf{All-to-All}} \\
	\hline
	\hline
	$N\!=\!32$, $d\!=\!4$ & $T_L$ & $T_B$ & \scalebox{0.9}[1]{$D(G)$} & MCF \\
	\hline
	$L(K_{4,4})$ & $3\alpha$ & $1.000\sfrac{M}{B}$ & $3$ & $5.71\text{e}\!-\!2$ \\
	$\text{DistReg}(4, 32)$ & $4\alpha$ & $0.969\sfrac{M}{B}$ & $4$ & $5.26\text{e}\!-\!2$ \\
	\hline
	Theoretical Bound & $3\alpha$ & $0.969\sfrac{M}{B}$ & $3$ & $5.80\text{e}\!-\!2$ \\
	\hline
	\hline
	$N\!=\!64$, $d\!=\!4$ & $T_L$ & $T_B$ & \scalebox{0.9}[1]{$D(G)$} & MCF \\
	\hline
	$\Pi_{4, 64}$ & $3\alpha$ & $1.312\sfrac{M}{B}$ & $3$ & $2.17\text{e}\!-\!2$ \\
	$L(\text{DBJMod}(4, 2))$ & $4\alpha$ & $1.000\sfrac{M}{B}$ & $4$ & $2.21\text{e}\!-\!2$ \\
	$\text{Diamond}^{\square 2}$ & $6\alpha$ & $0.984\sfrac{M}{B}$ & $6$ & $1.87\text{e}\!-\!2$ \\
	\hline
	Theoretical Bound & $3\alpha$ & $0.984\sfrac{M}{B}$ & $3$ & $2.42\text{e}\!-\!2$ \\
	\hline
	\hline
	$N\!=\!128$, $d\!=\!4$ & $T_L$ & $T_B$ & \scalebox{0.9}[1]{$D(G)$} & MCF \\
	\hline
	$L^2(K_{4,4})$ & $4\alpha$ & $1.031\sfrac{M}{B}$ & $4$ & $9.89\text{e}\!-\!3$ \\
	$L(\text{DistReg}(4, 32))$ & $5\alpha$ & $1.000\sfrac{M}{B}$ & $5$ & $9.26\text{e}\!-\!3$ \\
	$\text{BiRing}(2, 8)\square\text{UniRing}(1, 4)^{\square 2}$ & $10\alpha$ & $0.992\sfrac{M}{B}$ & $10$ & $5.21\text{e}\!-\!3$ \\
	\hline
	Theoretical Bound & $4\alpha$ & $0.992\sfrac{M}{B}$ & $4$ & $1.00\text{e}\!-\!2$ \\
	\hline
	\hline
	$N\!=\!256$, $d\!=\!4$ & $T_L$ & $T_B$ & \scalebox{0.9}[1]{$D(G)$} & MCF \\
	\hline
	$\text{DBJ}(4, 4)$ & $4\alpha$ & $1.328\sfrac{M}{B}$ & $4$ & $4.04\text{e}\!-\!3$ \\
	$L^2(\text{DBJMod}(4, 2))$ & $5\alpha$ & $1.016\sfrac{M}{B}$ & $5$ & $4.10\text{e}\!-\!3$ \\
	$L(\text{Diamond}^{\square 2})$ & $7\alpha$ & $1.000\sfrac{M}{B}$ & $7$ & $3.62\text{e}\!-\!3$ \\
	$\text{DBJMod}(2, 4)^{\square 2}$ & $10\alpha$ & $0.996\sfrac{M}{B}$ & $8$ & $2.94\text{e}\!-\!3$ \\
	\hline
	Theoretical Bound & $4\alpha$ & $0.996\sfrac{M}{B}$ & $4$ & $4.39\text{e}\!-\!3$ \\
	\hline
	\hline
	$N\!=\!512$, $d\!=\!4$ & $T_L$ & $T_B$ & \scalebox{0.9}[1]{$D(G)$} & MCF \\
	\hline
	$L^3(K_{4,4})$ & $5\alpha$ & $1.039\sfrac{M}{B}$ & $5$ & $1.88\text{e}\!-\!3$ \\
	$L^2(\text{DistReg}(4, 32))$ & $6\alpha$ & $1.008\sfrac{M}{B}$ & $6$ & $1.78\text{e}\!-\!3$ \\
	$L(\text{BiRing}(1, 4)^{\square 3}\square\text{UniRing}(1, 2))$ & $11\alpha$ & $1.000\sfrac{M}{B}$ & $11$ & $1.12\text{e}\!-\!3$ \\
	$\text{UniRing}(1, 4)^{\square 3}\square\text{UniRing}(1, 8)$ & $16\alpha$ & $0.998\sfrac{M}{B}$ & $16$ & $5.58\text{e}\!-\!4$ \\
	\hline
	Theoretical Bound & $5\alpha$ & $0.998\sfrac{M}{B}$ & $5$ & $1.90\text{e}\!-\!3$ \\
	\hline
	\hline
	$N\!=\!1024$, $d\!=\!4$ & $T_L$ & $T_B$ & \scalebox{0.9}[1]{$D(G)$} & MCF \\
	\hline
	$\Pi_{4,1024}$ & $5\alpha$ & $1.332\sfrac{M}{B}$ & $5$ & $8.01\text{e}\!-\!4$ \\
	$L^3(C(16,\{3,4\}))$ & $6\alpha$ & $1.020\sfrac{M}{B}$ & $6$ & $8.12\text{e}\!-\!4$ \\
	$L^2(\text{Diamond}^{\square 2})$ & $8\alpha$ & $1.004\sfrac{M}{B}$ & $8$ & $7.34\text{e}\!-\!4$ \\
	$L(\text{DBJMod}(2,4)^{\square 2})$ & $11\alpha$ & $1.000\sfrac{M}{B}$ & $9$ & $6.18\text{e}\!-\!4$ \\
	$(\text{UniRing}(1,4)\square \text{UniRing}(1,8))^{\square 2}$ & $20\alpha$ & $0.999\sfrac{M}{B}$ & $20$ & $2.79\text{e}\!-\!4$ \\
	\hline
	Theoretical Bound & $5\alpha$ & $0.999\sfrac{M}{B}$ & $5$ & $8.57\text{e}\!-\!4$ \\
	\hline
    \end{tabular}
    }
    \caption{Pareto-efficient topologies at $N\!\in\!\{32,64,128,256,512,1024\}$, $d\!=\!4$. {\normalfont The results are generated from the topology finder (\S\ref{sec:topofinder}). For notations of the topologies, see Table~\ref{tab:expansiontechniques} and \ref{table:toposummary}. For distance regular graphs ($\text{DistReg}$), see Table~\ref{tab:exampledistreg}. The MCF values are computed using LP~(\ref{eq:mcflp}).}}
    \label{tab:allourtopos}
\end{table}

We simulate distributed ML training by first collecting actual compute times for model layers, running the models on an NVIDIA A100-SXM-80GB GPU, and then adding communication times according to the specific parallelism, e.g., data or expert parallelism. The communication time is calculated using $\alpha$-$\beta$ model for allreduce (\S\ref{sec:cost-model}) and multi-commodity flow for all-to-all (\S\ref{app-sec:all-to-all}), assuming $\alpha\!=\!10\mu$s, $B\!=\!100\text{Gbps}$ over $d\!=\!4$. Our simulation is designed to match the compute-communication overlap pattern of PyTorch Distributed Data Parallel~\cite{pytorch-dist}. As in PyTorch DDP, we bucket gradients that are ready for allreduce during backward propagation. Once the gradient volume reaches a predefined bucket capacity, an allreduce is performed. While a large bucket size results in less latency overhead, a small bucket size enhances compute-communication overlap. We choose the best bucket size by comparing the iteration times of bucket sizes \{1MB, 10MB, 100MB, 1GB\}. To ensure overlap, computation and communication are handled as independent streams, with the communication stream executing one collective at a time.

For the simulated training of Mixture-of-Experts (MoE) models, we follow the standard practice of expert parallelism~\cite{switchtran,deepspeedmoe,gshard,lina}, where experts are sharded across all nodes while non-expert layers are replicated. All-to-all communications are needed before the expert layers to route tokens to the nodes of the assigned experts, and afterward to return tokens to the original nodes for the continuation of the forward/backward pass, thus blocking the computation stream. Furthermore, all-to-all and allreduce are not allowed to be overlapped as they occupy the same network bandwidth~\cite{lina}. Figure~\ref{fig:simmoetimeline} shows a timeline example of the simulated expert-parallel training. For simplicity, we assume a uniform token distribution among the experts, as MoE models are trained to balance expert load~\cite{zoph2022stmoe,shazeer2017,switchtran}. Consequently, the all-to-all communication is uniform across the nodes. We use the multi-commodity flow formulation~(\ref{eq:mcflp}) to compute the all-to-all communication time.

All hyperparameters, including sequence lengths and global batch sizes, are chosen according to the original paper of Switch Transformers~\cite{switchtran}. The topology degree is fixed at $4$, and the topology sizes are chosen such that the local batch size at each node is $\geq 1$ and not so large as to run out of GPU memory. Table~\ref{tab:allourtopos} includes all the Pareto-efficient topologies used in the simulation. For each model and topology size, we choose the topology that results in the smallest iteration time.

\subsection{All-to-All Throughput}\label{app-sec:all-to-all}

The problem of deriving the throughput of all-to-all communication on a topology can be nicely formulated as a multi-commodity flow (MCF) problem~\cite{basu2024efficient,alltoall1,alltoall2,alltoall3,xpander}. In an all-to-all MCF, each pair of nodes $(s,t)\in V_G^2$ acts as the source and sink of a commodity. The objective is to simultaneously route $f$ units of flows from each $s$ to $t$ such that $f$ is maximized, with flow allocation subject to flow conservation and edge capacities. In \cite{basu2024efficient}, the authors have devised an efficient LP formulation to compute the optimal $f$:
\begin{equation}\label{eq:mcflp}
	\begin{array}{lr@{}ll}
	\text{maximize}  & \displaystyle f & &\\
	\text{subject to}& \displaystyle\sum_{s} y_{s,(u,v)} &\leq 1, & \forall u,v\\
					 & \displaystyle f+\sum_{v} y_{s,(u,v)} &\displaystyle\leq\sum_{w} y_{s,(w,u)}, & \forall s,u\!:\!s\!\neq\! u\\
	                 & y_{s,(u,v)} & \geq 0. & \forall s,u,v
	\end{array}
\end{equation}
In LP~(\ref{eq:mcflp}), we assume the capacity/bandwidth of each link is 1 unit. Therefore, if the bandwidth of each link is $B/d$, then $fB/d$ represents the rate at which every node can send to every other node simultaneously.

%% file: undirected.tex
\subsection{Unidirectional to Bidirectional}\label{sec:unitobidirectional}

Unidirectional topologies are technically feasible on optical testbeds. The optical cable contains two fibers, one for each direction, and the fabric can link them to two distinct end-hosts, thus enabling unidirectional topologies at no additional hardware cost. 

However, in our evaluation, we only use bidirectional topologies.
While unidirectional topologies can be realized by configuring the patch panel in simplex mode, the requisite overlay routing for the reverse path traffic (acks, etc.) is currently only supported using routing rules performed by the host kernel as opposed to the NIC, leading to unpredictable RTTs.
Therefore, we can functionally validate unidirectional topologies on our testbed, but we cannot accurately evaluate their performance.
Note that newer NICs~\cite{bcomp2200g,connectx6} do support hardware offloading for these rules, which we will examine in future work.

While this paper considers unidirectional topologies a lot, many of the techniques can be conveniently applied to bidirectional topologies as well. For example, BFB schedule generation, degree expansion, and Cartesian product can all be used on bidirectional topologies by replacing each bidirectional edge with two opposite unidirectional edges. The resulting degree expanded and Cartesian product topologies still have unidirectional edges in opposite pairs. Although line graph expansion only works within unidirectional topologies, there is a way to convert unidirectional topology and schedule to bidirectional ones with zero performance sacrifice. In this section, we will show how to convert a reverse-symmetric $d$-regular unidirectional topology $G$ and its allgather schedule $A$ to a $2d$-regular bidirectional topology $G'$ and its schedule $A'$ such that $T_L(A)\!=\!T_L(A')$ and $T_B(A)\!=\!T_B(A')$.

Let $g\!:\!V_{G}\!\to\! V_{G^T}$ be the isomorphism from $G$ to $G^T$, then it is trivial to see that $g(A)$ (see Definition~\ref{def:scheduleisomorphism}) is an allgather schedule for $G^T$. Observe that $G'\!=\!G\cup G^T$ is a $2d$-regular bidirectional topology. Consider both $A$ and $g(A)$ as allgather schedules for bidirectional topology $G'$. Schedules $A$ and $g(A)$ use disjoint sets of edges, because they use opposite directions. Thus, we can divide each shard into two halves. Let one half follow schedule $A$ and the other half follow $g(A)$. Let such a schedule be $A'$.

It is trivial to see that $T_L(A)\!=\!T_L(A')$. As for $T_B(A)\!=\!T_B(A')$, it follows the fact that the total data size is halved for each of $A$ and $g(A)$, but the bandwidth per edge is also halved due to the doubling of degree. Note that if $A$ is BW-optimal, then $A'$ is BW-optimal; however, $A'$ is not necessarily Moore optimal if $A$ is Moore optimal.

%% file: proofs.tex
\section{Proofs}\label{sec:proof}

\thmlinegraphperformance*
\begin{proof}
    Let $v'\!v,uw$ be arbitrary two distinct vertices in $L(G)$. We want to show there exists a sequence in $A_{L(G)}$ going from $v'\!v$ to $uw$ like in Definition~\ref{def:allgather} for any $x\!\in\! S$. If $u\!=\!v$, then $((v'\!v, S),(v'\!v,uw),1)$ at the first comm step suffices. If $u\!\neq\! v$, because $A_G$ is allgather, there exists a sequence in $A_G$:
    \begin{multline*}
        ((v,C_1),(v,w_1),t_1),((v,C_2),(w_1,w_2),t_2),\dots\\
        ((v,C_n),(w_{n-1},u),t_n),
    \end{multline*}
    where $t_1<t_2<\dots<t_n$ and $x\!\in\! C_1\cap C_2\cap\dots\cap C_n$. Thus, by Definition~\ref{def:linegraph}, there exists a sequence in $A_{L(G)}$:
    \begin{multline*}
        ((v'\!v,S),(v'\!v,vw_1),1),((v'\!v,C_1),(vw_1,w_1w_2),t_1+1),\dots\\
        ((v'\!v,C_n),(w_{n-1}u,uw),t_n+1),
    \end{multline*}
    as desired. The new algorithm $(L(G),A_{L(G)})$ has $dM$ total data length, because the number of nodes has grown $d$-fold while the size of a shard remains the same.
    
    As for $T_L(A_{L(G)})$ and $T_B(A_{L(G)})$, equality (\ref{eq:linelatency}) trivially follows the Definition~\ref{def:linegraph}. Let $[A_{L(G)}]_t$ and $[A_G]_t$ be the subschedules of $A_{L(G)}$ and $A_G$ at comm step $t$. Given $v\!\in\! V_G$, because $G$ is $d$-regular, we have $|\{v'\!v\ |\ v'\!v\!\in\! V_{L(G)}\}|\!=\!|\{(v',v)\ |\ (v',v)\!\in\! E_G\}|\!=\!d$. Given any edge $(uw,ww')$ and $t$, there are at most $d$ number of $((v'\!v,C),(uw,ww'),t+1)\!\in\! A_{L(G)}$ for each $((v,C),(u,w),t)\!\in\! A_G$ by Definition~\ref{def:linegraph}. Thus, given $(uw,ww')$,
    \[
        \sum_{((v'\!v,C),(uw,ww'))\in [A_{L(G)}]_{t+1}}|C|\leq \sum_{((v,C),(u,w))\in [A_G]_t}d\cdot|C|.
    \]
    It follows that $T_B([A_{L(G)}]_{t+1},dM,B)\leq d\cdot T_B([A_G]_t,M,B)$ and hence $\sum_{t=2}^{t_{\max}+1}T_B([A_{L(G)}]_t,dM,B)\leq d\cdot T_B(A_G,M,B)$. For the first comm step, we have
    \[
        T_B([A_{L(G)}]_{1},dM,B)=\frac{|S|}{B/d}=\frac{M/N}{B/d}.
    \]
    Assuming $T_B(A_G,M,B)\!=\!\tau(M/B)$ for some constant $\tau$, we have $d\cdot T_B(A_G,M,B)\!=\!T_B(A_G,dM,B)$. It follows that
    \begin{multline*}
        T_B(A_{L(G)},dM,B)=\sum_{t=1}^{t_{\max}+1}T_B([A_{L(G)}]_t,dM,B)\\
        \leq\frac{M/N}{B/d}+d\cdot T_B(A_G,M,B)=T_B(A_G,dM,B)+\frac{dM}{B}\cdot\frac{1}{N}.
    \end{multline*}
    Replacing $dM$ by $M$ gives (\ref{eq:linebw}) as desired.
\end{proof}

\thmlinegraphconclusion*

\thmlinebwopt*
\begin{proof}
    If $(G,A_G)$ is BW-optimal, then $T_B(A_G)\!=\!\frac{M}{B}\cdot\frac{N-1}{N}$ and
    \begin{equation}\label{eq:bwoptlinenbw}
        T_B(A_{L^n(G)})\leq \frac{M}{B}\left[1+\frac{1}{d-1}\left(\frac{1}{N}-\frac{d}{d^nN}\right)\right].
    \end{equation}
    It is trivial to see that $(\ref{eq:bwoptlinenbw})/T_B^*(d^nN)\nearrow 1+[(d-1)N]^{-1}$ as $n\to\infty$.
\end{proof}

\thmscheduleswitch*
\begin{proof}
    Suppose $(G,A)$ is a reduce-scatter algorithm. For arbitrary $x\!\in\! S$ and distinct $u,v\!\in\! V_G$, there exists a sequence of tuples in $A$:
    \begin{multline*}
        ((v,C_1),(u,w_1),t_1),((v,C_2),(w_1,w_2),t_2),\dots\\
        ((v,C_n),(w_{n-1},v),t_n),
    \end{multline*}
    where $t_1<t_2<\dots<t_n$ and $x\!\in\! C_1\cap C_2\cap\dots\cap C_n$. It follows that there exists a sequence of tuples in $A^T$:
    \begin{multline*}
        ((v,C_n),(v,w_{n-1}),t_n'),\dots\\
        ((v,C_2),(w_2,w_1),t_2'),((v,C_1),(w_1,u),t_1').
    \end{multline*}
    where $t_i'\!=\!t_{\max}-t_i+1$, so $t_n'<\dots<t_2'<t_1'$. Since $u,v,x$ are abitrary, $A^T$ is an allgather schedule on $G^T$. One can similarly show that if $(G,A)$ is an allgather algorithm, then $(G^T,A^T)$ is a reduce-scatter algorithm.
\end{proof}

\thmscheduleswitchfunc*
\thmscheduleswitchmap*

\thmrescallather*

\begin{proof}
    By definition of $f(A^T)$,
    \begin{align*}
        ((f(v),C),(f(w),f(u)),&t_{\max}-t+1)\in f(A^T)\\
        \Leftrightarrow\ &((v,C),(w,u),t_{\max}-t+1)\in A^T\\
        \Leftrightarrow\ &((v,C),(u,w),t)\in A.
    \end{align*}
    Note that $(u,w)\!\in\! E_G\Leftrightarrow (w,u)\!\in\! E_{G^T}\Leftrightarrow(f(w),f(u))\!\in\! E_G$, so $f(A^T)$ is a valid schedule for $G$.
    
    Suppose $(G,A)$ is a reduce-scatter algorithm. For any $x\!\in\! S$ and distinct $u,v\!\in\! V_G$, there exists a sequence of tuples in $A$:
    \begin{multline*}
        ((v,C_1),(u,w_1),t_1),((v,C_2),(w_1,w_2),t_2),\dots\\
        ((v,C_n),(w_{n-1},v),t_n),
    \end{multline*}
    where $t_1<t_2<\dots<t_n$ and $x\!\in\! C_1\cap C_2\cap\dots\cap C_n$. It follows that there exists a sequence of tuples in $f(A^T)$:
    \begin{gather*}
        ((f(v),C_n),(f(v),f(w_{n-1})),t_n'),\\
        ((f(v),C_{n-1}),(f(w_{n-1}),f(w_{n-2})),t_{n-1}'),\\
        \vdots\\
        ((f(v),C_1),(f(w_1),f(u)),t_1'),
    \end{gather*}
    where $t_i'\!=\!t_{\max}-t_i+1$, and $x\!\in\! C_n\cap C_{n-1}\cap\dots\cap C_1$. Because $f$ is a bijection, $(G,f(A^T))$ is an allgather algorithm. $T_L(A)\!=\!T_L(f(A^T))$ and $T_B(A)\!=\!T_B(f(A^T))$ are trivial, and one can similarly prove that if $(G,A)$ is an allgather algorithm, then $(G,f(A^T))$ is a reduce-scatter algorithm.
\end{proof}

\thmdiameter*

\begin{proof}
    The proof is mentioned in text.
\end{proof}

\thmlatencyoptimaldiameter*

\thmbandwidthlowerbound*

\begin{proof}
    The proof is mentioned in text.
\end{proof}

\thmbwoptequation*

\thmbwopt*

\begin{proof}
    If $T_B(A)\!=\!T^*_B(N)\!=\!\frac{M}{B}\cdot\frac{N-1}{N}$, then the amount of data received by each vertex must be equal to $M\cdot\frac{N-1}{N}$, and the ingress bandwidth $B$ must be fully utilized. If condition \ref{thm:bwopt1} does not hold, then some link $(w,u)$ is not fully utilized. If condition \ref{thm:bwopt2} does not hold, then the amount of data received by some node is greater than $M\cdot\frac{N-1}{N}$.
    
    If both \ref{thm:bwopt1} and \ref{thm:bwopt2} hold, then every vertex receives exactly $M\cdot\frac{N-1}{N}$ in total and bandwidth are fully utilized. Thus, $T_B(A)\!=\!T^*_B(N)$ and $(G,A)$ is BW-optimal.
\end{proof}

\thmcomputebound*

\begin{proof}
    At any comm step $t$, suppose the BW runtime is $T_B(A_t)\!=\!\frac{M}{B}\cdot y_t$. It follows at comm step $t$, the amount of data each node receives is at most $B\cdot T_B(A_t)\!=\!M\cdot y_t$, so $T_C(A_t)\leq M\cdot\gamma\cdot y_t$. The theorem trivially follows $y\!=\!\sum_t y_t$.
\end{proof}

\thmbwopttocomputeopt*

\thmlinemooreopt*

\begin{proof}
    Suppose $T_L(A_G)\!=\!\alpha k$. Thus, $(G,A_G)$ is Moore optimal if and only if
    \begin{equation}\label{eq:basenodeopt}
        N>M_{d,k-1}=\sum_{i=0}^{k-1}d^i=\frac{d^k}{d-1}-\frac{1}{d-1}.
    \end{equation}
    $(L^n(G),A_{L^n(G)})$ is Moore optimal if and only if
    \begin{equation}\label{eq:resultnodeopt}
        d^nN>M_{d,k+n-1}\ \Leftrightarrow\  N>\frac{d^k}{d-1}-\frac{1}{d^n(d-1)}.
    \end{equation}
    Because $\eqref{eq:resultnodeopt}-\eqref{eq:basenodeopt}<1$ and \eqref{eq:basenodeopt} is an integer, \eqref{eq:basenodeopt} and \eqref{eq:resultnodeopt} are equivalent.
\end{proof}

\thmshortestpathlinegraph*

\begin{proof}
    It is trivial to see that $A_{L(G)}$ is a BFB allgather schedule on $L(G)$. For the sake of contradiction, suppose there exists a BFB schedule $A'_{L(G)}$ that $T_B(A'_{L(G)})<T_B(A_G)+\frac{M}{B}\cdot\frac{1}{N}$. Let $x_{v'\!v,(wu,uu'),t}^*$s be the solution of BFB LP (\ref{lpmodel}) corresponding to $A'_{L(G)}$. We build a schedule $A'_G$ by constructing a solution of (\ref{lpmodel}) such that
    \[
        x_{v,(w,u),t}=\frac{1}{d}\sum_{v'\in N^-(v)}x_{v'\!v,(wu,uu'),t+1}^*,
    \]
    where $u'\!\in\! N^+(u)\setminus\{v\}$ is arbitrary. To verify the construction is a valid solution, given any $u\!\in\! V_G$ and $v\!\in\! N^-_t(w)$, the equality of (\ref{lpmodel}) follows:
    \begin{multline*}
        \sum_{w}x_{v,(w,u),t}=\frac{1}{d}\sum_{v'}\sum_{w}x_{v'\!v,(wu,uu'),t+1}^*\\
        =\frac{1}{d}\sum_{v'}\sum_{wu}x_{v'\!v,(wu,uu'),t+1}^*=\frac{1}{d}\sum_{v'}1
        =\frac{1}{d}\cdot d
        =1.
    \end{multline*}
    The third equality follows the equality constraint in (\ref{lpmodel}). Now, given $(w,u)\!\in\! E_G$, observe that
    \begin{multline*}
        \sum_{v}x_{v,(w,u),t}=\frac{1}{d}\sum_{v}\sum_{v'}x_{v'\!v,(wu,uu'),t+1}^*\\
        =\frac{1}{d}\sum_{v'\!v}x_{v'\!v,(wu,uu'),t+1}^*\leq\frac{1}{d}U_{uu',t+1}.
    \end{multline*}
    Thus, $U_{u,t}\!=\!\max_{w}\sum_{v}x_{v,(w,u),t}\!\leq\!\frac{1}{d}U_{uu',t+1}$ and hence
    \[
        \max_{u\in V_G}U_{u,t}\leq\frac{1}{d}\max_{uu'\in V_{L(G)}}U_{uu',t+1}.
    \]
    By (\ref{eq:lpT_B}), we have
    \begin{multline*}
        T_B(A'_G)\leq T_B(A'_{L(G)})-\frac{M/(dN)}{B/d}\\
        =T_B(A'_{L(G)})-\frac{M}{B}\cdot\frac{1}{N}<T_B(A_G),
    \end{multline*}
    contradicting $A_G$ being the optimal BFB schedule. Thus, combined with inequality (\ref{eq:linebw}), we have proven $A_{L(G)}$ being optimal as well as the equality (\ref{eq:shortestpathT_B}).
\end{proof}

\thmshortestpathlinegraphcoro*

\thmdegexp*

\begin{proof}
    Let $u_i,v_j$ be arbitrary two distinct vertices in $G*n$. Suppose $u\!\neq\! v$ in $G$, then for any $x\!\in\! S$, there exists a sequence in $A_G$:
    \begin{multline*}
        ((v,C_1),(v,w^{(1)}),t_1),((v,C_2),(w^{(1)},w^{(2)}),t_2),\dots\\
        ((v,C_n),(w^{(n-1)},u),t_n),
    \end{multline*}
    where $t_1<t_2<\dots<t_n$ and $x\!\in\! C_1\cap C_2\cap\dots\cap C_n$. By Definition~\ref{def:degexpschedule}, there exists a sequence in $A_{G*n}$:
    \begin{multline*}
        ((v_j,C_1),(v_j,w_j^{(1)}),t_1),((v_j,C_2),(w_j^{(1)},w_j^{(2)}),t_2),\dots\\
        ((v_j,C_n),(w_j^{(n-1)},u_i),t_n),
    \end{multline*}
    as desired. Now, suppose $u\!=\!v$ in $G$. By previous proof, the shard of $v_j$ reaches every in-neighbor $u'_\alpha$ of $u_i$ by the end of comm step $t_{\max}$ since $u'\!\neq\! v$. Then, the last comm step $t_{\max}+1$ added in step 2 of Definition~\ref{def:degexpschedule} delivers the shard to $u_i$ with each edge $(u'_\alpha,u_i)$ delivering $1/nd$ of a shard. Thus, $A_{G*n}$ is a complete allgather.
    
    In step 1 of Definition~\ref{def:degexpschedule}, we have $T_B([A_{G*n}]_t,nM,nB)\!=\!T_B([A_G]_t,M,B)$ and hence $\sum_{t=1}^{t_{\max}}T_B([A_{G*n}]_t,nM,nB)\!=\!T_B(A_G,M,B)$. The $nM$ and $nB$ are due to the fact that both the number of nodes and degree have grown $n$-fold. Thus,
    \begin{align*}
        T_B(A_{G*n},&M,B)=T_B(A_{G*n},nM,nB)\\
        &=T_B(A_G,M,B)+T_B([A_{G*n}]_{t_{\max}+1},nM,nB)\\
        &=T_B(A_G,M,B)+(n-1)\cdot\frac{(nM)/(nN)}{nd}\cdot\frac{1}{nB/(nd)}\\
        &=T_B(A_G,M,B)+\frac{M}{B}\cdot\frac{n-1}{nN}.
    \end{align*}
    The first equality follows the assumption that $T_B(A_G,M,B)\!=\!\tau(M/B)$ for some constant $\tau$.
\end{proof}

\thmdegexpbwopt*

\thmcartesianpowperformance*

\begin{proof}
    We will show that $A^{(1)}$ is a valid allgather schedule. Since $A^{(i)}$s are simply starting at different dimensions, this also shows that $A^{(i)}$s and hence $A_{G^{\square n}}$ are all valid allgather schedules for $G^{\square n}$.

    Let $\mathbf{u}$ be arbitrary vertex in $G^{\square n}$. For any $x\!\in\! S$, we will show that schedule $A^{(1)}$ broadcasts $x$ from $\mathbf{u}$ to all vertices in $G^{\square n}$. At $j\!=\!1$, $A^{(1)}$ performs an allgather over vertices $\{(v_1,\mathbf{u}[2\!:\!n])\ |\ v_1\!\in\! V_G\}$ which induce a subgraph of $G^{\square n}$ isomorphic to $G$. Thus, $x$ has been broadcast to all vertices in $\{(v_1,\mathbf{u}[2\!:\!n])\ |\ v_1\!\in\! V_G\}$. At $j\!=\!2$, $A^{(1)}$ performs an allgather over vertices $\{(v_1,v_2,\mathbf{u}[3\!:\!n])\ |\ v_2\!\in\! V_G\}$ for each $v_1$. By the end of $j\!=\!2$, $x$ has been broadcast to all vertices in $\{(v_1,v_2,\mathbf{u}[3\!:\!n])\ |\ v_1,v_2\!\in\! V_G\}$. By the end of $j\!=\!n$, $x$ has been broadcast to all vertices in $\{\mathbf{v}\ |\ \mathbf{v}\!\in\! V_{G}^n\}\!=\!V_{G^{\square n}}$. Since $\mathbf{u}$ and $x$ are arbitrary, $A^{(1)}$ is a valid allgather schedule for $G^{\square n}$.
    
    As for performance, (\ref{eq:cartesianlatency}) is trivial. To prove (\ref{eq:cartesianbw}), observe that at each $j$ in $A^{(1)}$, allgather $A_G$ is performed with a data size $N^{j-1}M/n$ over the subgraph induced by $\{(\mathbf{y},v,\mathbf{z})\ |\ v\!\in\! V_G\}$ for each $\mathbf{y}\!\in\! V_G^{j-1},\mathbf{z}\!\in\! V_G^{n-j}$. The bandwidth of each node within the subgraph is $1/n$ of that in $G^{\square n}$. It follows that
    \begin{align*}
        T_B(A^{(1)},N^{n-1}M/n,nB)&=\sum_{j=1}^n T_B(A_G,N^{j-1}M/n,B)\\
        &=\sum_{j=1}^n \frac{N^{j-1}}{n}T_B(A_G,M,B)\\
        &=\frac{N^n-1}{n(N-1)}T_B(A_G,M,B).
    \end{align*}
    Therefore,
    \begin{align*}
        T_B(A_{G^{\square n}},M,B)
        &=\frac{n}{N^{n-1}}T_B(A_{G^{\square n}},N^{n-1}M,nB)\\
        &=\frac{n}{N^{n-1}}T_B(A^{(1)},N^{n-1}M/n,nB)\\
        &=\frac{n}{N^{n-1}}\cdot\frac{N^n-1}{n(N-1)}T_B(A_G,M,B)\\
        &=T_B(A_G,M,B)\cdot\frac{N}{N-1}\cdot\frac{N^n-1}{N^n}.
    \end{align*}
\end{proof}

\thmcartesianpowerbwopt*

\thmspathcartesianexpopt*

\begin{proof}
    To prove the theorem, it is sufficient to show that if $G_1$ and $G_2$ have BW-optimal BFB schedules, then $G_1\square G_2$ has a BW-optimal BFB schedule. By Theorem~\ref{thm:shortestlpopt}, let $x^*_{v_1,(w_1,u_1),t_1}$s and $x^*_{v_2,(w_2,u_2),t_2}$s be the solutions of (\ref{lpmodel}) on $G_1$ and $G_2$ respectively. Let $\mathbf u\!=\!(u_1,u_2),\mathbf v\!=\!(v_1,v_2)$. Define $r\!\in\![0,1]$, which we will decide later. We construct a solution of (\ref{lpmodel}) for $G_1\square G_2$ such that:
    \begin{equation}\label{eq:cartesiansolu}
        \begin{aligned}
            x_{\mathbf v,((w_1,u_2),\mathbf u),t_1+t_2}&=\begin{cases}
                r\cdot x^*_{v_1,(w_1,u_1),t_1} & \text{if $u_2\neq v_2$,} \\
                x^*_{v_1,(w_1,u_1),t_1} & \text{if $u_2=v_2$,} 
            \end{cases}\\
            x_{\mathbf v,((u_1,w_2),\mathbf u),t_1+t_2}&=\begin{cases}
                (1-r)\cdot x^*_{v_2,(w_2,u_2),t_2} & \text{if $u_1\neq v_1$,} \\
                x^*_{v_2,(w_2,u_2),t_2} & \text{if $u_1=v_1$.}
            \end{cases}
        \end{aligned}
    \end{equation}
    First of all, because $d_{G_1\square G_2}(\mathbf v,\mathbf u)\!=\!d_{G_1}(v_1,u_1)+d_{G_2}(v_2,u_2)$, it is easy to verify that (\ref{eq:cartesiansolu}) gives a BFB schedule. In addition, for any distinct $\mathbf u,\mathbf v\!\in\! G_1\square G_2$ with $u_1\!\neq\! v_1$ and $u_2\!\neq\! v_2$,
    \begin{align*}
        \sum_{\mathbf w}x_{\mathbf v,(\mathbf w,\mathbf u),t_1+t_2}
        &=r\sum_{w_1}x^*_{v_1,(w_1,u_1),t_1}+(1-r)\sum_{w_2}x^*_{v_2,(w_2,u_2),t_2}\\
        &=r+(1-r)\\
        &=1
    \end{align*}
    satisfying the equality in (\ref{lpmodel}). The $u_1\!=\!v_1$ or $u_2\!=\!v_2$ case is trivial. Because $G_1$ and $G_2$ have BW-optimal BFB schedule, by Theorem~\ref{thm:shortestbwopt}, for any $(w_1,u_1)\!\in\! E_{G_1}$,
    \begin{equation}\label{eq:cartesiang1edge}
        \sum_{v_1\in N^{-G_1}_{t}(u_1)}x^*_{v_1,(w_1,u_1),t}=\frac{N^{-G_1}_{t}}{d_1},
    \end{equation}
    where $N^{-G_1}_{t}(u_1)$ is $N^{-}_{t}(u_1)$ in $G_1$. Define $N^{-G_1\square G_2}_{t_1,t_2}(\mathbf u)\!=\!N_{t_1}^{-G_1}(u_1)\times N_{t_2}^{-G_2}(u_2)$, then it holds that
    \begin{multline*}
        N_t^{-G_1\square G_2}(\mathbf u)=\bigcup_{t_1=0}^tN^{-G_1\square G_2}_{t_1,t-t_1}(\mathbf u)\\
        =N_t^{-G_1}(u_1)\times\{u_2\}\cup\{u_1\}\times N_t^{-G_2}(u_2)\cup\bigcup_{t_1=1}^{t-1}N^{-G_1\square G_2}_{t_1,t-t_1}(\mathbf u).
    \end{multline*}
    Thus, (\ref{eq:cartesiang1edge}) gives
    \begin{equation}\label{eq:cartesianlink1}
        \sum_{\mathbf v\in N_t^{-G_1\square G_2}(\mathbf u)}x_{\mathbf v,((w_1,u_2),\mathbf u),t}=\frac{N^{-G_1}_t}{d_1}+r\sum_{t_1=1}^{t-1}\frac{N^{-G_1}_{t_1}N^{-G_2}_{t-t_1}}{d_1}.
    \end{equation}
    for any $((w_1,u_2),\mathbf u)\!\in\! E_{G_1\square G_2}$. For $G_2$, one can similarly get
    \begin{equation}\label{eq:cartesianlink2}
        \sum_{\mathbf v\in N_t^{-G_1\square G_2}(\mathbf u)}x_{\mathbf v,((u_1,w_2),\mathbf u),t}=\frac{N^{-G_2}_t}{d_2}+(1-r)\sum_{t_2=1}^{t-1}\frac{N^{-G_1}_{t-t_2}N^{-G_2}_{t_2}}{d_2}.
    \end{equation}
    The value of $r$ is the solution to $(\ref{eq:cartesianlink1})\!=\!(\ref{eq:cartesianlink2})$:
    \[
        \frac{N^{-G_1}_t}{d_1}+r\sum_{t_1=1}^{t-1}\frac{N^{-G_1}_{t_1}N^{-G_2}_{t-t_1}}{d_1}=\frac{N^{-G_2}_t}{d_2}+(1-r)\sum_{t_2=1}^{t-1}\frac{N^{-G_1}_{t-t_2}N^{-G_2}_{t_2}}{d_2}.
    \]
    To see there is always a solution $r\!\in\![0,1]$, we have $N^{-G_1}_t\leq d_1\cdot N^{-G_1}_{t-1}$ and $N^{-G_2}_t\leq d_2\cdot N^{-G_2}_{t-1}$, so
    \begin{gather*}
        \frac{N^{-G_1}_t}{d_1}-\frac{N^{-G_2}_t}{d_2}\leq\frac{N^{-G_1}_t}{d_1}\leq N^{-G_1}_{t-1}\leq\sum_{t_2=1}^{t-1}\frac{N^{-G_1}_{t-t_2}N^{-G_2}_{t_2}}{d_2},\\
        \frac{N^{-G_2}_t}{d_2}-\frac{N^{-G_1}_t}{d_1}\leq\frac{N^{-G_2}_t}{d_2}\leq N^{-G_2}_{t-1}\leq\sum_{t_1=1}^{t-1}\frac{N^{-G_1}_{t_1}N^{-G_2}_{t-t_1}}{d_1}.
    \end{gather*}
    The last inequality follows that because $G_2$ is nontrivial simple digraph, $N^{-G_2}_{1}\!=\!d_2$ and hence $N^{-G_1}_{t-1}\!=\!N^{-G_1}_{t-1}N^{-G_2}_{1}/d_2$. Note that $a+rb\!=\!c+(1-r)d$ always has a solution $r\!\in\![0,1]$ if $a-c\leq d$ and $c-a\leq b$. With $(\ref{eq:cartesianlink1})\!=\!(\ref{eq:cartesianlink2})$, by Theorem~\ref{thm:shortestbwopt}, we have constructed a BW-optimal solution of (\ref{lpmodel}) for $G_1\square G_2$. The theorem trivially follows by induction.
\end{proof}

\thmshortestpath*

\begin{proof}
    Let $v_0,v_k$ be arbitrary two distinct vertices in $V_G$ with $d(v_0,v_k)\!=\!k$. For any $x\!\in\! S$, we want to show that there exists a path taking $x$ from $v_0$ to $v_k$. At comm step $k$, conditions \ref{thm:shortestpath1} and \ref{thm:shortestpath2} guarantee that there exists $v_{k-1}\!\in\! N^-(v_k)$ and $((v_0,C_k),(v_{k-1},v_k),k)\!\in\! A$ such that $d(v_0,v_{k-1})\!=\!k-1$ and $x\!\in\! C_k$. At comm step $k-1$, similarly, it is guaranteed that there exists $v_{k-2}\!\in\! N^-(v_{k-1})$ and $((v_0,C_{k-1}),(v_{k-2},v_{k-1}),k-1)\!\in\! A$ such that $d(v_0,v_{k-2})\!=\!k-2$ and $x\!\in\! C_{k-1}$. Thus, we have a sequence of tuples in $A$:
    \begin{multline*}
        ((v_0,C_1),(v_0,v_1),1),((v_0,C_2),(v_1,v_2),2),\dots\\
        ((v_0,C_k),(v_{k-1},v_k),k),
    \end{multline*}
    where $x\!\in\! C_1\cap C_2\cap\dots\cap C_k$ as desired. In the other direction, if condition \ref{thm:shortestpath1} fails, then $A$ is not a BFB schedule; if condition \ref{thm:shortestpath2} fails, then $A$ is not a valid allgather schedule.
\end{proof}

\thmbfslatency*
\begin{proof}
    The proof is trivial.
\end{proof}

\thmshortestlpopt*

\begin{proof}
    The proof is mentioned in text.
\end{proof}

\thmshortestbwopt*

\begin{proof}
    At comm step $t$, each vertex needs to receive shards from vertices in $N^-_t(u)$. By condition \ref{thm:bwopt1} of Theorem~\ref{thm:bwopt}, each in-edge of vertex $u$ receives equal amount of data, so each in-edge receives $\frac{M}{N}|N^-_t(u)|/d$. In addition, condition \ref{thm:bwopt1} of Theorem~\ref{thm:bwopt} also forces every edge in $G$ receiving equal amount of data at any given comm step, so BW optimality is achieved if and only if $\frac{M}{N}|N^-_t(u)|/d\!=\!\frac{M}{N}|N^-_t(v)|/d\!=\!\frac{M}{N}N^-_t/d$ for all $u,v\!\in\! V_G$. Note that condition \ref{thm:bwopt2} of Theorem~\ref{thm:bwopt} is automatically satisfied. Thus, the conditions of Theorem~\ref{thm:shortestbwopt} lead to Theorem~\ref{thm:bwopt}, and vice versa.
\end{proof}

\thmshortestbwoptsuff*

\begin{proof}
    It is easy to see $N_x=da_x/b_x$. Constant $N_x$s satisfy condition \ref{thm:shortestbwopt1} of Theorem~\ref{thm:shortestbwopt}. As for \ref{thm:shortestbwopt2} of Theorem~\ref{thm:shortestbwopt}, at comm step $t$, consider a BFB schedule such that for any $u,v,w\!\in\! V_G$ with $d(v,u)\!=\!d(v,w)+1\!=\!t$, node $w$ sends $1/b_t$ of $v$'s shard to $u$. Thus, $\sum_{((v,C),(w,u))\!\in\! A_t}|C|\!=\!\frac{M}{N}a_t/b_t\!=\!\frac{M}{N}N_t/d$.
\end{proof}

\thmbalancedbfs*
\begin{proof}
    Consider a flow network, where each $j_a$ is connected to each $p_b\!\in\! f(j_a)$ with $\infty$ capacity. Source $s$ is connected to each $j_a$ with capacity 1, and each $p_b$ is connected to sink $t$ with capacity $m/d$. Thus, the workloads can be balanced if and only if the max flow is $m$. Given any subset $J$, consider the $s$-$t$ cut $(A,\bar{A})$ that $A\!=\!s+J+f(J)$. The cut has capacity $m-|J|+\frac{m}{d}|f(J)|$, which is less than $m$ if and only if the inequality is true.
\end{proof}

\thmshortestIP*

\begin{proof}
    For any $(w,u)$ at comm step $t$, since $|N^-_{t-1}(w)|\!\leq\! d^{t-1}$,
    \[
        \sum_{v} y_{v,(w,u),t}<\sum_{v}1+y_{v,(w,u),t}^{\text{LP}}\leq d^{t-1}+\sum_{v}y_{v,(w,u),t}^{\text{LP}}.
    \]
    Thus, we have $W_{u,t}\leq W_{u,t}^{\text{LP}}+d^{t-1}$. By (\ref{eq:lpT_B}),
    \begin{multline*}
        T_B-T_B^{\text{OPT}}\leq T_B-T_B^{\text{LP}}\\
        \leq\frac{M/N}{B/d}\cdot\frac{1}{P}\sum_{t=1}^{D(G)}d^{t-1}=\frac{M}{B}\cdot\frac{d(d^{D(G)}-1)}{(d-1)PN}.
    \end{multline*}
    Note that we need to divide (\ref{eq:lpT_B}) by $P$, because $y_{v,(w,u),t}\!\in\![0,P]$ in (\ref{integerprog}) while $x_{v,(w,u),t}\!\in\![0,1]$ in (\ref{lpmodel}). If $G$ is Moore optimal (i.e., $N>M_{d,D(G)-1}$), it follows that
    \[
        T_B-T_B^{\text{OPT}}<\frac{M}{B}\cdot\frac{d(d^{D(G)}-1)}{(d-1)PM_{d,D(G)-1}}=\frac{M}{B}\cdot\frac{d}{P}.
    \]
\end{proof}

\thmgenkautzdiameter*

\begin{proof}
    From \cite{genkautz}, we know that $k\leq\lceil\log_d m\rceil$. Then,
    \[
        m\geq d^{k-1}>\frac{d^{k-1}-1}{d-1}=M_{d,k-2}.
    \]
\end{proof}

\thmdegfourconstruct*
\begin{proof}
    See~\cite{circulantk2}.
\end{proof}

%% file: figures/diamond.tex
\begin{tikzpicture}[roundnode/.style={circle,draw=black,minimum size=7mm}]
	\node[roundnode] (2) at (0, 3*4/5) {$2$};
	\node[roundnode] (1) at (-2.12*4/5,2.12*4/5) {$1$};
	\node[roundnode] (6) at (2.12*4/5,2.12*4/5) {$6$};
	\node[roundnode] (5) at (-1*4/5,0) {$5$};
	\node[roundnode] (7) at (1*4/5,0)  {$7$};
	\node[roundnode] (4) at (-3*4/5, 0) {$4$};
	\node[roundnode] (3) at (3*4/5, 0) {$3$};
	\node[roundnode] (0) at (0, -1.89*4/5) {$0$};

	\path[-latex] (0) edge (3);
	\path[-latex] (0) edge (4);
	\path[-latex] (1) edge (2);
	\path[-latex] (1) edge (7);
	\path[-latex] (2) edge (3);
	\path[-latex] (2) edge (4);
	\path[-latex] (3) edge (6);
	\path[-latex] (3) edge (7);
	\path[-latex] (4) edge (1);
	\path[-latex] (4) edge (5);
	\path[-latex] (5) edge (0);
	\path[-latex] (5) edge (1);
	\path[-latex] (6) edge (2);
	\path[-latex] (6) edge (5);
	\path[-latex] (7) edge (0);
	\path[-latex] (7) edge (6);
\end{tikzpicture}

%% file: basegraphtable.tex
\setlength{\rotFPtop}{0pt plus 1fil}

\begin{sidewaystable*}[p]
    \centering
    \resizebox{\textwidth}{!}{
    \begin{tabular}{|c|c|c|c|c|c|c|c|c|c|}
        \hline
        \multirow{2}{*}{Topology} & \multirow{2}{*}{Notation} & \multirow{2}{*}{Degree} & \multirow{2}{*}{Size} & \multirow{2}{*}{\shortstack[c]{Reverse- \\ Symmetric}} & \multirow{2}{*}{\shortstack[c]{Bandwidth \\ Optimal}} & \multirow{2}{*}{\shortstack[c]{Moore \\ Optimal}} & \multirow{2}{*}{\shortstack[c]{BFB \\ Schedule}} & \multirow{2}{*}{Self-Loop} & \multirow{2}{*}{MultiEdge} \\
        & & & & & & & & & \\
        \hline
        Complete & $K_m$ & $m-1$ & $m$ & \checkmark & \checkmark & \checkmark & \checkmark & $\times$ & $\times$ \\
        \hline
        \multirow{2}{*}{\shortstack[c]{Complete \\ Bipartite (Fig~\ref{fig:notationexample})}} & \multirow{2}{*}{$K_{d,d}$} & \multirow{2}{*}{$d$} & \multirow{2}{*}{$2d$} & \multirow{2}{*}{\checkmark} & \multirow{2}{*}{\checkmark} & \multirow{2}{*}{\checkmark} & \multirow{2}{*}{\checkmark} & \multirow{2}{*}{$\times$} & \multirow{2}{*}{$\times$} \\ 
        & & & & & & & & & \\
        \hline
        Hamming & $H(n,q)=K_q^{\square n}$ & $n(q-1)$ & $q^n$ & \checkmark & \checkmark & $T_L=n$ & \checkmark & $\times$ & $\times$ \\
        \hline
        \multirow{2}{*}{\shortstack[c]{Kautz}} & \multirow{2}{*}{\shortstack[c]{$K(d,n)=$ \\ $L^n(K_{d+1})$}} & \multirow{2}{*}{$d$} & \multirow{2}{*}{$d^n(d+1)$} & \multirow{2}{*}{\checkmark} & \multirow{2}{*}{\shortstack[c]{when \\ $n=0$}} & \multirow{2}{*}{\checkmark} & \multirow{2}{*}{\checkmark} & \multirow{2}{*}{\shortstack[c]{$\times$}} & \multirow{2}{*}{$\times$} \\ 
        & & & & & & & & & \\
        \hline
        \multirow{2}{*}{\shortstack[c]{Generalized \\ Kautz (\S\ref{sec:genkautz})}} & \multirow{2}{*}{$\Pi_{d,m}$} & \multirow{2}{*}{$d$} & \multirow{2}{*}{$m\geq d+1$} & \multirow{2}{*}{$\times$} & \multirow{2}{*}{\shortstack[c]{when \\ $m=d+1$}} & \multirow{2}{*}{$T_L\leq T_L^*+1$} & \multirow{2}{*}{\checkmark} & \multirow{2}{*}{\shortstack[c]{when $m\bmod$ \\ $(d+1)\neq 0$}} & \multirow{2}{*}{$\times$} \\ 
        & & & & & & & & & \\
        \hline
        Circulant (\S\ref{sec:circulant}) & $C(m,\{a_1,\dots,a_d\})$ & $d$ & $m$ & \checkmark & \checkmark & $\times$ & \checkmark & $\times$ & $\times$ \\
        \hline
        Directed Circulant & & $d$ & $d+2$ & \checkmark & \checkmark & \checkmark & \checkmark & $\times$ & $\times$ \\
        \hline
        \multirow{2}{*}{\shortstack[c]{Bidirectional \\ Ring}} & \multirow{2}{*}{BiRing$(d,m)$} & \multirow{2}{*}{even $d$} & \multirow{2}{*}{$m\geq 3$} & \multirow{2}{*}{\checkmark} & \multirow{2}{*}{\checkmark} & \multirow{2}{*}{$T_L=\lfloor m/2\rfloor$} & \multirow{2}{*}{\checkmark} & \multirow{2}{*}{$\times$} & \multirow{2}{*}{\shortstack[c]{when \\ $d>2$}} \\ 
        & & & & & & & & & \\
        \hline
        \multirow{2}{*}{\shortstack[c]{Unidirectional \\ Ring}} & \multirow{2}{*}{UniRing$(d,m)$} & \multirow{2}{*}{$d$} & \multirow{2}{*}{$m$} & \multirow{2}{*}{\checkmark} & \multirow{2}{*}{\checkmark} & \multirow{2}{*}{$T_L=m-1$} & \multirow{2}{*}{\checkmark} & \multirow{2}{*}{$\times$} & \multirow{2}{*}{\shortstack[c]{when \\ $d>1$}} \\ 
        & & & & & & & & & \\
        \hline
        Diamond (Fig~\ref{fig:diamondtopo}) & & $2$ & $8$ & $\times$ & \checkmark & \checkmark & $\times$ & $\times$ & $\times$ \\
        \hline
        \multirow{2}{*}{de Bruijn} & \multirow{2}{*}{$\text{DBJ}(d,n)$} & \multirow{2}{*}{$d$} & \multirow{2}{*}{$d^n$} & \multirow{2}{*}{\checkmark} & \multirow{2}{*}{\shortstack[c]{when \\ $n\leq 1$}} & \multirow{2}{*}{\checkmark} & \multirow{2}{*}{\checkmark} & \multirow{2}{*}{\checkmark} & \multirow{2}{*}{$\times$} \\ 
        & & & & & & & & & \\
        \hline
        \multirow{4}{*}{\shortstack[c]{Modified \\ de Bruijn \\ (Fig~\ref{fig:modifieddbjg})}} & DBJMod$(2,3)$ & $2$ & $8$ & \checkmark & \checkmark & $T_L=4$ & $\times$ & $\times$ & $\times$ \\
        \cline{2-10}
        & DBJMod$(2,4)$ & $2$ & $16$ & $\times$ & \checkmark & $T_L=5$ & $\times$ & $\times$ & $\times$ \\
        \cline{2-10}
        & DBJMod$(3,2)$ & $3$ & $9$ & $\times$ & \checkmark & $T_L=3$ & $\times$ & $\times$ & $\times$ \\
        \cline{2-10}
        & DBJMod$(4,2)$ & $4$ & $16$ & $\times$ & \checkmark & $T_L=3$ & $\times$ & $\times$ & $\times$ \\
        \hline
        \multirow{2}{*}{\shortstack[c]{Distance-Regular \\ Graphs (\S\ref{sec:dist-reg})}} & \multirow{2}{*}{$\text{DistReg}(d,m)$} & \multirow{2}{*}{$d$} & \multirow{2}{*}{$m$} & \multirow{2}{*}{\checkmark} & \multirow{2}{*}{\checkmark} & & \multirow{2}{*}{\checkmark} & \multirow{2}{*}{$\times$} & \multirow{2}{*}{$\times$} \\
        & & & & & & & & & \\
        \hline
    \end{tabular}
    }
    \vspace{0.5\baselineskip}
    \caption{Summary of Important Topologies.}
    \label{table:toposummary}
\end{sidewaystable*}